\documentclass[12pt]{article}
\usepackage[utf8]{inputenc}
\usepackage[english]{babel}
\usepackage{graphicx,amsfonts,amsmath,amssymb,amsthm,mathrsfs,mathabx}
\usepackage[usenames,dvipsnames]{xcolor}
\usepackage[toc,page]{appendix}
\usepackage{algorithm}
\usepackage{authblk}
\usepackage[bookmarksnumbered=true]{hyperref}

\usepackage[final]{showkeys}

\usepackage[capitalise]{cleveref}
\crefname{equation}{}{}

\numberwithin{equation}{section}

\usepackage[disable]{todonotes}

\usepackage[style=alphabetic,
  sorting = nyt,
  sortcites=true,
  giveninits=true,
  date=year,
  isbn=false,
  maxbibnames=99,
  maxalphanames=6,
  backend=biber]{biblatex}
\DeclareFieldFormat*{titlecase}{\MakeSentenceCase{#1}}

\DeclareSourcemap{
  \maps[datatype=bibtex]{
    \map[overwrite]{
      \step[fieldsource=doi, final]
      \step[fieldset=url, null]
      \step[fieldset=eprint, null]
    }  
  }
}

\DeclareSourcemap{
  \maps[datatype=bibtex, overwrite]{
    \map{
      \step[fieldset=language, null]
      \step[fieldset=month, null]
      \step[fieldset=urldate, null]
    }
  }
}

\DeclareSourcemap{
  \maps[datatype=bibtex, overwrite]{
    \map{
      \pertype{article}
      \step[fieldset=note, null]
    }
  }
}

\AtEveryBibitem{%
  \clearlist{language}%
  \clearlist{month}%
  \clearlist{urldate}%
  \clearfield{urldate}%
}

\renewbibmacro*{in:}{%
  \setunit{\addcomma\space}%
  \ifentrytype{article}
    {}
    {\printtext{%
       \bibstring{in}\intitlepunct}}}

\renewcommand*{\intitlepunct}{\addspace}

\DeclareFieldFormat[misc]{title}{\mkbibquote{#1}}
\DeclareFieldFormat[article,inproceedings]{volume}{\mkbibbold{#1}}
\DeclareFieldFormat[inproceedings]{series}{\mkbibitalic{#1}}

\title{Upper Bound on the Correlation Energy of a Dense Fermi Gas in the Thermodynamic Limit}

\author[1]{Niels Benedikter}
\author[2]{Sascha Lill}
\author[3]{Marcello Porta}
\author[4]{Benjamin Schlein}
\affil[1]{
Universit\`a degli Studi di Milano, Dipartimento di Matematica, Via Cesare Saldini 50, 20133 Milano, Italy

\smallskip

External Scientific Member of Basque Center for Applied Mathematics, Alameda de Mazarredo 14, 48009 Bilbao, Bizkaia, Spain

\smallskip

ORCID: \href{https://orcid.org/0000-0002-1071-6091}{0000-0002-1071-6091}, e--mail: \href{mailto:niels.benedikter@unimi.it}{niels.benedikter@unimi.it}
\bigskip
}

\affil[2]{
University of Copenhagen, \textsc{QMath}, Lyngbyvej 2, 2100 Copenhagen, Denmark

\smallskip

ORCID: \href{https://orcid.org/0000-0002-9474-9914}{0000-0002-9474-9914}, e--mail: \href{mailto:sali@math.ku.dk}{sali@math.ku.dk}
\bigskip
}

\affil[3]{
Scuola Internazionale Superiore di Studi Avanzati, Mathematics Area, Via Bonomea 265, 34136 Trieste, Italy

\smallskip

ORCID: \href{https://orcid.org/0000-0002-7725-8398}{0000-0002-7725-8398}, e--mail: \href{mailto:marcello.porta@sissa.it}{marcello.porta@sissa.it}
\bigskip
}

\affil[4]{
Universität Zürich, Institut für Mathematik, Winterthurerstrasse 190, 8057 Zürich, Switzerland

\smallskip

ORCID: \href{https://orcid.org/0000-0003-0910-9136}{0000-0003-0910-9136}, e--mail: \href{mailto:benjamin.schlein@math.uzh.ch}{benjamin.schlein@math.uzh.ch}
}

\usepackage[utf8]{inputenc}
\graphicspath{ {figures/}}
\usepackage{tikz}
\usetikzlibrary{decorations.pathreplacing, patterns}

\usepackage[section]{placeins}
\usepackage{capt-of, caption}
\usepackage{paralist}
\usepackage{booktabs}
\usepackage{hhline}
\usepackage{color}		

\newcommand{\cE}{\mathcal{E}}
\newcommand{\cF}{\mathcal{F}}

\newcommand{\cH}{\mathcal{H}}
\newcommand{\cI}{\mathcal{I}}

\newcommand{\cM}{\mathcal{M}}
\newcommand{\cN}{\mathcal{N}}
\newcommand{\cO}{\mathcal{O}}
\newcommand{\cP}{\mathcal{P}}

\newcommand{\cR}{\mathcal{R}}

\newcommand{\cT}{\mathcal{T}}

\newcommand{\cW}{\mathcal{W}}
\newcommand{\cX}{\mathcal{X}}

\newcommand{\fh}{\mathfrak{h}}
\newcommand{\fl}{\mathfrak{l}}

\newcommand{\fD}{\mathfrak{D}}

\newcommand{\fR}{\mathfrak{R}}
\newcommand{\um}{\underline{m}}

\newcommand{\ulambda}{\underline{\lambda}}
\newcommand{\CCC}{\mathbb{C}}
\newcommand{\DDD}{\mathbb{D}}

\newcommand{\HHH}{\mathbb{H}}
\newcommand{\NNN}{\mathbb{N}}
\newcommand{\QQQ}{\mathbb{Q}}
\newcommand{\RRR}{\mathbb{R}}

\newcommand{\TTT}{\mathbb{T}}
\newcommand{\XXX}{\mathbb{X}}
\newcommand{\ZZZ}{\mathbb{Z}}

\newcommand{\B}{{\rm B}}

\renewcommand{\d}{\textnormal{d}}
\newcommand{\di}{\textnormal{d}}

\newcommand{\diam}{{\rm diam}}
\newcommand{\dist}{{\rm dist}}

\newcommand{\err}{{\rm err}}

\newcommand{\F}{{\rm F}}
\newcommand{\FB}{{\rm FB}}
\newcommand{\GS}{{\rm GS}}
\newcommand{\HF}{\mathrm{HF}}
\newcommand{\HS}{\mathrm{HS}}
\newcommand{\I}{\mathrm{I}}
\newcommand{\II}{\mathrm{II}}
\newcommand{\III}{\mathrm{III}}
\newcommand{\IV}{\mathrm{IV}}

\newcommand{\kin}{\mathrm{kin}}

\newcommand{\op}{{\rm op}}

\newcommand{\RPA}{\mathrm{RPA}}

\newcommand{\supp}{\mathrm{supp}}
\newcommand{\tr}{\mathrm{tr}}
\newcommand{\norm}[1]{\left\Vert #1 \right\Vert}

\newcommand{\BF}{B_\F}

\newcommand{\edge}{\mathrm{edge}}

\newcommand{\Ncal}{\mathcal{N}}

\newtheorem{theorem}{Theorem}[section]
\newtheorem{corollary}[theorem]{Corollary}
\newtheorem{proposition}[theorem]{Proposition}
\newtheorem{lemma}[theorem]{Lemma}
\theoremstyle{definition}
\newtheorem{remark}[theorem]{Remark}
\theoremstyle{definition}
\newtheorem{remarks}[theorem]{Remarks}
\theoremstyle{definition}

\theoremstyle{definition}
\newtheorem{assumption}[theorem]{Assumption}

\newcounter{remarks}
\begin{document}
\maketitle
\begin{abstract}
We prove that the random phase approximation provides an upper bound on the correlation energy of a three-dimensional, weakly interacting, Fermi gas in the thermodynamic limit at high density. More precisely, we consider a mean-field scaling regime where the kinetic and interaction Hamiltonian are scaled by $\rho^{-2/3}$ and $\rho^{-1}$, respectively. We use a correlated trial state constructed through a quasi-bosonic Bogoliubov transformation of collective particle–hole excitations associated to patches on the Fermi surface. Key of the proof is to control the position-space localization of the collective excitations. Combined with an iterated Duhamel expansion of the expectation values of (localized) particle number operators we obtain extensive estimates, which justify the reduction to an effective quadratic Hamiltonian and its approximate diagonalization.
\end{abstract}

\clearpage

\tableofcontents

\section{Introduction and Main Result}
\label{sec:intro}

We consider a gas of $ N $ fermions on a three-dimensional torus $ \TTT_L := \big[ -\frac L2, \frac L2 \big]^3 $ with side length $ L > 0 $. The density of the gas is $ \rho := \frac{N}{L^3} $. The one-particle Hilbert space is $ \fh := L^2(\TTT_L) $ with inner product $ \langle ., . \rangle $, and the state of the system is described by a vector in the antisymmetric Hilbert space $ L^2(\TTT_L)^{\otimes_{\mathrm{a}} N} $. The system is modeled by the Hamiltonian
\begin{equation} \label{eq:HN}
	H_N := \sum_{j = 1}^N \hbar^2 (- \Delta_{x_j}) + \rho^{-1} \sum_{i < j}^N V(x_i - x_j) \;,
\end{equation}
where $ \hbar > 0 $ is the semiclassical parameter, and $ V: \TTT_L \to \RRR $ is the periodization of a potential $ V_\infty: \RRR^3 \to \RRR $. Introducing the Fourier transform
\begin{equation} \label{eq:Fourier_1}
	\hat{V}(k) := \int_{\RRR^3} V_\infty(x) e^{-i k \cdot x} \d x \;, \qquad k \in \RRR^3 \;,
\end{equation}
we have
\begin{equation} \label{eq:Fourier_convention}
	\hat{V}(k) = \int_{\TTT_L} V(x) e^{-i k \cdot x} \d x \;, \qquad
	V(x) = \frac{1}{L^3} \sum_{k \in \Lambda_L} \hat{V}(k) e^{i k \cdot x} \;.
\end{equation}
\begin{assumption} \label{as:V}
We assume that $ \hat{V} \in C_0^\infty(\RRR^3;\RRR) $ and $ \hat{V} \ge 0 $.
\end{assumption}
In particular, there exists $ R > 0 $, independent of $ L$ and $\rho $, such that $ \hat{V} = 0 $ if $ |k| > R $.
The Hamiltonian is self-adjoint on a dense domain in $ L^2(\TTT_L)^{\otimes_{\mathrm{a}} N} $.
We consider a high-density scaling limit, where
\begin{equation}
\hbar := \rho^{- \frac 13} \;,
\end{equation}
so the relative coupling is $ \rho^{-1} / \hbar^2 = \rho^{-\frac 13} $.
The remaining parameters are the density $\rho$ and the side length $L$ of the torus. We are interested in the ground state energy
\begin{equation} \label{eq:EGS}
	E_{\GS} := \inf \sigma(H_N)
\end{equation}
in the thermodynamic limit at high density. That means, we first take the limit $ L \to  \infty $ at fixed, large $ \rho $ and then study the asymptotics as $ \rho \to \infty $.

To obtain a first approximation to $ E_{\GS} $, we may consider the non-interacting case $ V = 0 $. The eigenvectors of $ - \Delta $ are plane waves
\begin{equation} \label{eq:Lambda_L}
	f_k(x) := L^{-\frac 32} e^{i k \cdot x} \;, \qquad
	k \in \Lambda_L := \frac{2\pi}{L} \ZZZ^3 \;,
\end{equation}
and each antisymmetrized tensor product of $ N $ plane waves with pairwise disjoint wave vectors $ k $ is an eigenstate of $ H_N $. In particular, a ground state is obtained by choosing the $ k $ to fill up the Fermi ball
\begin{equation} \label{eq:BF}
	\BF := \{ k \in \Lambda_L : |k| \le k_{\F} \} \;,
\end{equation}
where the Fermi momentum $ k_{\F} > 0 $ is chosen such that $ |\BF| = N $. (We  restrict to such $ (\rho, L) $ that the constraint $ |\BF| = N $ is fulfilled.) We then have
\begin{equation}
	k_\F = (6 \pi^2 \rho)^{\frac 13} + o_L(1) \quad \Leftrightarrow \quad
	\rho = \frac{k_\F^3}{6 \pi^2} + o_L(1)  \;,
\end{equation}
where $ o_L(1) $ denotes an error term that vanishes as $ L \to \infty $.
The corresponding state
\begin{equation} \label{eq:psiFB}
	\psi_{\FB}
	:= \bigwedge_{k \in \BF} f_k
\end{equation}
is a ground state for $ H_N $ for $ V = 0 $. For general $ V \neq 0 $, $ \psi_{\FB} $ will no longer be a ground state, but one may expect that
\begin{equation} \label{eq:EFB}
\begin{aligned}
	E_{\FB}
	& := \langle \psi_{\FB}, H_N \psi_{\FB} \rangle \\
	& = \hbar^2 \sum_{p \in \BF} |p|^2
		- \frac{L^3}{2} V(0)
		+ \frac{N}{2} \hat{V}(0)
		+ \frac{1}{2N} \sum_{k \in \Lambda_L \setminus \{0\}} \hat{V}(k)  |(\BF + k) \setminus \BF| \\
	&= \hbar^2 \sum_{p \in \BF} |p|^2
		+ \frac{N-1}{2} \hat{V}(0) 
		+ \frac{1}{2N} \sum_{k \in \Lambda_L \setminus \{0\}} \hat{V}(k) \big( |(\BF + k) \setminus \BF| - N \big)\;,
\end{aligned}
\end{equation}
is still a good approximation for $ E_{\GS} $. By the variational principle, we have $ E_{\GS} \le E_{\FB} $, and one easily verifies that $ E_{\FB} \sim N = \rho L^3 $. (It has been shown by \cite{GL18,GHL19} that the Hartree--Fock ground state energy $ E_{\HF} $, which is the minimum over the energy over all Slater determinants, is smaller than \eqref{eq:EFB} only by a term of order $e^{-C \rho^{1/3}}$, or $e^{-C \rho^{1/6}}$ if the interaction is given by the Coulomb potential.)

The difference between $ E_{\GS} $ and the Hartree--Fock energy is called the correlation energy.
The dominant contribution to the correlation energy was first computed by \cite{Mac50} by a partial resummation of perturbation theory, by \cite{SBFB57,Saw57} by a bosonization approximation, and by \cite{BP53,Pin53} through identification of collective modes. Gell-Mann and Brueckner~\cite{GB57} obtained a further order by  a refined partial resummation of perturbation theory. Rigorous derivations of the correlation energy were recently obtained in a finite--volume mean--field scaling limit by~\cite{BNPSS20,BNPSS21,BPSS22} and \cite{CHN21,CHN23,CHN24}, based on bounds from~\cite{HPR20}. See also~\cite{FRS25} for results in a related, strongly interacting scaling regime. In the thermodynamic limit, the dominant contribution to the correlation energy density is
\begin{equation} \label{eq:Ecorr}
\begin{aligned}
	e_{\RPA}
	&:= \hbar^2 k_{\F} \frac{1}{(2 \pi)^3} \int_{\RRR^3} \frac{|k|}{\pi} \int_0^\infty F \big( Q^{(0)}_k(\mu) \big) \d \mu \; \d k\;,  \quad \text{where}\\
	F(x)
	& := \log(1+x) - x \;, \quad Q^{(0)}_k(\mu)
	 := \frac{\hat{V}(k) k_{\F}}{(2 \pi)^2 \rho^{\frac 13}} \big( 1 - \mu \arctan(\mu^{-1}) \big) \;.
\end{aligned}
\end{equation}
Typically $ e_{\RPA} \sim \rho^{-\frac 13} \sim \hbar $.
Our main result establishes $e_{\RPA}$ for an upper bound on the ground state energy in the thermodynamic limit.
\begin{theorem}[Main result] \label{thm:main}
Let $ V $ satisfy Assumption~\ref{as:V} and consider the energy densities
\begin{equation} \label{eq:energydensities}
	e_{\GS} := \lim_{L \to \infty} \frac{E_{\GS}}{L^3} \;, \qquad
	e_{\FB} := \lim_{L \to \infty} \frac{E_{\FB}}{L^3} \;.
\end{equation}
Then
\begin{equation} \label{eq:main}
	e_{\GS} \le e_{\FB} + e_{\RPA} + \mathfrak{E} \;, \qquad \text{with }
	\lim_{\rho \to \infty} \frac{\mathfrak{E}}{\hbar} = 0  \;.
\end{equation}
\end{theorem}

The error $ \mathfrak{E} $ is subleading compared to $ e_{\RPA} \sim \hbar $ (and $ e_{\FB} \sim \rho $).

\medskip

Bounds on $ E_{\GS} $ in the dense thermodynamic limit were previously obtained in~\cite{GS94} for stronger coupling ($ \rho^{-\frac 23} $ instead of $ \rho^{-1} $ in our notation) but at much lower precision ($ \rho^{\frac 23} $ instead of $ \rho^{-\frac 13} = \hbar $ in our notation). Recently, also bounds on $ E_{\GS} $ for a Fermi gas in the dilute thermodynamic limit have been established in~\cite{FGHP21,Gia23,Gia24,GHNS26,GHNS25,HSS26,LS24a,LS24b,Lau25}. Similar results for Fermi gases address the thermodynamic pressure~\cite{LS24c,LS24d} as well as the momentum distribution~\cite{BL23,BLN25,BGLL26}. Our proof further uses ideas from~\cite{FPS23,FPS25} to control the large-volume limit.

\medskip

This article is structured as follows: In Section~\ref{sec:setting} we construct the trial state. In Section~\ref{sec:positionspace} we compile useful formulas for expressing operators in position space. Section~\ref{sec:gronwall_bounds} is dedicated to the proof of a preliminary bound on the number of excitations in the trial state. We improve this bound in Section~\ref{sec:improved_bounds} using a Duhamel expansion. With the improved excitation bound at hand, we  reduce $ H_N $ to an effective Hamiltonian in Section~\ref{sec:reduction_Heff} and extract the RPA energy from it in Section~\ref{sec:proofmain}, completing the proof of Theorem~\ref{thm:main}.

\smallskip

\paragraph{Notation.}
Integrals are over $ \TTT_L $, unless otherwise stated, e.\,g., $ \int \d x \d y = \int_{\TTT_L \times \TTT_L} \d x \d y $. We use the O-notation with subscript $ L $ to indicate the first (thermodynamic) limit $ L \to \infty $:
\[
\begin{aligned}
	A &= B + o_L(f(L)) \qquad &&:\Leftrightarrow \qquad
	\lim_{L \to \infty} \frac{|A-B|}{f(L)} = 0 \;, \\
	A &= B + \cO_L(f(L)) \qquad &&:\Leftrightarrow \qquad
	\exists C > 0 : \quad |A - B| \le C f(L) \quad \forall L > 0 \;,
\end{aligned}
\]
The analogues without a subscript ($ o(f(\rho)) $ and $ \cO(f(\rho)) $) indicate the second (high-density) limit $ \rho \to \infty $. In general, $ C > 0 $ denotes a sufficiently large constant independent of $ (\rho, L) $, and $ \varepsilon > 0 $ an arbitrarily small constant independent of $ (\rho, L) $. We also write $ C_a > 0 $ for a sufficiently large constant depending on some parameter $ a $, but not on $ (\rho,L) $. All three constants may change from line to line. Also, all bounds in this article are understood to be valid for large enough $ L $ and $ \rho $, meaning that we first choose a large enough $ \rho $ and then a large enough $ L $ depending on $ \rho $.

\section{Trial State Construction}
\label{sec:setting}

\subsection{Second Quantization}
\label{subsec:secondquant}

We introduce the fermionic Fock space
\begin{equation}
	\cF := \bigoplus_{N = 0}^\infty L^2(\TTT_L)^{\otimes_{\mathrm{a}} N} \;.
\end{equation}
On $\cF$, let $ a^*_x$ and $a_x $ be the standard fermionic creation and annihilation operator valued distributions, so for $ f \in L^2(\TTT_L) $,
\begin{equation}
	a^*(f) := \int \d x \; f(x) a^*_x  \;, \qquad
	a(f) := \int \d x \; \overline{f(x)} a_x \;,
\end{equation}
are the creation and annihilation operators satisfying the canonical anticommutation relations (CAR)
\begin{equation} \label{eq:CAR}
	\{ a(f), a^*(g) \} = \langle f, g \rangle\;, \quad
	\{ a(f), a(g) \} = \{ a^*(f), a^*(g) \} = 0 \quad
	\forall f, g \in L^2(\TTT_L) \;.
\end{equation}
The vacuum vector $ \Omega = (1, 0, 0, \ldots) \in \cF $ is characterized by $ a(f) \Omega = 0 $ for all $f \in L^2(\TTT_L)$. We lift $ H_N $ to the operator $ \cH_N $ on Fock space as
\begin{equation}
\begin{aligned}
	\cH_N &= \cH_{\kin} + \cH_{\mathrm{int}} \;, \\
	\cH_{\kin} &:= \int \d x \; \hbar^2  a^*_x (- \Delta_x a_x) \;,
	\qquad
	\cH_{\mathrm{int}} := \frac{1}{2 \rho} \int \d x \d y \; V(x - y) a_x^* a_y^* a_y a_x \;.
\end{aligned}
\end{equation}
It is straightforward to check that the restriction of $ \cH_N $ to $ L^2(\TTT_L)^{\otimes_{\mathrm{a}} N} $ is just $ H_N $. It is convenient to write $ \cH_N $ in momentum space, i.\,e., using the orthonormal basis of plane waves $ (f_k)_{k \in \Lambda_L} $ defined in \eqref{eq:Lambda_L}, and introducing the momentum mode creation and annihilation operators
\begin{equation}
	\hat{a}_k^* := a^*(f_k) = \frac{1}{L^{3/2}} \int \d x \; e^{ik \cdot x} a_x^* \;, \qquad
	\hat{a}_k := a(f_k) = \frac{1}{L^{3/2}} \int \d x \; e^{-ik \cdot x} a_x \;.
\end{equation}
In particular, $ \{ \hat{a}_k, \hat{a}^*_{k'} \} = \delta_{k, k'} $, and
\begin{equation}
	a_x^* = \sum_{k \in \Lambda_L}  \frac{e^{-i k \cdot x}}{L^{3/2}} \hat{a}_k^* \;, \qquad 
	a_x = \sum_{k \in \Lambda_L}  \frac{e^{i k \cdot x}}{L^{3/2}} \hat{a}_k\;.
\end{equation}
Furthermore, for an operator $ A $ on $ L^2(\TTT_L) $ with integral kernel $ A(x,y) $, we denote its second quantization by
\begin{equation} \label{eq:dGamma}
	\d \Gamma(A)
	:= \int \d x \d y \; A(x,y) a_x^* a_y \;.
\end{equation}
In particular, we have the second quantization of the identity,
\begin{equation}
	\cN := \d \Gamma(1) \;,
\end{equation}
which is the number operator.
The kinetic part of $ \cH_N $ can be written as
\begin{equation}
	\cH_{\kin}
	= \hbar^2 \d \Gamma(-\Delta)
	= \sum_{p \in \Lambda_L} \hbar^2 |p|^2 \hat{a}_p^* \hat{a}_p \;.
\end{equation}
Using~\eqref{eq:Fourier_convention} and the CAR, the interaction Hamiltonian in momentum space becomes
\begin{equation}
	\cH_{\mathrm{int}}
	= \frac{1}{2 \rho L^3} \sum_{p,q,k \in \Lambda_L} \hat{V}(k) \hat{a}_{p+k}^* \hat{a}_p \hat{a}^*_{q-k} \hat{a}_q
		- \frac{1}{2 \rho} \cN V(0) \;.
\end{equation} 
We conclude
\begin{equation}
\label{eq:cHNmomentum}
	\cH_N = \sum_{p \in \Lambda_L} \hbar^2 |p|^2 \hat{a}_p^* \hat{a}_p
		+ \frac{1}{2 \rho L^3} \sum_{p,q,k \in \Lambda_L} \hat{V}(k) \hat{a}_{p+k}^* \hat{a}_p \hat{a}^*_{q-k} \hat{a}_q
		- \frac{1}{2 \rho} \cN V(0) \;.
\end{equation}

\subsection{Particle--Hole Transformation}
\label{subsec:fermiball}

Recall the Fermi ball $ \BF \subset \Lambda_L $ as defined in \eqref{eq:BF} and denote the complement of the Fermi ball as
\begin{equation}
	\BF^c := \Lambda_L \setminus \BF \;.
\end{equation}
We introduce the particle--hole transformation $ R : \cF \to \cF $ with $ R^* = R = R^{-1} $ by
\begin{equation} \label{eq:R}
	R^* \hat{a}_k R := \begin{cases}
		\hat{a}_k		&\quad \text{if } k \in \BF^c\\
		\hat{a}_{-k}^*	&\quad \text{if } k \in \BF\\
	\end{cases} \;, \qquad
	R \Omega := \psi_{\FB} = \bigwedge_{k \in \BF} f_k \;.
\end{equation}
With respect to~\cite{BNPSS20,BNPSS21}, we introduced a momentum reflection inside $ \BF $ to facilitate later computations.

Conjugating $ \cH_N $ with $ R $, we extract $ E_{\FB} $ as in~\cite[Sect.~1]{BNPSS21}: Introducing the pair creation operator
\begin{equation} \label{eq:bstar_full}
	B^*(k) := \sum_{p \in \BF^c \cap (\BF + k)} \hat{a}_p^* \hat{a}_{-p+k}^* \;,
\end{equation}
and the shift operator
\begin{equation} \label{eq:fDk}
	\fD(k)^* := \sum_{p \in \BF^c \cap (\BF^c + k)} \hat{a}_p^* \hat{a}_{p-k}
		- \sum_{h \in \BF \cap (\BF - k)} \hat{a}_{-h}^* \hat{a}_{-h-k} \;,
\end{equation}
we may write
\begin{equation}
	R^* \left( \sum_{p \in \Lambda_L} \hat{a}_{p+k}^* \hat{a}_p \right) R
	= B^*(k) + B(-k) + \fD(k)^* + N \delta_{k,0} \;.
\end{equation}
For states $ \psi = R \xi $ with $ \cN \psi = N \psi $, so $ \d \Gamma(\chi(p \in B_{\F})) \xi = \d \Gamma(\chi(p \in B_{\F}^c)) \xi $, we obtain with~\eqref{eq:cHNmomentum}
\begin{equation}
\begin{aligned}
	R^* \cH_N R \xi
	&= (E_{\FB} + \HHH_0 + \QQQ_{\B} + \cE_1 + \cE_2 + \XXX) \xi \;,
	\end{aligned}
\end{equation}
where
\begin{equation}\label{eq:RHNR}
\begin{aligned}
	\HHH_0 &:= \sum_{p \in \Lambda_L} e(p) \hat{a}_p^* \hat{a}_p, \qquad \text{with } e(p) := \hbar^2 \lvert |p|^2 - k_{\F}^2 \rvert \;, \\
	\QQQ_{\B} &:= \frac{1}{2 \rho L^3} \sum_{k \in \Lambda_L}
		\hat{V}(k) \Big( 2 B^*(k) B(k) + B^*(k) B^*(-k) + B(-k) B(k) \Big) \;, \\
	\cE_1 &:= \frac{1}{2 \rho L^3} \sum_{k \in \Lambda_L}
		\hat{V}(k) \fD(k)^* \fD(k) \;, \\
	\cE_2 &:= \frac{1}{2 \rho L^3} \sum_{k \in \Lambda_L}
		\hat{V}(k) \Big( \fD(k)^* B(k) + \fD(k)^* B^*(-k) + \mathrm{h.c.} \Big) \;, \\
	\XXX &:= -\frac{1}{2 \rho L^3} \sum_{k \in \Lambda_L}
		\hat{V}(k) \Bigg( \sum_{p \in \BF^c \cap (\BF + k)} \hat{a}_p^* \hat{a}_p + \sum_{h \in \BF \cap (\BF^c - k)} \hat{a}_{-h}^* \hat{a}_{-h} \Bigg)\;. \\
\end{aligned}
\end{equation}

\subsection{Effective Hamiltonian and Quasi-Bosonic Bogoliubov Transformation}
\label{subsec:bogoliubov}

To extract the RPA energy, we will approximate the operator $ \HHH_0 + \QQQ_{\B} $ by an effective quasi-bosonic quadratic Hamiltonian, based on a patch decomposition of a neighbourhood of the Fermi surface. Subsequently, we approximately diagonalize this effective Hamiltonian by a quasi-bosonic Bogoliubov transformation, implemented by a unitary operator $ T : \cF \to \cF $.

We adopt a continuum version of the patches $ B_\alpha $ constructed around the Fermi sphere $ \partial B_{k_{\F}}(0) $ in \cite{BNPSS20,BNPSS21,BPSS22}. That is, we define $ (B_\alpha)_{\alpha = 1}^M \subset \RRR^3 $ following the construction in~\cite[Sect.~3]{BNPSS20}, while omitting the intersection with the lattice $ \ZZZ^3 $ in $ \partial B_{\F}^R $, where we choose the number of patches $ M $ independent of $ L $ as an even integer such that
\begin{equation}
M \sim \log(\rho^\eta)
\end{equation}
for some $ \eta \in (0,\frac{1}{46}) $ to be fixed later. Every patch covers approximately the same surface area $4 \pi k_{\F}^2/M $, has thickness of $ 2 R $, and is separated from the neighboring patches by a corridor of width larger than $ R $, which, we recall, is the radius of the support of $ \hat{V} $. Furthermore, the patches satisfy the regularity condition
\begin{equation}
	\diam(B_\alpha)
	\le C \rho^{\frac{1}{3}} M^{-\frac{1}{2}} \;.
\end{equation}
We denote the center of each patch $ B_\alpha $ by $ \omega_\alpha \in \RRR^3 $ with $ |\omega_\alpha| = k_{\F} $, and introduce the corresponding unit vector
\begin{equation}
	\hat{\omega}_\alpha := \frac{\omega_\alpha}{|\omega_\alpha|} \;.
\end{equation}
The patch construction is symmetric under reflections $ k \mapsto -k $. In particular, denoting the opposite patch of $ B_\alpha $ by $ B_{\bar\alpha} $ with $ \bar\alpha \in \{1,\ldots,M\} $ such that
\begin{equation} \label{eq:baralpha}
	\bar\alpha \equiv \alpha + M/2 \quad \textnormal{mod } M  \;,
\end{equation}
we have
\begin{equation}
	B_{\bar\alpha} = - B_\alpha \;.
\end{equation}

Given a momentum $ k \in \Lambda_L $, we proceed to define operators $ b^*_\alpha(k) $ that excite all particle--hole pairs $ (p,h) \in \Lambda_L \times \Lambda_L $ in patch $ B_\alpha $ with $ p = h + k $ such that $ h \in B_{\F} $ and $ p \in B_{\F}^c $. Since the indicator functions of $ B_{k_{\F}}(0) $ and $ B_\alpha $ are not differentiable, we regularize the patches as follows: We choose some cutoff scale $ \fl $, which depends on $ \rho $ but not on $ L $, such that
\begin{equation}
	\fl \sim M^{-\eta} \;,
\end{equation}
where we recall that $ \eta \in (0,\frac{1}{46}) $ is an arbitrarily small constant to be fixed later. In particular, $ \fl \ll 1 $ and $ \fl^{-2} \ll M $. Now, define the regularized indicator functions $ \hat{u}^\alpha, \hat{v}^\alpha: \RRR^3 \to \RRR $ with $ 0 \le \hat{u}^\alpha(p) \le 1 $, $ 0 \le \hat{v}^\alpha(p) \le 1 $ such that $ \hat{u}^{\bar\alpha}(p) = \hat{u}^\alpha(-p) $, $ \hat{v}^{\bar\alpha}(p) = \hat{v}^\alpha(-p) $ and
\begin{equation} \label{eq:uhat_vhat}
\begin{aligned}
	\hat{u}^\alpha(p) = \begin{cases}
		1 \quad &\textnormal{if} \; p \in B_{k_{\F}}(0)^c \cap B_\alpha \quad \textnormal{and} \quad \dist(p, \partial (B_{k_{\F}}(0)^c \cap B_\alpha)) > 2 \fl \;, \\
		0 \quad & \textnormal{if} \; p \notin B_{k_{\F}}(0)^c \cap B_\alpha \quad \textnormal{or} \quad \dist(p, \partial (B_{k_{\F}}(0)^c \cap B_\alpha)) < \fl \;,
	\end{cases} \\
	\hat{v}^\alpha(p) = \begin{cases}
		1 \quad &\textnormal{if} \; p \in B_{k_{\F}}(0) \cap B_\alpha \quad \textnormal{and} \quad \dist(p, \partial (B_{k_{\F}}(0) \cap B_\alpha)) > 2 \fl \;, \\
		0 \quad & \textnormal{if} \; p \notin B_{k_{\F}}(0) \cap B_\alpha \quad \textnormal{or} \quad \dist(p, \partial (B_{k_{\F}}(0) \cap B_\alpha)) < \fl \;,
	\end{cases}
\end{aligned}
\end{equation}
and where the smooth transition is such that any $ m $-th derivative, tracked by a multi-index $ \um \in \NNN_0^3 $, $ |\um| = m $, is bounded by
\begin{equation}\label{eq:smooth}
 \Vert \partial^{\um} \hat{u}^\alpha \Vert_\infty \le C_m \fl^{-m}\,, \qquad \Vert \partial^{\um} \hat{v}^\alpha \Vert_\infty \le C_m \fl^{-m} .
\end{equation}
(The reason why we set $ \hat{u}^\alpha $ and $ \hat{v}^\alpha $ to 0 close to the boundary of the patch is that we later use discrete derivatives, which can enlarge the support of functions by additional lattice sites around the boundary. The distance between $ \supp(\hat{u}^\alpha) $, $ \supp(\hat{v}^\alpha) $ and the boundary ensures that also discrete derivatives of $\hat{u}^\alpha$ and $\hat{v}^\alpha$ remain supported in the respective patch $B_\alpha$.) In the following, we will construct a suitable normalization constant $ \tilde{n}_{\alpha,k} $ and define the \emph{regularized pair creation operators} for $ k \in \Lambda_L $, $ \alpha \in \{1, \ldots, M\} $ as
\begin{equation}
\label{eq:bstar}
	b^*_\alpha(k)
	:= \frac{1}{\tilde{n}_{\alpha,k}} \sum_{p \in \Lambda_L} \hat{u}^\alpha(p) \hat{v}^\alpha(p-k) \hat{a}_p^* \hat{a}_{-p+k}^* \;,
\end{equation}
with its adjoint being $ b_\alpha(k) $. A natural candidate for $ \tilde{n}_{\alpha,k}^2 $, such that $ [b_\alpha(k), b_\alpha^*(k)] \approx 1 $, would be the regularized number of pairs in a patch
\begin{equation} \label{eq:n}
	n_{\alpha,k}^2 
	:= \sum_{p \in \Lambda_L} \hat{u}^\alpha(p)^2 \hat{v}^\alpha(p-k)^2 \;.
\end{equation}
However, this quantity has the drawback that $ n_{\alpha,k} \to 0 $ as $ \hat{k} \cdot \hat\omega_\alpha $ becomes negative, and derivatives like $ \nabla_k n_{\alpha,k}^{-1} $ become ill-behaved. To regularize $ n_{\alpha,k} $, we fix a smooth, non-decreasing function $ f_>: \RRR \to [0,\infty) $ such that
\begin{equation} \label{eq:f}
\begin{aligned}
	f_>(x) &= 1/2 \quad 			&&\textnormal{for } x < 0 \;, \\
	f_>(x) &\in [1/2,1] \quad	&&\textnormal{for } x \in [0,1] \;,\\
	f_>(x) &= x \quad 			&&\textnormal{for } x > 1 \;,
\end{aligned}
\end{equation}
and introduce
\begin{equation}
\label{eq:lambdatilde}
	\tilde\lambda_{\alpha,k}
	:= \fl f_>(\fl^{-1} k \cdot \hat{\omega}_\alpha ) \;,
\end{equation}
so that $ \tilde\lambda_{\alpha,k} \ge \fl/2 $. We then define the regularized normalization constant as
\begin{equation} \label{eq:ntilde}
	\tilde{n}_{\alpha,k}^2 
	:= \frac{L^3}{(2 \pi)^3}
		\frac{4 \pi k_{\F}^2}{M} \tilde\lambda_{\alpha,k} \;.
\end{equation}
We show in Lemma~\ref{lem:approxID} that $ \tilde{n}_{\alpha,k}^2 $ approximates $ n_{\alpha,k}^2 $. The operators $ b^*_\alpha(k) $ satisfy the following commutation relations.

\begin{lemma} [Approximate canonical commutation relations]
\label{lem:approxCCR}
For $ k, k' \in \Lambda_L $ and $ \alpha, \beta \in \{1, \ldots M\} $ we have
\begin{equation}
\label{eq:approxCCR}
	[b_\alpha(k), b^*_\beta(k')]
	= \delta_{\alpha, \beta} \big( \delta_{k, k'} I(k)_{\alpha,\alpha} + \hat\varepsilon_\alpha(k, k') \big) \;,
\end{equation}
with
\begin{equation}
\label{eq:cEalphahat}
\begin{aligned}
	I(k)_{\alpha,\beta}
	&:= \delta_{\alpha,\beta} \frac{n_{\alpha,k}^2}{\tilde{n}_{\alpha,k}^2} \;, \\
	\hat\varepsilon_\alpha(k, k')
	&:=
	- \sum_{p \in \Lambda_L} \frac{\hat{u}^\alpha(p)^2 \hat{v}^\alpha(p-k) \hat{v}^\alpha(p-k')}{\tilde{n}_{\alpha,k} \tilde{n}_{\alpha,k'}} \hat{a}_{-p+k'}^* \hat{a}_{-p+k} \\
	& \,\,\quad - \sum_{p \in \Lambda_L} \frac{\hat{u}^\alpha(p) \hat{u}^\alpha(p+k-k') \hat{v}^\alpha(p-k')^2}{\tilde{n}_{\alpha,k} \tilde{n}_{\alpha,k'}} \hat{a}_p^* \hat{a}_{p+k-k'} \;.\\
\end{aligned}
\end{equation}
\end{lemma}
\begin{proof}
Direct computation using~\eqref{eq:bstar} and the CAR~\eqref{eq:CAR}.
\end{proof}

\begin{remark}\label{rem:no_pairs}
We call $\{z \in \mathbb{R}^3: k \cdot z = 0\}$ the equator of the Fermi ball (with respect to the choice of $k$ defining the north pole). For patches clearly located in the southern hemisphere, we have $b^*_\alpha(k) = 0$ because the operator is given by a sum without summands. In particular, if $\alpha$ or $\alpha'$ indicates a patch far in the southern hemisphere, then $I(k)_{\alpha,\alpha'} = 0$. However, very close to the equator, the center $\omega_\alpha$ of the patch may be on the southern hemisphere, but parts of the patch still fall into the northern hemisphere. In this case we cannot be certain of $b^*_\alpha(k) = 0$. Given that patches cover an angle of order $M^{-1/2}$, this case is excluded by the condition $k\cdot \hat{\omega}_\alpha \ll - \lvert k\rvert M^{-1/2}$. Note that this condition is implied by $ k\cdot \hat{\omega}_\alpha \ll - M^{-1/2} $ since $ |k| \le R $. 
Actually, since $ M^{-\frac 12} \ll \fl $, we know that $ k \cdot \hat\omega_\alpha \ge \fl $ implies a patch completely in the northern hemisphere. In this case $I(k)_{\alpha,\alpha} \simeq 1$. Conversely,  $ k \cdot \hat\omega_\alpha < - \fl $ implies a patch completely in the southern hemisphere, so $ I(k)_{\alpha,\alpha} = 0 $. We can therefore think of it as approximately being the block matrix
\[
	I \simeq \begin{pmatrix} \mathbf{1} & 0 \\ 0 & 0 \end{pmatrix} \;,
\]
where $\mathbf{1}$ is the identity matrix of dimension $\tfrac{M}{2} \times \tfrac{M}{2}$. The off-diagonal elements in this formula are exact, the diagonal elements are approximate, and $ - \fl < k \cdot \hat\omega_\alpha < \fl $ corresponds to a transition region where $\textbf{1}$ and the matrix $0$ interpolate.

In \cite{BNPSS20,BNPSS21}, the regularization of $ n_{\alpha,k}^2 $ by $\tilde{n}_{\alpha,k}^2$ was not present, so the diagonal elements were exactly $1$; the zero blocks were not present because it was possible to explicitly restrict to the northern hemisphere.
\end{remark}

Our next goal is to approximate $ \HHH_0 + \QQQ_{\B} $ by an operator that is quadratic in $ b^*_\alpha(k)$ and $b_\alpha(k) $. Observe that
\begin{equation}
	B^*(k) \simeq \sum_{\alpha = 1}^M \tilde{n}_{\alpha,k} b_{\alpha}^*(k) \;.
\end{equation}
To keep track of $ (\alpha,k) $ where $ b^*_\alpha(k) $ does not create any particles, let us define the cutoff $ \hat\chi_{>,\alpha} $ through a smooth, non-decreasing cutoff function $ g_> $:
\begin{equation} \label{eq:chi>}
	\hat\chi_{>,\alpha}(k) := g_>(\fl^{-1} k \cdot \hat\omega_\alpha) \qquad \textnormal{with} \qquad
	g_>(x) = \begin{cases}
		1 \quad &\text{for } x > -1 \;, \\
		0 \quad &\text{for } x < -2 \;.
			\end{cases}
\end{equation}
Since $ b_{\alpha}^*(k) = 0 $ if $ k \cdot \hat\omega_\alpha < -\fl $, see Remark~\ref{rem:no_pairs}, we have 
\begin{equation} \label{eq:bstar_chi_insertion}
	b_{\alpha}^*(k) = \hat\chi_{>,\alpha}(k) b_{\alpha}^*(k) \;.
\end{equation}
The additional factor $ \hat\chi_{>,\alpha}(k) $ can now be included in matrix elements of $ K(k) $~\eqref{eq:K}, making them vanish whenever also $ b_{\alpha}^*(k) $ does. This will ensure smoothness in momentum space and therefore decay in position space.
The quadratic bosonic operator $ \QQQ_{\B} $ is approximated by the effective interaction operator
\begin{equation}
\begin{aligned}
\label{eq:QQQBcR}
	\QQQ_{\B}^{\cR}
	& := \frac{1}{2 \rho L^3} \sum_{k \in \Lambda_L}
		\hat{V}(k) \sum_{\alpha, \beta = 1}^M
		\Big( 2 \tilde{n}_{\alpha,k} \tilde{n}_{\beta,k} 
		\hat\chi_{>,\alpha}(k) \hat\chi_{>,\beta}(k)
		b_{\alpha}^*(k) b_{\beta}(k) \\
	& \qquad + \tilde{n}_{\alpha,k} \tilde{n}_{\beta,-k}
		\hat\chi_{>,\alpha}(k) \hat\chi_{>,\beta}(-k)
		\big( b_{\alpha}^*(k) b_{\beta}^*(-k)
		+ b_{\beta}(-k) b_{\alpha}(k) \big)
		\Big) \;.
\end{aligned}
\end{equation}
To bosonize the kinetic term $ \HHH_0 $, observe that for $ k \cdot \hat{\omega}_\alpha > 0 $, we have
\begin{equation}
	[\HHH_0, b^*_{\alpha}(k)]
	\simeq 2 \hbar^2 k_{\F} (k \cdot \hat{\omega}_\alpha) b^*_{\alpha}(k) \;.
\end{equation}
Also here, we replace $ k \cdot \hat{\omega}_\alpha $ by $ \tilde\lambda_{\alpha,k} \ge \fl/2 $, which allows controlling derivatives like $ \nabla_k \tilde\lambda_{\alpha,k}^{-1} $, and introduces a spectral gap of order $ \hbar\fl $. Moreover, we cut off momenta outside $ \supp(\hat V) $ by introducing a radial function $ \hat\chi_R \in C_0^\infty(\RRR^3) $ with $ \hat\chi_R(k) = 1 $ for $ |k| \le R $ and $ \supp(\hat\chi_R) = \overline{B_{2R}(0)} $. Then, in commutators with $ b^*_\alpha(k) $ and for $ |k| \le R $, the operator $ \HHH_0 $ approximately behaves like
\begin{equation}\label{eq:DB}
	\DDD_{\B} := 2 \hbar^2 k_{\F}
		\sum_{k \in \Lambda_L} \hat\chi_R(k) \sum_{\alpha = 1}^M
		\tilde\lambda_{\alpha,k} b^*_{\alpha}(k) b_{\alpha}(k) \;.
\end{equation}
Our effective Hamiltonian is $ \DDD_{\B} + \QQQ_{\B}^{\cR} $. To approximately diagonalize this operator, let us write it in the following form, using the $ \bar\alpha $-notation~\eqref{eq:baralpha}, so that $ \tilde{n}_{\bar\alpha,k} = \tilde{n}_{\alpha,-k} $ and $ \hat\chi_{>,\bar\alpha}(k) = \hat\chi_{>,\alpha}(-k) $:
\begin{equation} \label{eq:Bosonizedoperator}
\begin{aligned}
	\DDD_{\B} + \QQQ_{\B}^{\cR}
	= \hbar^2 k_{\F}
		\sum_{k \in \Lambda_L} \sum_{\alpha, \beta = 1}^M
		\Big( &
		2 (h + P)(k)_{\alpha, \beta}
		 b_{\alpha}^*(k) b_{\beta}(k) \\
		 &
		+ P(k)_{\alpha, \bar\beta}
		 \big( b_{\alpha}^*(k) b_{\beta}^*(-k)
		 + b_{\beta}(-k) b_{\alpha}(k) \big)
		 \Big) \;,
\end{aligned}
\end{equation}
with matrices $ h(k), P(k) \in \RRR^{M \times M} $ given by 
\begin{equation}
\label{eq:hP}
	h(k)_{\alpha,\beta} := \delta_{\alpha,\beta} \tilde\lambda_{\alpha,k} \hat\chi_R(k) \;, \quad
	P(k)_{\alpha,\beta} := \tilde{v}_{\alpha, k} \tilde{v}_{\beta, k} \;, \quad
	\tilde{v}_{\alpha, k}^2 := \frac{\hat{V}(k)}{2 k_{\F} \hbar} \frac{\tilde{n}_{\alpha,k}^2}{\rho^{\frac 23} L^3} \hat\chi_{>,\alpha}(k)^2 \;.
\end{equation}

The approximate diagonalization of $ \QQQ_{\B}^{\cR} + \DDD_{\B} $ is performed as
in~\cite{BNPSS20}: We introduce the unitary, approximate Bogoliubov implementation $ T = T_1 $, where $ T_\lambda : \cF \to \cF$, $\lambda \in [-1,1] $ is given by
\begin{equation} \label{eq:T}
	T_\lambda := e^{-\lambda S}\;, \qquad
	S := S_- - S_-^* \;, \qquad
	S_- := \frac{1}{2} \sum_{k \in \Lambda_L} \sum_{\alpha, \beta = 1}^M K(k)_{\alpha, \bar\beta}  b_\alpha(k) b_\beta(-k)  \;,
\end{equation}
with $ K(k) \in \RRR^{M \times M} $ the symmetric matrix defined as
\begin{equation} \label{eq:K}
\begin{aligned}
	K(k) &:= - \frac 12 \log \left( h(k)^{-\frac 12}
		\Big( h(k)^2 + 2 h(k)^{\frac 12} P(k) h(k)^{\frac 12} \Big)^{\frac 12}
		h(k)^{-\frac 12} \right) \qquad \mathrm{for } \;  k \in \supp(\hat\chi_R) \;,
\end{aligned}
\end{equation}
and 0 otherwise.
Note that $ K(k)_{\alpha,\beta} = 0 $ if $ k \notin \supp(\hat\chi_{>,\alpha} \hat\chi_{>,\beta}) $, and already for $ |k| > R $, see Lemma~\ref{lem:regularity_K}. This is the reason why we introduced the cutoff $ \hat\chi_{>,\alpha} $: It forces matrix elements of $ K $ to 0, whenever there is no pair creation possible.

Moreover we have the symmetries
\begin{equation}
	h(k)_{\alpha,\beta} = h(k)_{\beta,\alpha} = h(-k)_{\bar\alpha,\bar\beta}
\qquad
\text{and}
\qquad
	P(k)_{\alpha,\beta} = P(k)_{\beta,\alpha} = P(-k)_{\bar\alpha,\bar\beta} \;,
\end{equation}
and therefore also
\begin{equation}
	K(k)_{\alpha,\beta} = K(k)_{\beta,\alpha} = K(-k)_{\bar\alpha,\bar\beta} \;.
\end{equation}

Our trial state is defined as
\begin{equation} \label{eq:trialstate}
	\psi := R T \Omega \;.
\end{equation}
Indeed, $ \psi $ is an $ N $-particle state: using the definition of the particle-hole transformation one shows that $ R^* \cN R = N + \d \Gamma(\chi(p \in B_{\F}^c)) - \d \Gamma(\chi(p \in B_{\F})) $. Moreover we have the obvious
\begin{equation}
    [S,\d \Gamma(\chi(p \in B_{\F}^c)) - \d \Gamma(\chi(p \in B_{\F}))] = 0 \;.
\end{equation}
Therefore, by the Duhamel formula, $(\d \Gamma(\chi(p \in B_{\F}^c)) - \d \Gamma(\chi(p \in B_{\F}))) T\Omega =0$ and $\cN R T \Omega = N R T \Omega$.

\section{Position Space Description}
\label{sec:positionspace}

\subsection{Pair Operators in Position Space}
\label{subsec:positionspace_rewriting}

In many bounds, we will exploit the decay of certain quantities at large distances, which is best done in a position space description. For instance, bounding $ S $ will require decay of the Fourier inverse of the function $ k \mapsto \frac{K(k)_{\alpha, \beta}}{\tilde{n}_{\alpha, k} \tilde{n}_{\beta, k}} $.

In general, we will encounter functions $ k \mapsto \frac{X(k)_{\alpha, \beta}}{\tilde{n}_{\alpha, k} \tilde{n}_{\beta, k}} $ with some $ X: \Lambda_L \to \RRR^{M \times M} $, $ X \in \ell^1(\Lambda_L) $ that (approximately) scales like $ X(k)_{\alpha, \beta} \sim M^{-1} $, while $ \tilde{n}_{\alpha, k}^2 \sim M^{-1} \rho^{\frac 23} L^3 $. This motivates the introduction of the (approximately) normalized functions
\begin{equation}
\label{eq:varphi}
	\hat\varphi_X(k)_{\alpha, \beta}
	:= L^3 \rho^{\frac 23} \frac{X(k)_{\alpha, \beta}}{\tilde{n}_{\alpha, k} \tilde{n}_{\beta, k}} \;, \qquad
	\varphi_X(z)_{\alpha, \beta}
	:= \frac{1}{L^3} \sum_{k \in \Lambda_L} \hat\varphi_X(k)_{\alpha, \beta} e^{i k \cdot z} \;,
\end{equation}
in agreement with our Fourier convention~\eqref{eq:Fourier_convention}.

Representing operators like $ S $ in position space involves the inverse Fourier transforms of $ \hat{u}^\alpha $ and $ \hat{v}^\alpha $ from~\eqref{eq:uhat_vhat}. More precisely, we introduce the Fourier inverses $ u^\alpha, v^\alpha : \TTT_L \to \CCC $ for $ \alpha \in \{1, \ldots, M\} $:
\begin{equation}
\label{eq:uxvx}
	u^\alpha(x) := \frac{1}{L^3} \sum_{p \in \Lambda_L} \hat{u}^\alpha(p)
		e^{i p\cdot x} \;, \qquad
	v^\alpha(x) := \frac{1}{L^3} \sum_{p \in \Lambda_L} \hat{v}^\alpha(p)
		e^{i p\cdot x} \;,
\end{equation}
and we use the notation $ u_y^\alpha(x) = u^\alpha(x-y) $ and $ v_y^\alpha(x) = v^\alpha(x-y) $ for shifted functions. Here $ u^\alpha(x-y) $ and $ v^\alpha(x-y) $ are the kernels of some equally denoted operators $ 0 \le u^\alpha \le 1 $ and $ 0 \le v^\alpha \le 1 $, where
\begin{equation}
\label{eq:uv}
	u^\alpha = \sum_{p \in \Lambda_L} \hat{u}^\alpha(p) |f_p \rangle \langle f_p| \;, \qquad
	v^\alpha = \sum_{p \in \Lambda_L} \hat{v}^\alpha(p) |f_p \rangle \langle f_p| \;,
\end{equation}
where we recall that $ (f_p)_{p \in \Lambda_L} $ is the orthonormal basis of plane waves. Note that $ u^\alpha v^\beta = 0 = v^\beta u^\alpha $ for all $ \alpha,\beta \in \{1,\ldots,M\} $. One easily verifies that
\begin{equation}
	\frac{1}{L^{3/2}} \int \d x \; e^{i p \cdot x} a^*(u_x^\alpha) = \hat{u}^\alpha(p) \hat{a}_p^*
	 \;, \quad
	\frac{1}{L^{3/2}} \int \d x \; e^{-i h \cdot x} a^*(v_x^{\bar\alpha}) = \hat{v}^\alpha(h) \hat{a}_{-h}^*
	= \hat{v}^{\bar\alpha}(-h) \hat{a}_{-h}^* \;.
\end{equation}
Introducing the pair operators
\begin{equation} \label{eq:bx_bxstar}
	b_{\alpha,x}
	:= a(v_x^{\bar\alpha}) a(u_x^\alpha) \;, \qquad
	b^*_{\alpha,x}
	:= a^*(u_x^\alpha) a^*(v_x^{\bar\alpha}) \;,
\end{equation}
by a computation using $ \tilde{n}_{\bar\beta, k} = \tilde{n}_{\beta, -k} $, we obtain the transformation law
\begin{equation}
    b^*_{\alpha,x} = \frac{1}{L^3} \sum_{k \in \Lambda_L} e^{-ik\cdot x} \tilde{n}_{\alpha,k} b^*_\alpha(k)\;, \qquad b^*_\alpha(k) = \frac{1}{\tilde{n}_{\alpha,k}} \int \di x e^{ik\cdot x} b^*_{\alpha,x} \;.
\end{equation}
This implies
\begin{equation} \label{eq:positionspace_trafo}
\begin{aligned}
	\sum_{k \in \Lambda_L} X(k)_{\alpha, \beta} b^*_\alpha(k) b_\beta(k)
	&= \frac{1}{\rho^{\frac 23}} \int \d x \d y \; \varphi_X(x-y)_{\alpha,\beta} b^*_{\alpha,x} b_{\beta,y} \;, \\
	\sum_{k \in \Lambda_L} X(k)_{\alpha, \bar\beta} b_\alpha^*(k) b_\beta^*(-k)
	&= \frac{1}{\rho^{\frac 23}} \int \d x \d y \; \varphi_X(x-y)_{\alpha,\bar\beta} b^*_{\alpha,x} b^*_{\beta,y} \;. \\
\end{aligned}
\end{equation}
In particular, we may express
\begin{equation}
\label{eq:S}
	S
	= \frac 12 \sum_{k \in \Lambda_L} \sum_{\alpha, \beta = 1}^M
		\frac{K(k)_{\alpha, \bar\beta}}{\tilde{n}_{\alpha, k} \tilde{n}_{\bar\beta, k}}
		\sum_{p, q \in \Lambda_L} \hat{v}^\alpha(p-k) \hat{u}^\alpha(p) \hat{v}^\beta(q+k) \hat{u}^\beta(q)
		\hat{a}_{-p+k} \hat{a}_p \hat{a}_{-q-k} \hat{a}_q - \mathrm{h.c.}
\end{equation}
in position space as
\begin{equation}
\label{eq:S_position}
	S = \frac{1}{2 \rho^{\frac 23}} \sum_{\alpha, \beta = 1}^M
		\int \d x \d y \; \varphi_K(y-x)_{\alpha, \bar\beta}
		a(v_x^{\bar\alpha}) a(u_x^\alpha) a(v_y^{\bar\beta}) a(u_y^\beta) - \mathrm{h.c.} \;,
\end{equation}
where we used that $ K(k)_{\alpha,\beta} $ is real, so $ \overline{\varphi_K(x-y)_{\alpha, \bar\beta}} = \varphi_K(y-x)_{\alpha, \bar\beta} $.

\subsection{Fourier Bounds and Discrete Derivatives}
\label{subsec:discrete_derivatives}

Within our Fourier convention~\eqref{eq:Fourier_convention}, it is natural to define the $ \ell^1 $- and $ \ell^2 $-norms as
\begin{equation} \label{eq:ell_norms}
	\Vert \hat{f} \Vert_1 := \frac{1}{L^3} \sum_{k \in \Lambda_L} |\hat{f}(k)| \;, \qquad
	\Vert \hat{f} \Vert_2^2 := \frac{1}{L^3} \sum_{k \in \Lambda_L} |\hat{f}(k)|^2 \;.
\end{equation}
This leads to the following versions of standard bounds; we omit the proof.

\begin{lemma} [Fourier bounds] \label{lem:FT_bounds}
Let $ \hat{g}_1 \in \ell^1(\Lambda_L) $ and $ \hat{g}_2 \in \ell^2(\Lambda_L) $ with inverse Fourier transforms $ g_1 $ and $ g_2 $. Then,
\begin{equation} \label{eq:FT_bounds}
	\Vert g_1 \Vert_\infty \le \Vert \hat{g}_1 \Vert_1 \;, \qquad
	\Vert g_2 \Vert_2 = \Vert \hat{g}_2 \Vert_2 \;.
\end{equation}
Moreover, if $ f_1, f_2 \in L^2(\TTT_L) $ with position and momentum space multiplication operators $ A_{f_1}, B_{f_2}: L^2(\TTT_L) \to L^2(\TTT_L) $,
\begin{equation}
	(A_{f_1} g)(x)
	:= f_1(x) g(x) \;, \qquad
	(\widehat{B_{f_2} g})(k)
	:= \hat{f}_2(k) \hat{g}(k) \;,
\end{equation}
then the Hilbert--Schmidt norm of $ A_{f_1} B_{f_2} $ is
\begin{equation} \label{eq:KSS_equality}
	\Vert A_{f_1} B_{f_2} \Vert_{\HS}
	= \Vert f_1 \Vert_2 \Vert \hat{f}_2 \Vert_2 \;.
\end{equation}
\end{lemma}

Since the Fourier transforms $ \hat{f} $ are functions on the discrete momentum lattice $ \Lambda_L $ with spacing $ 2 \pi L^{-1} $, we introduce discrete derivatives in Fourier space. We define the right and left discrete derivatives $ \partial_{k_j}$ and $\bar\partial_{k_j} $, $ j \in \{1,2,3\} $, as
\begin{equation} \label{eq:partial_k}
	(\partial_{k_j} \hat{f})(k)
	:= \frac{\hat{f}(k + 2 \pi L^{-1} e_j) - \hat{f}(k)}{2 \pi L^{-1}} \;, \qquad
	(\bar\partial_{k_j} \hat{f})(k)
	:= \frac{\hat{f}(k) - \hat{f}(k - 2 \pi L^{-1} e_j)}{2 \pi L^{-1}} \;,
\end{equation}
where $ e_j $ is the $ j $-th unit vector. One easily verifies $ [\partial_{k_i}, \bar\partial_{k_{j}}] = 0 $ for $ i,j \in \{1,2,3\} $, as well as the product rule
\begin{equation} \label{eq:productrule}
\begin{aligned}
	\partial_{k_j} (\hat{f} \hat{g})(k)
	&= (\partial_{k_j} \hat{f})(k) \hat{g}(k + 2 \pi L^{-1} e_j)
		+ \hat{f}(k)(\partial_{k_j} \hat{g})(k) \;, \\
	\bar\partial_{k_j} (\hat{f} \hat{g})(k)
	&= (\bar\partial_{k_j} \hat{f})(k) \hat{g}(k - 2 \pi L^{-1} e_j)
		+ \hat{f}(k)(\bar\partial_{k_j} \hat{g})(k) \;.
\end{aligned}
\end{equation}
The discrete Laplacian is defined as
\begin{equation} \label{eq:Delta_k}
	- \Delta_k
	:= - \sum_{j=1}^3 \bar\partial_{k_j} \partial_{k_j} \;.
\end{equation}
Denoting the inverse Fourier transform on the torus by $ F^{-1}(\hat{f}) = f $, the derivatives $ i \partial_{k_j} $, $ i \bar\partial_{k_j} $, and $ - \Delta_k $ amount to multiplication operators in position space as follows
\begin{equation} \label{eq:derivatives_position_multiplication}
\begin{aligned}
	F^{-1} (i \partial_{k_j} \hat{f}) (x)
	&= e^{-i \pi L^{-1} x_j} \frac{\sin(\pi L^{-1} x_j)}{\pi L^{-1}} f(x)
	=: \cX_j(x) f(x) \;, \\
	F^{-1} (i \bar\partial_{k_j} \hat{f}) (x)
	&= e^{i \pi L^{-1} x_j} \frac{\sin(\pi L^{-1} x_j)}{\pi L^{-1}} f(x)
	= \overline{\cX_j(x)} f(x) \;, \\
	F^{-1} (- \Delta_k \hat{f})(x)
	&= \sum_{j=1}^3 \frac{1 - \cos(2 \pi L^{-1} x_j)}{2 \pi^2 L^{-2}} f(x)
	= \sum_{j=1}^3 |\cX_j(x)|^2 f(x) \;.
\end{aligned}
\end{equation}
Correspondingly, we will introduce the position space localization function $ \cW \in L^1(\TTT_L) $,
\begin{equation}
\label{eq:cWz}
	\cW(x) := \Bigg( 1 + \sum_{j=1}^3 \frac{1 - \cos(2 \pi L^{-1} x_j)}{2 \pi^2 L^{-2}} \Bigg)^{-2}
	= \Bigg( 1 + \sum_{j=1}^3 |\cX_j(x)|^2 \Bigg)^{-2} \;,
\end{equation}
and denote $ \cW_z(x) := \cW(x-z) $, so $ \Vert \cW_z \Vert_1 \le C $ (independent of $L$), and multiplication by $ \cW^{-1} $ in position representation corresponds to application of the operator $ (1-\Delta_k)^2 $ in momentum representation. The corresponding localized number operator is
\begin{equation}
\label{eq:cNz}
	\cN_z
	:= \d \Gamma(\cW_z)
	= \int \d x \; \cW_z(x) a_x^* a_x  \;.
\end{equation}
We may think of $ \cW(x) $ as a torus version of the localization function $ \frac{1}{(1+\lvert x\rvert^2)^2} $. Moreover, taking a product $ \cW_z \cW_{z'} $ will localize $ z $ near $ z' $, as shown by the following lemma.

\begin{lemma}[Products of $ \cW $] \label{lem:Peetre}
There exists $ C > 0 $ such that for all $ z,z' \in \TTT_L $ we have
\begin{equation}
\label{eq:Peetre}
	\cW(z-z')^{-1}
	\le C(\cW_z^{-1} + \cW_{z'}^{-1}) \;, \qquad
	\cW_z \cW_{z'}
	\le C \cW(z-z') \;.
\end{equation}
\end{lemma}

\begin{proof}
For any $ z \in \TTT_L $, define
\begin{equation}
	\langle z \rangle := \sqrt{1 + \lvert z \rvert^2} \;.
\end{equation}
From $ \left( \frac{2}{\pi} \right)^2 \frac{x^2}{2} \le 1 - \cos(x) \le \frac{x^2}{2} $ on the interval $[-\pi,\pi]$, we obtain
\begin{equation}
	\langle z \rangle^{-4}
	\le \cW(z)
	\le \Big( \frac{\pi}{2} \Big)^4 \langle z \rangle^{-4} \;.
\end{equation}
So the first claim follows if we can show that there exists $C > 0$ such that
\begin{equation}
	\langle z-z' \rangle^4
	\le C(\langle z-x \rangle^4 + \langle x-z' \rangle^4) \qquad \forall x,z,z' \in \TTT_L \;.
\end{equation}
In fact, by the triangle inequality, $ |z-x| \ge \frac{|z-z'|}{2} $ or $ |x-z'| \ge \frac{|z-z'|}{2} $, so
\begin{equation}
	\langle z-x \rangle^4 + \langle x-z' \rangle^4
	\ge \Big\langle \frac{z-z'}{2} \Big\rangle^4
	\ge \frac{1}{16} \langle z-z' \rangle^4 \;.
\end{equation}
The second claim follows by taking the reciprocal of the first, then multiplying it with $ \cW_z \cW_{z'} $ on both sides and using $ \cW_z ,\cW_{z'} \le 1 $.
\end{proof}

\section{Bounds on Number of Excitations}
\label{sec:gronwall_bounds}

The central result of this section are estimates on the number of excitations operator $ \cN $ and localized number of excitations operator $ \cN_z $~\eqref{eq:cNz}. The bounds hold on translation invariant states. To be precise, we define the translation operator
\begin{equation} \label{eq:cTy}
	\cT_y: L^2(\TTT_L) \to L^2(\TTT_L) \;, \qquad
	(\cT_y f)(x) := f(x-y) \;,
\end{equation}
as well as its lift to Fock space, which acts on the $ N $-particle sector as
\begin{equation} \label{eq:Gamma_cTy}
	(\Gamma(\cT_y) \xi)(x_1,\ldots,x_N)
	:= \xi(x_1-y,\ldots,x_N-y) \;.
\end{equation}
We call $ \xi $ \emph{translation invariant} if $ \Gamma(\cT_y) \xi = \xi $ for all $ y \in \TTT_L $.

\begin{proposition}[Grönwall estimates] \label{prop:gronwall}
Recall $ T_\lambda = e^{-\lambda S} $ from \eqref{eq:T}. Then, there exists $ C > 0 $ such that for $ \lambda \in [-1,1] $, $ z \in \TTT_L $, and for translation-invariant $ \xi \in \cF $,
\begin{equation}
\label{eq:gronwall}
\begin{aligned}
	\langle \xi, T_\lambda^* (\cN + L^3) T_\lambda \xi \rangle
	\le &e^{C \fl^{-7}} \langle \xi, (\cN + L^3) \xi \rangle \;,\\
	\langle \xi, T_\lambda^* (\cN_z + 1)^2 T_\lambda \xi \rangle
	\le &e^{C \fl^{-17} M^{\frac 12}} \langle \xi, (\cN_z + 1)^2 \xi \rangle \;.
\end{aligned}
\end{equation}
\end{proposition}

Our trial state~\eqref{eq:trialstate} is $ \psi = R T \Omega $, which corresponds to the excitation vector $ \xi = T \Omega $. The vacuum vector $ \Omega $ is translation-invariant. Moreover, $\Gamma(\cT_y) = e^{-i y \cdot \cP} $ for $ \cP := \sum_{k \in \Lambda_L} k \hat{a}^*_k \hat{a}_k $ the total momentum operator. Since $ S $ commutes with $\cP$, it follows that $ T_\lambda \Omega $ is indeed translation invariant. Therefore Proposition~\ref{prop:gronwall} provides us with the bounds
\begin{equation} \label{eq:preliminary_bounds}
	\langle T_\lambda \Omega, (\cN + L^3) T_\lambda \Omega \rangle
	\le e^{C \fl^{-7}} L^3 \;, \qquad
	\langle T_\lambda \Omega, (\cN_z + 1)^2 T_\lambda \Omega \rangle
	\le e^{C \fl^{-17} M^{\frac 12}} \;,
\end{equation}
which are not uniform in $ \rho $.
While the second bound will turn out to be sufficient, the estimate on $ \cN $ will have to be improved in Proposition~\ref{prop:improved_excitation} to the extensive bound $ \langle T_\lambda \Omega, \cN T_\lambda \Omega \rangle \le C L^3 $.

\medskip

Proposition~\ref{prop:gronwall} will be proven in Subsection~\ref{subsec:gronwall}, using the lemmas we establish in Subsections~\ref{subsec:opest} and \ref{subsec:HSest}.

\subsection{Bounds by Localized Number Operators}
\label{subsec:opest}

\begin{lemma}[Bounding $ a^\sharp a^\sharp $ by $\Ncal_z^{1/2}$] \label{lem:aastar_conversion}
Let $ \xi \in \cF $ and consider a compact operator $ A $ on $ L^2(\TTT_L) $ with integral kernel $ A(x,x') $. Suppose that $ A \cW_z^{-\frac 12} $ is Hilbert--Schmidt. Then, we have the following bounds, where $ \sharp \in \{ *, \cdot\} $:
\begin{equation} \label{eq:aastar_conversion}
\begin{aligned}
	\Big\Vert \int \d x \d x' A(x,x') a_x^\sharp a_{x'} \xi \Big\Vert 
	&\le \Vert A \cW_z^{-\frac 12} \Vert_{\HS}
	\Vert \cN_z^{\frac 12} \xi \Vert \;, \\
	\Big\Vert \int \d x \d x' A(x,x') a^*_x a^*_{x'} \xi \Big\Vert 
	&\le 2 \Vert A \cW_z^{-\frac 12} \Vert_{\HS}
		\Vert (\cN_z + 1)^{\frac 12} \xi \Vert \;.
\end{aligned}
\end{equation}
\end{lemma}

\begin{proof}
For the first bound, we use $ \Vert a^\sharp(f) \Vert_{\op} \le \Vert f \Vert_2 $:
\begin{equation} \label{eq:aabound_in_proof}
\begin{aligned}
	&\Big\Vert \int \d x \d x' A(x,x') a_x^\sharp a_{x'} \xi \Big\Vert 
	\le \int \d x' \Vert A(\cdot, x') \Vert_2 \Vert a_{x'} \xi \Vert \\
	&\le \Big( \int \d x \d x' |A(x, x')|^2 \cW_z(x')^{-1} \Big)^{\frac 12}
	\Big( \int \d x' \cW_z(x') \Vert a_{x'} \xi \Vert^2 \Big)^{\frac 12}\\
	&= \Vert A \cW_z^{-\frac 12} \Vert_{\HS}
	\Vert \cN_z^{\frac 12} \xi \Vert \;.
\end{aligned}
\end{equation}
For the second bound, we first normal order. Introducing $ A^T $ and $ \overline{A} $, with kernels $ A^T(x,x') := A(x',x) $ and $ \overline{A}(x,x') := \overline{A(x,x')} $, we obtain
\begin{equation} \label{eq:aastarconversion_commutation}
\begin{aligned}
	&\Big\Vert \int \d x \d x' A(x,x') a^*_x a^*_{x'} \xi \Big\Vert^2 = \int \d x \d x' \d y \d y'
		\overline{A(x,x')} A(y,y') \langle \xi, a_{x'} a_x a_y^* a_{y'}^* \xi \rangle \\
	&= \Big\Vert - \int \d x \d x' \overline{A(x,x')} a_x a_{x'} \xi \Big\Vert^2
	- \langle \xi, \d \Gamma(A A^*) \xi \rangle
	- \langle \xi, \d \Gamma(A^T \overline{A}) \xi \rangle \\
	&\quad + \langle \xi, \d \Gamma(A \overline{A}) \xi \rangle
	+ \langle \xi, \d \Gamma(A^T A^*) \xi \rangle
	+ \tr(A^* A) \Vert \xi \Vert^2
	- \tr(\overline{A} A) \Vert \xi \Vert^2 \\
	&\le \Vert \overline{A} \cW_z^{-\frac 12} \Vert_{\HS}^2
	\Vert \cN_z^{\frac 12} \xi \Vert^2
	+ \langle \xi, \d \Gamma(A \overline{A}) \xi \rangle
	+ \langle \xi, \d \Gamma(A^T A^*) \xi \rangle
	+ 2 \Vert A \Vert_{\HS}^2 \Vert \xi \Vert^2 \;,
\end{aligned}
\end{equation}
where in the last line, we used~\eqref{eq:aabound_in_proof}, as well as $ AA^* \ge 0 $, $ A^T \overline{A} \ge 0 $ and $ |\tr(\overline{A} A)| \le \Vert \overline{A} \Vert_{\HS} \Vert A \Vert_{\HS} = \Vert A \Vert_{\HS}^2 $. The $\d \Gamma$-operators are estimated by
\begin{equation} \label{eq:diGamma_bound_startegy}
\begin{aligned}
	&|\langle \xi, \d \Gamma(A^* B) \xi \rangle|
	= \Big\vert \int \d x \d y (A^* B)(x,y) \langle \xi, a_x^* a_y \xi \rangle \Big\vert
	\le \int \d x' 
		\Vert a(\overline{A}(x',\cdot)) \xi \Vert
		\Vert a(\overline{B}(x', \cdot)) \xi \Vert \\
	&\le \Big( \int \d x' \Vert \overline{A}(x',\cdot) \Vert_2^2 \Big)^{\frac 12}
		\Big( \int \d x' \Vert \overline{B}(x', \cdot) \Vert_2^2 \Big)^{\frac 12}
		\norm{\xi}^2
	= \norm{A}_{\HS} \norm{B}_{\HS} \norm{\xi}^2 \;,
\end{aligned}
\end{equation}
so
we conclude
\begin{equation}
	\Big\Vert \int \d x \d x' A(x,x') a^*_x a^*_{x'} \xi \Big\Vert^2
	\le \Vert A \cW_z^{-\frac 12} \Vert_{\HS}^2
	\Vert \cN_z^{\frac 12} \xi \Vert^2
	+ 4 \Vert A \Vert_{\HS}^2 \Vert \xi \Vert^2 \;.
\end{equation}
With $ \Vert A \Vert_{\HS} \le \Vert A \cW_z^{-\frac 12} \Vert_{\HS} $ (since $ 0 \le \cW_z^{\frac 12} \le 1 $), this concludes the proof.
\end{proof}

We will often use the following corollary of Lemma~\ref{lem:aastar_conversion}.

\begin{corollary}[Bounding pair operators by $\Ncal_z^{1/2}$] \label{cor:bbstar_conversion}
Let $ z \in \TTT_L $, $ \xi \in \cF $ and consider a family of functions $ (\varphi_\alpha)_{\alpha = 1}^M \subset L^\infty(\TTT_L) $, where $ \varphi_{\alpha,z} $ is the multiplication operator defined by $ (\varphi_{\alpha,z} f)(x) := \varphi_\alpha(x-z) f(x) $. Then, with $ b_{\alpha,x} = a(v^{\bar\alpha}_x) a(u^\alpha_x) $ as defined in~\eqref{eq:bx_bxstar},
\begin{equation} \label{eq:bbstar_conversion}
\begin{aligned}
	\Bigg\Vert \sum_{\alpha = 1}^M \int \d x \; \varphi_\alpha(x-z) b_{\alpha,x} \xi \Bigg\Vert 
	&\le \Bigg\Vert \sum_{\alpha = 1}^M v^\alpha \varphi_{ \alpha,z} u^\alpha \cW_z^{-\frac 12} \Bigg\Vert_{\HS}
	\Vert \cN_z^{\frac 12} \xi \Vert \;, \\
	\Bigg\Vert \sum_{\alpha = 1}^M \int \d x \; \varphi_\alpha(x-z) b^*_{\alpha,x} \xi \Bigg\Vert 
	&\le 2 \Bigg\Vert \sum_{\alpha = 1}^M u^\alpha \varphi_{ \alpha,z} v^\alpha \cW_z^{-\frac 12} \Bigg\Vert_{\HS}
		\Vert (\cN_z + 1)^{\frac 12} \xi \Vert \;.
\end{aligned}
\end{equation}
\end{corollary}

\begin{proof}
Recall $ u^\alpha_x(y) = u^\alpha(y-x) $. Here, $ \hat{u}^\alpha $ is real-valued with $ \hat{u}^\alpha(-k) = \hat{u}^{\bar\alpha}(k) $, so $ \overline{u^\alpha(x)} = u^\alpha(-x) = u^{\bar\alpha}(x) $, and the same is true for $ v^\alpha $. Thus,
\begin{equation}
\begin{aligned}
	&\int \d x \; \varphi_\alpha(x-z) b_{\alpha,x}
	= \int \d x \d y \d y' \varphi_{\alpha,z}(x)
		\overline{v^{\bar\alpha}(y-x)} \;
		\overline{u^\alpha(y'-x)}
		a_y a_{y'} \\
	&= \int \d y \d y' \big( v^\alpha \varphi_{\alpha,z} u^\alpha \big)(y,y')
		a_y a_{y'} \;.
\end{aligned}
\end{equation}
Analogously, we may re-write $ \int \d x \; \varphi_\alpha(x-z) b_{\alpha,x}^* $, and Lemma~\ref{lem:aastar_conversion} then applies.
\end{proof}

We also need to estimate two pairs of creation and annihilation operators by localized number operators, for which the following lemma is useful.

\begin{lemma}[Pull-through bounds] \label{lem:pullthroughbound}
Let $ x, z \in \TTT_L $. Then, for any measurable function $ f : [0,\infty) \to [0,\infty) $, we have the pull--through formulas
\begin{equation} \label{eq:pullthrough}
	f(\cN_z) a^*_x = a^*_x f(\cN_z + \cW_z(x))  \;, \qquad
	f(\cN_z + \cW_z(x)) a_x = a_x f(\cN_z) \;.
\end{equation}
Moreover, for $ g \in L^2(\TTT_L) $ and $c \in [0,\infty)$, we have the estimate
\begin{equation} \label{eq:pullthroughbound}
\begin{aligned}
	\norm{(\cN_z + c)^{\frac 12} a^*(g) \xi}
	& \le \norm{g}_2 \norm{(\cN_z + c + 3)^{\frac 12} \xi} \;,\\
	\norm{(\cN_z + c)^{\frac 12} a(g) \xi}
	& \le \norm{g}_2 \norm{(\cN_z + c)^{\frac 12} \xi} \;.
\end{aligned}
\end{equation}
\end{lemma}
\begin{proof}
Eq.~\eqref{eq:pullthrough} follows by functional calculus, using that $ \cN_z a^*_x = a^*_x (\cN_z + \cW_z(x))$ due to the CAR.
To prove~\eqref{eq:pullthroughbound}, we use~\eqref{eq:pullthrough}, the CAR,
and then again \eqref{eq:pullthrough}:
\begin{equation}
\begin{aligned}
	&\norm{(\cN_z + c)^{\frac 12} a^*(g) \xi}^2
	= \int \d x \d y \; \overline{g(x)} g(y)
	\big\langle (\cN_z + c + \cW_z(x))^{\frac 12} \xi, a_x  a^*_y (\cN_z + c + \cW_z(y))^{\frac 12}\xi \big\rangle \\
	&= \int \d x \; |g(x)|^2 \big\langle \xi, (\cN_z + c + \cW_z(x)) \xi \rangle \\
	&\quad - \int \d x \d y \; \overline{g(x)} g(y)
	\big\langle (\cN_z + c + \cW_z(x))^{\frac 12} \xi, a^*_y a_x  (\cN_z + c + \cW_z(y))^{\frac 12}\xi \big\rangle  \\
	&= \int \d x \; |g(x)|^2 \big\langle \xi, (\cN_z + c + \cW_z(x)) \xi \rangle \\
	&\quad - \int \d x \d y \; \overline{g(x)} g(y)
	\big\langle \xi, a^*_y (\cN_z + c + \cW_z(x) + \cW_z(y)) a_x  \xi \big\rangle \;.
\end{aligned}
\end{equation}
Using $ 0 \le \cW_z(x) \le 1 $, the first term is bounded by
\begin{equation} \label{eq:pullthrough_firstbound}
	\int \d x \; |g(x)|^2 \big\langle \xi, (\cN_z + c + \cW_z(x)) \xi \rangle
	\le \norm{g}_2^2 \norm{(\cN_z + c + 1)^{\frac 12} \xi}^2  \;.
\end{equation}
Moreover we have
\begin{equation}
	- \int \d x \d y \; \overline{g(x)} g(y)
	\langle \xi, a^*_y (\cN_z + c) a_x  \xi \rangle
	=  - \langle \xi, a^*(g) (\cN_z + c) a(g) \xi \rangle
	\le 0 \;,
\end{equation}
so we can drop this contribution. We estimate the remaining contribution by the Cauchy--Schwarz inequality and $ \norm{a(g)}_{\op} \le \norm{g}_2 $:
\begin{equation}
\begin{aligned}
	&\left\vert \int \d x \d y \; \overline{g(x)} g(y)
	\langle \xi, a^*_y (\cW_z(x) + \cW_z(y)) a_x \xi \rangle \right\vert \\
	&= \vert \langle \xi, \big( a^*(\cW_z g) a(g) + a^*(g) a(\cW_z g) \big) \xi \rangle \vert
	\le 2 \norm{\cW_z g}_2 \norm{g}_2 \norm{\xi}^2
	\le 2 \norm{g}_2^2 \norm{\xi}^2 \;.
\end{aligned}
\end{equation}
This yields the first bound in~\eqref{eq:pullthroughbound}. The second bound follows immediately with the CAR and $ \norm{a(g)}_{\op} \le \norm{g}_2 $
\begin{equation}
\begin{aligned}
	&\norm{(\cN_z + c)^{\frac 12} a(g) \xi}^2
	= \int \d x \; \cW_z(x) \norm{a_x a(g) \xi}^2 + c \norm{a(g) \xi}^2 \\
	&\le \int \d x \; \cW_z(x) \norm{g}_2^2 \norm{a_x \xi}^2 + c \norm{g}_2^2 \norm{\xi}^2
	= \norm{g}_2^2 \norm{(\cN_z + c)^{\frac 12} \xi}^2.
\end{aligned}
\end{equation}
\end{proof}

The pull--through bounds allow us to prove the following estimates.

\begin{lemma}[Bounding quartic operators by $\cN_z$] \label{lem:aastar_conversion_multi}
Let $ \xi \in \cF $ and consider two compact operators $ A$ and $B $ on $ L^2(\TTT_L) $. Suppose that $ A \cW_z^{-\frac 12} $ and $ B \cW_z^{-\frac 12} $ are Hilbert--Schmidt. Then, we have the following estimates, where $ \sharp_1, \sharp_2 \in \{ *, \cdot\} $:
\begin{equation} \label{eq:aastar_conversion_multi}
\begin{aligned}
	\Big\Vert \int \d x \d x' \d y \d y' A(x,x') B(y,y') a_x^{\sharp_1} a_{x'} a_y^{\sharp_2} a_{y'} \xi \Big\Vert 
	&\le 4 \Vert A \cW_z^{-\frac 12} \Vert_{\HS}
		\Vert B \cW_z^{-\frac 12} \Vert_{\HS}
		\Vert (\cN_z + 1) \xi \Vert \;, \\
	\Big\Vert \int \d x \d x' \d y \d y' A(x,x') B(y,y') a_x^* a_{x'}^* a_y^{\sharp_2} a_{y'} \xi \Big\Vert 
	&\le 8 \Vert A \cW_z^{-\frac 12} \Vert_{\HS}
		\Vert B \cW_z^{-\frac 12} \Vert_{\HS}
		\Vert (\cN_z + 1) \xi \Vert \;, \\
	\Big\Vert \int \d x \d x' \d y \d y' A(x,x') B(y,y') a_x^{\sharp_1} a_{x'} a_y^* a_{y'}^* \xi \Big\Vert 
	&\le 8 \Vert A \cW_z^{-\frac 12} \Vert_{\HS}
		\Vert B \cW_z^{-\frac 12} \Vert_{\HS}
		\Vert (\cN_z + 1) \xi \Vert \;, \\
	\Big\Vert \int \d x \d x' \d y \d y' A(x,x') B(y,y') a_x^* a_{x'}^* a_y^* a_{y'}^* \xi \Big\Vert 
	&\le 16 \Vert A \cW_z^{-\frac 12} \Vert_{\HS}
		\Vert B \cW_z^{-\frac 12} \Vert_{\HS}
		\Vert (\cN_z + 1) \xi \Vert \;.
\end{aligned}
\end{equation}
\end{lemma}

\begin{proof}
We start with the second bound in~\eqref{eq:aastar_conversion_multi} for $ a^{\sharp_2} = a^* $, since the case $ a^{\sharp_2} = a $ and the first bound in~\eqref{eq:aastar_conversion_multi} are completely analogous. With Lemma~\ref{lem:aastar_conversion},
\begin{equation} \label{eq:multiconversion_firstbound}
\begin{aligned}
	&\norm{ \int \d x \d x' \d y \d y' A(x,x') B(y,y') a_x^* a_{x'}^* a_y^* a_{y'} \xi } \\
	&\le 2 \Vert A \cW_z^{-\frac 12} \Vert_{\HS} \norm{ \int \d y \d y' B(y,y') (\cN_z + 1)^{\frac 12} a_y^* a_{y'} \xi } \;.
\end{aligned}
\end{equation}
To bound the second factor in this estimate, we use~\eqref{eq:pullthroughbound}:
\begin{equation}
\begin{aligned}
	&\norm{ \int \d y \d y' B(y,y') (\cN_z + 1)^{\frac 12} a_y^* a_{y'} \xi }
	\le \int \d y' \norm{ (\cN_z + 1)^{\frac 12} a^*(B(\cdot,y')) a_{y'} \xi } \\
	&\le \int \d y' \norm{B(\cdot,y')}_2 \norm{ (\cN_z + 4)^{\frac 12} a_{y'} \xi } \;.
\end{aligned}
\end{equation}
By $\cW_z(y') \geq 0$ and~\eqref{eq:pullthrough}, we obtain
\begin{equation}
	\norm{ (\cN_z + 4)^{\frac 12} a_{y'} \xi }^2
	\le \langle a_{y'} \xi,(\cN_z + 4 + \cW_z(y')) a_{y'} \xi \rangle
	= \norm{ a_{y'} (\cN_z + 4 )^{\frac 12} \xi }^2 \;,
\end{equation}
and therefore, with the Cauchy--Schwarz inequality as in~\eqref{eq:aabound_in_proof},
\begin{equation} \label{eq:astara_bound_pullthrough}
\begin{aligned}
	&\norm{ \int \d y \d y' B(y,y') (\cN_z + 1)^{\frac 12} a_y^* a_{y'} \xi }
	\le \int \d y' \norm{B(\cdot,y')} \norm{ a_{y'} (\cN_z + 4)^{\frac 12} \xi } \\
	&\le \Vert B \cW_z^{-\frac 12} \Vert_{\HS}
		\norm{\cN_z^{\frac 12} (\cN_z + 4 )^{\frac 12} \xi} \;.
\end{aligned}
\end{equation}
By $ \cN_z + 4 \le 4(\cN_z + 1) $ we obtain the second bound in~\eqref{eq:aastar_conversion_multi}.

\medskip

We next prove the fourth bound in~\eqref{eq:aastar_conversion_multi}.
The third one is analogous. Starting with Lemma~\ref{lem:aastar_conversion} again,
\begin{equation} \label{eq:multiconversion_secondbound}
\begin{aligned}
	&\norm{ \int \d x \d x' \d y \d y' A(x,x') B(y,y') a_x^* a_{x'}^* a_y^* a_{y'}^* \xi } \\
	&\le 2 \Vert A \cW_z^{-\frac 12} \Vert_{\HS} \norm{ \int \d y \d y' B(y,y') (\cN_z + 1)^{\frac 12} a_y^* a_{y'}^* \xi } \;.
\end{aligned}
\end{equation}
We then apply the pull--through formula in Lemma~\ref{lem:pullthroughbound}, using the notation
\begin{equation} \label{eq:psiX_notation}
	\xi_{x,x'} := (\cN_z + 1 + \cW_z(x) + \cW_z(x'))^{\frac 12} \xi \;, \qquad
	\xi_{y,y'} := (\cN_z + 1 + \cW_z(y) + \cW_z(y'))^{\frac 12} \xi \;,
\end{equation}
and normal--order as in~\eqref{eq:aastarconversion_commutation}
\begin{equation}
\begin{aligned}
	&\norm{ \int \d y \d y' B(y,y') (\cN_z + 1)^{\frac 12} a_y^* a_{y'}^* \xi }^2
	= \int \d x \d x' \d y \d y' \overline{B(x,x')} B(y,y')
		\langle \xi_{x,x'}, a_{x'} a_x a_y^* a_{y'}^* \xi_{y,y'} \rangle  \\
	&= \I_4 - \I_{2,1} - \I_{2,2} + \I_{2,3} + \I_{2,4} + \I_{0,1} - \I_{0,2} \;,
\end{aligned}
\end{equation}
where we defined
\begin{equation}
\begin{aligned}
	\I_4
	&:= \int \d x \d x' \d y \d y' \overline{B(x,x')} B(y,y')
		\langle \xi_{x,x'}, a_y^* a_{y'}^* a_{x'} a_x \xi_{y,y'} \rangle \;, \\
	\I_{2,1}
	&:= \int \d x \d x' \d y \d y' \overline{B(x,x')} B(y,x')
		\delta_{x',y'} \langle \xi_{x,x'}, a_y^* a_x \xi_{y,y'} \rangle \;, \\
	\I_{2,2}
	&:= \int \d x \d x' \d y \d y' \overline{B(x,x')} B(x,y')
		\delta_{x,y} \langle \xi_{x,x'}, a_{y'}^* a_{x'} \xi_{y,y'} \rangle \;, \\
	\I_{2,3}
	&:= \int \d x \d x' \d y \d y' \overline{B(x,x')} B(x',y')
		\delta_{x',y} \langle \xi_{x,x'}, a_{y'}^* a_x \xi_{y,y'} \rangle \;, \\
	\I_{2,4}
	&:= \int \d x \d x' \d y \d y' \overline{B(x,x')} B(y,x)
		\delta_{x,y'} \langle \xi_{x,x'}, a_y^* a_{x'} \xi_{y,y'} \rangle \;, \\
	\I_{0,1}
	&:= \int \d x \d x' \d y \d y' \overline{B(x,x')} B(x,x')
		\delta_{x,y} \delta_{x',y'} \langle \xi_{x,x'}, \xi_{y,y'} \rangle \;, \\
	\I_{0,2}
	&:= \int \d x \d x' \d y \d y' \overline{B(x,x')} B(x',x)
		\delta_{x,y'} \delta_{x',y} \langle \xi_{x,x'}, \xi_{y,y'} \rangle \;.
\end{aligned}
\end{equation}
To bound $\I_4$, observe that by the pull--through formula~\eqref{eq:pullthrough},
\begin{equation}
\begin{split}
	\I_4
	& = \int \d x \d x' \d y \d y' \overline{B(x,x')} B(y,y') \\
	& \quad \times
		\langle \xi, a_y^* a_{y'}^* \big(\cN_z + 1 + \cW_z(x) + \cW_z(x') + \cW_z(y) + \cW_z(y') \big) a_{x'} a_x \xi \rangle \;.
\end{split}
\end{equation}
The contribution from $ (\cN_z + 1) $ is bounded analogously to~\eqref{eq:astara_bound_pullthrough}:
\begin{equation}
\begin{aligned}
	&\int \d x \d x' \d y \d y' \overline{B(x,x')} B(y,y')
		\langle \xi, a_y^* a_{y'}^* (\cN_z + 1) a_{x'} a_x \xi \rangle \\
	&= \norm{ \int \d x \d x' \overline{B(x,x')} (\cN_z + 1)^{\frac 12} a_{x'} a_x \xi }^2
	\le \Vert B \cW_z^{-\frac 12} \Vert_{\HS}^2
		\norm{(\cN_z + 4 )\xi}^2 \;.
\end{aligned}
\end{equation}
The contribution from $ \cW_z(x) $ is bounded using Lemma~\ref{lem:aastar_conversion}, and then $ \norm{\cW_z}_{\op} \le 1 $ and $ \cN_z^{\frac 12} \le (\cN_z + 1) $:
\begin{equation}
\begin{aligned}
	&\left\vert \int \d x \d x' \d y \d y' \overline{B(x,x')} B(y,y')
		\langle \xi, a_y^* a_{y'}^* \cW_z(x) a_{x'} a_x \xi \rangle \right\vert \\
	&\le \norm{ \int \d x \d x' (\cW_z \overline{B})(x,x') a_x a_{x'} \xi }
		\norm{ \int \d y \d y' \overline{B}(y,y') a_y a_{y'} \xi } \\
	&\le \Vert \cW_z \overline{B} \cW_z^{-\frac 12} \Vert_{\HS}
		\Vert \overline{B} \cW_z^{-\frac 12} \Vert_{\HS}
		\Vert \cN_z^{\frac 12} \xi \Vert^2
	\le \Vert B \cW_z^{-\frac 12} \Vert_{\HS}^2
		\norm{(\cN_z + 1)\xi}^2 \;.
\end{aligned}
\end{equation}
The contributions from $ \cW_z(x') $, $ \cW_z(y) $ and $ \cW_z(y') $ are bounded analogously, and we conclude
\begin{equation} \label{eq:I4_bound}
	\I_4
	\le 20 \Vert B \cW_z^{-\frac 12} \Vert_{\HS}^2
		\norm{(\cN_z + 1)\xi}^2 \;.
\end{equation}
The bounds on the other contributions are proven by the same strategy with final result
\begin{equation}
	|\I_{2,j}|
	\le 6 \Vert B \Vert_{\HS}^2 \norm{ (\cN_z + 1) \xi}^2 \;, \qquad
	|\I_{0,j}|
	\le 3 \norm{B}_{\HS}^2 \norm{(\cN_z + 1) \xi}^2 \;.
\end{equation}
Combining the bounds for $ \I_4 $, $ \I_{2,j} $ and $ \I_{0,j} $, and using $ \Vert \cW_z^{\frac 12} \Vert_{\op} \le 1 $, we arrive at
\begin{equation}
	\norm{ \int \d y \d y' B(y,y') (\cN_z + 1)^{\frac 12} a_y^* a_{y'}^* \xi }^2
	\le 50 \Vert B \cW_z^{-\frac 12} \Vert_{\HS}^2
		\norm{(\cN_z + 1) \xi}^2 \;,
\end{equation}
and using this in~\eqref{eq:multiconversion_secondbound} with $ \sqrt{50} \le 8 $ yields the desired fourth bound in~\eqref{eq:aastar_conversion_multi}.

The third bound is analogous, where applying Lemma~\ref{lem:aastar_conversion} as in~\eqref{eq:multiconversion_secondbound} gives the same expression on the right-hand side without the factor of 2.
\end{proof}

\subsection{Hilbert--Schmidt Norm Estimates}
\label{subsec:HSest}

Next, we prove some preliminary estimates for norms which arise in the above bounds by number operators. Let us start with the following commutator bound.

\begin{lemma}[Commutator Bounds] \label{lem:ualpha_derivative_bound}
Recall $ \cW_z $ from~\eqref{eq:cWz} and $ u^\alpha, v^\alpha $ from~\eqref{eq:uv}. Then, for any $ \alpha \in \{1,\ldots,M\} $ and $ a \in \NNN $, there exists $ C_a > 0 $ such that
\begin{equation} \label{eq:ualpha_derivative_bound}
	\big\Vert [\cW_z^{-\frac a2}, u^\alpha] \cW_z^{\frac a2} \big\Vert_{\op}
	\le C_a \fl^{-2a} \;.
\end{equation}
Moreover, if $ u^{(b),\alpha} $ is an operator multiplying in momentum space by $ p_b \hat{u}^\alpha $, where $ p_b $ is any polynomial in $ \partial_{k_j}$, $\bar\partial_{k_j} $ of degree $ b $, then for any $ a,b \in \NNN $, there exists $ C_{a,b} > 0 $ such that
\begin{equation} \label{eq:ualpha_derivative_bound_bis}
	\big\Vert [\cW_z^{-\frac a2}, u^{(b),\alpha}] \cW_z^{\frac a2} \big\Vert_{\op}
	\le C_{a,b} \fl^{-2a-b} \;.
\end{equation}
All bounds remain true if the operator $ u^\alpha $ is replaced by the operators $ \sum_{\alpha=1}^M u^\alpha $ or by $ v^\alpha $, $ (u^\alpha)^2 $, $ (v^\alpha)^2 $, or sums thereof.
\end{lemma}

\begin{proof}
First, we remove the position shift $ z $ within $ \cW_z $ as follows: $ u^\alpha $ commutes with the unitary translation operator \eqref{eq:cTy} since in momentum space both are multiplication operators. With $ \cW_z^b = \cT_z \cW^b \cT_{-z} $ for all $ b \in \RRR $, we then have
\begin{equation} \label{eq:remove_z_from_cW}
	\big\Vert [\cW_z^{-\frac a2}, u^\alpha] \cW_z^{\frac a2} \big\Vert_{\op}
	= \big\Vert \cT_z [\cW^{-\frac a2}, u^\alpha] \cT_{-z} \cT_z \cW^{\frac a2} \cT_{-z} \big\Vert_{\op}
	= \big\Vert [\cW^{-\frac a2}, u^\alpha] \cW^{\frac a2} \big\Vert_{\op} \;.
\end{equation}
Next, observe that $ \cW^{-\frac a2} $ acts in momentum space as $ (1-\Delta_k)^a $, where we recall that $ -\Delta_k = - \sum_{j=1}^3 \bar\partial_{k_j} \partial_{k_j} $ with $ \bar\partial_{k_j}$, $\partial_{k_j}  $ being the discrete derivatives~\eqref{eq:partial_k}. Given a function $\hat{g}$ in momentum space, we introduce the unitary momentum space translation operator
\begin{equation} \label{eq:cThat}
	(\hat{\cT}_p \hat{g})(k)
	:= \hat{g}(k-p) \;,
\end{equation}
with $ [\hat{\cT}_p, \partial_{k_j}] = 0 $ and $[\hat{\cT}_p, \bar\partial_{k_j}] = 0$. Then the product rule~\eqref{eq:productrule} implies
\begin{equation}
	[\partial_{k_j},\hat{g}]
	= (\partial_{k_j} \hat{g}) \hat{\cT}_{-2 \pi L^{-1} e_j} \;, \qquad
	[\bar\partial_{k_j},\hat{g}]
	= (\bar\partial_{k_j} \hat{g}) \hat{\cT}_{2 \pi L^{-1} e_j} \;,
\end{equation}
which are operator identities in momentum space. By the triangle inequality, we can bound $ \Vert [\cW^{-\frac a2}, u^\alpha] \cW^{\frac a2} \Vert_{\op} = \Vert [(1-\Delta_k)^a, \hat{u}^\alpha] (1-\Delta_k)^{-a} \Vert_{\op} $ by a finite ($ a $-dependent) sum of terms of the form
\begin{equation*}
	\big\Vert (\partial_{k_{j_1}} \ldots \partial_{k_{j_r}}
		\bar\partial_{k_{j_{r+1}}} \ldots \bar\partial_{k_{j_s}} \hat{u}^\alpha)\hat{\cT}_p
		\partial_{k_{j_{s+1}}} \ldots \partial_{k_{j_t}}
		\bar\partial_{k_{j_{t+1}}} \ldots \bar\partial_{k_{j_m}}  (1 - \Delta_k)^{-a} \big\Vert_{\op} \;,
\end{equation*}
with $ 0 \le r \le s \le t \le m \le 2a $ and $ p := 2 \pi L^{-1} \left(- \sum_{\ell = 1}^r e_{j_\ell} + \sum_{\ell = r+1}^s e_{j_\ell} \right) $. We estimate this by
\begin{equation} \label{eq:multi_derivatives_bound}
\begin{aligned}
	&\big\Vert (\partial_{k_{j_1}} \ldots \partial_{k_{j_r}}
		\bar\partial_{k_{j_{r+1}}} \ldots \bar\partial_{k_{j_s}} \hat{u}^\alpha)\hat{\cT}_p
		\partial_{k_{j_{s+1}}} \ldots \partial_{k_{j_t}}
		\bar\partial_{k_{j_{t+1}}} \ldots \bar\partial_{k_{j_m}}  (1 - \Delta_k)^{-a} \big\Vert_{\op} \\
	&\le \big\Vert \partial_{k_{j_1}} \ldots \partial_{k_{j_r}}
		\bar\partial_{k_{j_{r+1}}} \ldots \bar\partial_{k_{j_s}} \hat{u}^\alpha \big\Vert_\infty \lVert \hat{\cT}_p \rVert_{\op}
		\lVert \partial_{k_{j_{s+1}}} \ldots \partial_{k_{j_t}}
		\bar\partial_{k_{j_{t+1}}} \ldots \bar\partial_{k_{j_m}} (1 - \Delta_k)^{-a} \rVert_{\op}
		 \\
	&\le C \fl^{-s} \;,
\end{aligned}
\end{equation}
where we used~\eqref{eq:smooth}, and that the operator $ \partial_{k_{j_{s+1}}} \ldots \partial_{k_{j_t}} \bar\partial_{k_{j_{t+1}}} \ldots \bar\partial_{k_{j_m}} (1 - \Delta_k)^{-a} $ acts in position space as a multiplication by a function whose norm is uniformly bounded by $1$ due to $ m \le 2a $. Since $ s \le 2a $ (and recalling that the number of summands is $ a $-dependent but finite) this establishes~\eqref{eq:ualpha_derivative_bound}.

The other claimed estimates are proven analogously.
\end{proof}

With these commutator bounds at hand, we may prove the following lemma.

\begin{lemma}[Removing $ v^\alpha $ or $ u^\alpha $ from Hilbert--Schmidt norms] \label{lem:HS_uv_removal}
Let $ (A^\alpha)_{\alpha = 1}^M $ be a family of Hilbert--Schmidt operators on $ L^2(\TTT_L) $, and recall $ u^\alpha$, $v^\alpha $ from~\eqref{eq:uv}. Then, we have the bounds
\begin{equation} \label{eq:HS_uv_removal}
	\Bigg\Vert \sum_{\alpha = 1}^M u^\alpha A^\alpha \Bigg\Vert_{\HS}^2
	\le \sum_{\alpha = 1}^M \Vert A^\alpha \Vert_{\HS}^2 \;.
\end{equation}
Moreover, recalling $ \cW_z $ from~\eqref{eq:cWz}, for any $ a \in \mathbb{N} $, there exists $ C_a > 0 $ such that
\begin{equation} \label{eq:HS_uv_removal_with_W}
\begin{aligned}
	\Bigg\Vert \sum_{\alpha = 1}^M \cW_z^{-\frac a2} u^\alpha A^\alpha \Bigg\Vert_{\HS}^2
	&\le C_a \fl^{-4a} \sum_{\alpha = 1}^M \big\Vert \cW_z^{-\frac a2} A^\alpha \big\Vert_{\HS}^2 \;, \\
	\big\Vert \cW_z^{-\frac a2} u^\alpha A^\alpha \big\Vert_{\HS}
	&\le C_a \fl^{-2a} \big\Vert \cW_z^{-\frac a2} A^\alpha \big\Vert_{\HS} \;.
\end{aligned}
\end{equation}
All bounds also hold true with $ u^\alpha $ replaced by $ v^\alpha $.
\end{lemma}

\begin{proof}
If $ \alpha \neq \alpha' $, then $ \hat{u}^\alpha $ and $ \hat{u}^{\alpha'} $ have disjoint supports and thus  $ u^\alpha u^{\alpha'} = 0 $, so
\begin{equation}
	\Bigg\Vert \sum_{\alpha = 1}^M u^\alpha A^\alpha \Bigg\Vert_{\HS}^2
	= \sum_{\alpha = 1}^M \tr \big( (A^\alpha)^* (u^\alpha)^2 A^\alpha \big)
	\le \sum_{\alpha = 1}^M \tr \big( (A^\alpha)^* A^\alpha \big) 
	= \sum_{\alpha = 1}^M \Vert A^\alpha \Vert_{\HS}^2 \;,
\end{equation}
where we used that $ 0 \le \hat{u}^\alpha(k) \le 1 $ implies $ 0 \le (u^\alpha)^2 \le 1 $. The proof is analogous in the case where $ u^\alpha $ is replaced by $ v^\alpha $.\\
To prove~\eqref{eq:HS_uv_removal_with_W}, we use the triangle inequality
\begin{equation} \label{eq:norm_cWuA_split}
	\Bigg\Vert \sum_{\alpha = 1}^M \cW_z^{-\frac a2} u^\alpha A^\alpha \Bigg\Vert_{\HS}^2
	\le 2 \Bigg\Vert \sum_{\alpha = 1}^M u^\alpha \cW_z^{-\frac a2} A^\alpha \Bigg\Vert_{\HS}^2
		+ 2 \Bigg\Vert \sum_{\alpha = 1}^M [\cW_z^{-\frac a2}, u^\alpha] A^\alpha \Bigg\Vert_{\HS}^2 \;.
\end{equation}
The first norm is bounded by~\eqref{eq:HS_uv_removal}. To bound the second norm, observe that discrete derivatives enlarge the momentum space support by at most $ 2 \pi L^{-1} \ll 1 $, so we conclude for $ \alpha \neq \alpha' $ and $ L $ large enough that $ [\cW_z^{-\frac a2}, u^\alpha]^* [\cW_z^{-\frac a2}, u^{\alpha'}] = 0 $, and
\begin{equation} \label{eq:commutator_disjoint_support_bound}
\begin{aligned}
	&\Bigg\Vert \sum_{\alpha = 1}^M [\cW_z^{-\frac a2}, u^\alpha] A^\alpha \Bigg\Vert_{\HS}^2
	= \sum_{\alpha = 1}^M \tr \big( [\cW_z^{-\frac a2}, u^\alpha]^* [\cW_z^{-\frac a2}, u^\alpha] A^\alpha (A^\alpha)^* \big) \\
	&\le \sum_{\alpha = 1}^M
		\big\Vert [\cW_z^{-\frac a2}, u^\alpha] \cW_z^{\frac a2} \big\Vert_{\op}^2
		\Vert \cW_z^{-\frac a2} A^\alpha \Vert_{\HS}^2
	\le C_a \fl^{-4a}
		\sum_{\alpha = 1}^M \big\Vert \cW_z^{-\frac a2} A^\alpha \big\Vert_{\HS}^2 \;,
\end{aligned}
\end{equation}
where we used Lemma~\ref{lem:ualpha_derivative_bound}.
This yields the desired bound on the second norm in~\eqref{eq:norm_cWuA_split}. The bound with $ v^\alpha $ in place of $ u^\alpha $ is analogous, and the second bound in~\eqref{eq:HS_uv_removal_with_W} follows by removing the sum over $ \alpha $ in all proof steps.
\end{proof}

\begin{lemma}[Localized Hilbert--Schmidt norms] \label{lem:HS_estimate}
Let $ a,b \in \mathbb{N}_0 $, and consider a family of functions $ (\hat\varphi_\alpha)_{\alpha =1}^M \subset W^{2a+2b+2,1}(\frac{2\pi}{L} \mathbb{Z}^3) $ with corresponding position space multiplication operators $ (\varphi_{\alpha,z} f)(x) := \varphi_\alpha(x-z) f(x) $. (The discrete Sobolev space $W^{n,1}(\frac{2\pi}{L} \mathbb{Z}^3)$ is defined as functions such that up to $n$ discrete derivatives remain summable.) Recall the definitions of $ u^\alpha $ and $ v^\alpha $ from~\eqref{eq:uv} and of $ \cW_z $ from~\eqref{eq:cWz}. Then, there exists $ C_{a,b} > 0 $ such that
\begin{equation} \label{eq:HS_estimate}
\begin{aligned}
	\Bigg\Vert \sum_{\alpha = 1}^M \cW_z^{-\frac a2} v^\alpha \varphi_{\alpha,z} u^\alpha \cW_z^{-\frac b2} \Bigg\Vert_{\HS}
	&\le C_{a,b} \fl^{-2a-2b} \rho^{\frac 13} \max_{1 \le \alpha \le M} \Vert (1 - \Delta_k)^{a+b+1} \hat\varphi_\alpha \Vert_1 \;, \\
	\big\Vert \cW_z^{-\frac a2} v^\alpha \varphi_{\alpha,z} u^\alpha \cW_z^{-\frac b2} \big\Vert_{\HS}
	&\le C_{a,b} \fl^{-2a-2b} \rho^{\frac 13} M^{-\frac 12} \max_{1 \le \alpha \le M} \Vert (1 - \Delta_k)^{a+b+1} \hat\varphi_\alpha \Vert_1 \;, \\
\end{aligned}
\end{equation}
where $ \Vert \cdot \Vert_1 $ is the $ \ell^1 $--norm as in~\eqref{eq:ell_norms} and $ -\Delta_k $ is the discrete Laplacian~\eqref{eq:Delta_k}. All bounds hold true after swapping $ u^\alpha $ and $ v^\alpha $.
\end{lemma}

\begin{proof}
By Lemma~\ref{lem:HS_uv_removal}, we may remove the $ u^\alpha $ as
\begin{equation} \label{eq:HS_bound_commutator_2}
\begin{aligned}
	&\Bigg\Vert \sum_{\alpha = 1}^M \cW_z^{-\frac a2} v^\alpha \varphi_{\alpha,z} u^\alpha \cW_z^{-\frac b2} \Bigg\Vert_{\HS}^2
	\le C_b \fl^{-4b} \sum_{\alpha = 1}^M \big\Vert \cW_z^{-\frac a2} v^\alpha \varphi_{\alpha,z} \cW_z^{-\frac b2} \big\Vert_{\HS}^2 \\
	&\le C_b \fl^{-4b} \sum_{\alpha = 1}^M \big\Vert v^\alpha \cW_z^{-\frac a2} \varphi_{\alpha,z} \cW_z^{-\frac b2} \big\Vert_{\HS}^2
		+ C_b \fl^{-4b} \sum_{\alpha = 1}^M \big\Vert  [\cW_z^{-\frac a2}, v^\alpha] \varphi_{\alpha,z} \cW_z^{-\frac b2} \big\Vert_{\HS}^2 \;.
\end{aligned}
\end{equation}
For the first term, we use $ [\cW_z, \varphi_{\alpha,z}] = 0 $ and recall that $ \cW^{-\frac 12} $ acts in momentum space as $ (1 - \Delta_k) $:
\begin{equation} \label{eq:HS_bound_firstterm}
\begin{aligned}
	&\sum_{\alpha = 1}^M \big\Vert v^\alpha \cW_z^{-\frac a2} \varphi_{\alpha,z} \cW_z^{-\frac b2} \big\Vert_{\HS}^2
	\le \sum_{\alpha = 1}^M \big\Vert v^\alpha \cW_z^{\frac 12} \big\Vert_{\HS}^2
		\big\Vert \cW_z^{-\frac{a+b+1}{2}} \varphi_{\alpha,z} \big\Vert_{\op}^2 \\
	& \le \sum_{\alpha = 1}^M
		\big\Vert v^\alpha \cW_z^{\frac 12} \big\Vert_{\HS}^2
		\big\Vert \cW^{-\frac{a+b+1}{2}} \varphi_\alpha \big\Vert_\infty^2 \\
	&\le M \max_{1 \le \alpha \le M} 
		\Vert \hat{v}^\alpha \Vert_2^2
		\Vert \cW_z^{\frac 12} \Vert_2^2
		\Vert (1 - \Delta_k)^{a+b+1} \hat\varphi_\alpha \Vert_1^2
	\le C \rho^{\frac 23} \max_{1 \le \alpha \le M} \Vert (1 - \Delta_k)^{a+b+1} \hat\varphi_\alpha \Vert_1^2 \;,
\end{aligned}
\end{equation}
where we used that the $ L^\infty $-norm is invariant under a shift by $ z $, and we subsequently bounded the $ L^\infty $ and the Hilbert--Schmidt norm by Lemma~\ref{lem:FT_bounds}, using that $ \Vert \hat{v}^\alpha \Vert_2^2 \le C\rho^{\frac 23} M^{-1} $, as the patch $ B_\alpha $ has volume of order $ C \rho^{\frac 23} M^{-1} $.
For the second term in~\eqref{eq:HS_bound_commutator_2}, we likewise extract and bound a factor of $ \big\Vert \cW_z^{-\frac{a+b+1}{2}} \varphi_{\alpha,z} \big\Vert_{\op}^2 $, yielding
\begin{equation} \label{eq:HS_bound_commutator_sum}
\begin{aligned}
	&\sum_{\alpha = 1}^M \big\Vert [\cW_z^{-\frac a2}, v^\alpha] \varphi_{\alpha,z} \cW_z^{-\frac b2} \big\Vert_{\HS}^2
	\le \sum_{\alpha = 1}^M \big\Vert [\cW_z^{-\frac a2}, v^\alpha] \cW_z^{\frac a2 + \frac 12} \big\Vert_{\HS}^2
		\Vert (1 - \Delta_k)^{a+b+1} \hat\varphi_\alpha \Vert_1^2 \;.
\end{aligned}
\end{equation}
To bound the Hilbert--Schmidt norm, we proceed as in~\eqref{eq:multi_derivatives_bound} and use the triangle inequality to write it as a finite, $ a $-dependent sum of terms
\begin{equation*}
	\big\Vert (\partial_{k_{j_1}} \ldots \partial_{k_{j_r}}
		\bar\partial_{k_{j_{r+1}}} \ldots \bar\partial_{k_{j_s}} \hat{v}^\alpha)
		\partial_{k_{j_{s+1}}} \ldots \partial_{k_{j_t}}
		\bar\partial_{k_{j_{t+1}}} \ldots \bar\partial_{k_{j_m}} (1 - \Delta_k)^{-a-1} \hat{\cT}_p \big\Vert_{\HS} \;,
\end{equation*}
with $ 0 \le r \le s \le t \le m \le 2a $, where $ \hat{\cT}_p $ is the momentum shift operator from~\eqref{eq:cThat} for a suitable $ p \in \Lambda_L $. Recalling that $ \hat{\cT}_p $ commutes with $ \partial_{k_j}, \bar\partial_{k_j}, -\Delta_k $ and that $ \Vert \hat{\cT}_p \Vert_{\op} \le 1 $ since $ \hat{\cT}_p $ is unitary:
\begin{equation}
\begin{aligned}
	&\Vert (\partial_{k_{j_1}} \ldots \partial_{k_{j_r}}
		\bar\partial_{k_{j_{r+1}}} \ldots \bar\partial_{k_{j_s}} \hat{v}^\alpha)
		\partial_{k_{j_{s+1}}} \ldots \partial_{k_{j_t}}
		\bar\partial_{k_{j_{t+1}}} \ldots \bar\partial_{k_{j_m}} (1 - \Delta_k)^{-a-1} \hat{\cT}_p \Vert_{\HS} \\
	&\le \Vert (\partial_{k_{j_1}} \ldots \partial_{k_{j_r}}
		\bar\partial_{k_{j_{r+1}}} \ldots \bar\partial_{k_{j_s}} \hat{v}^\alpha)
		(1 - \Delta_k)^{-1} \Vert_{\HS} 
		\Vert \partial_{k_{j_{s+1}}} \ldots \partial_{k_{j_t}}
		\bar\partial_{k_{j_{t+1}}} \ldots \bar\partial_{k_{j_m}} (1 - \Delta_k)^{-a} \Vert_{\op} \\
	&\le \Vert \partial_{k_{j_1}} \ldots \partial_{k_{j_r}}
		\bar\partial_{k_{j_{r+1}}} \ldots \bar\partial_{k_{j_s}} \hat{v}^\alpha \Vert_2
		\Vert \cW^{\frac 12} \Vert_2
	\le C_a \rho^{\frac 13} M^{-\frac 12} \fl^{-2a} \;,
\end{aligned}
\end{equation}
where we used that since $ m \le 2a $, $ \partial_{k_{j_{s+1}}} \ldots \partial_{k_{j_t}} \bar\partial_{k_{j_{t+1}}} \ldots \bar\partial_{k_{j_m}} (1 - \Delta_k)^{-a} $ is a position space multiplication operator by a function bounded by $1$. Also, we bounded the Hilbert--Schmidt norm by Lemma~\ref{lem:FT_bounds} and used that $ \Vert \cW^{\frac 12} \Vert_2 \le C $, as well as that the support of $ (\partial_{k_{j_1}} \ldots \partial_{k_{j_r}} \bar\partial_{k_{j_{r+1}}} \ldots \bar\partial_{k_{j_s}} \hat{v}^\alpha) $ has volume $ \le C_a \rho^{\frac 23} M^{-1} $, and the function $ \hat{v}^\alpha $ is according to \eqref{eq:uhat_vhat} regularized on a scale $ \fl $, so each derivative produces a factor of $ \fl^{-1} $. We conclude that
\begin{equation}
	\big\Vert [\cW_z^{-\frac a2}, v^\alpha] \cW_z^{\frac a2 + \frac 12} \big\Vert_{\HS}
	\le C_a \rho^{\frac 13} M^{-\frac 12} \fl^{-2a} \;.
\end{equation}
Plugging this into~\eqref{eq:HS_bound_commutator_sum} and then into~\eqref{eq:HS_bound_commutator_2}, yields the first claim in~\eqref{eq:HS_estimate}. The second claim follows analogously, where removing the sum over $ \alpha $, we gain a factor of $ M^{-\frac 12} $.
\end{proof}

\subsection{Proof of Proposition~\ref{prop:gronwall}}
\label{subsec:gronwall}

The strategy for proving the first bound in Proposition~\ref{prop:gronwall} is to apply Grönwall's lemma to $ f_{(1)}(\lambda) := \langle \xi, T_\lambda^* (\cN + L^3) T_\lambda \xi \rangle $, after bounding $ |\partial_\lambda f_{(1)}(\lambda)| = |\langle \xi_\lambda, [S, \cN] \xi_\lambda \rangle| $. To estimate this commutator, we will use the following lemma (stated in more general form for later convenience).
\begin{lemma}[Bound on Bogoliubov Exponents] \label{lem:Q2_bound}
Let $ X : \Lambda_L \to \RRR^{M \times M} $ with $ \supp(X_{\alpha, \beta}) \subset B_R(0) $ satisfy the following regularity estimate: There exists some $ D_X > 0 $, possibly depending on $ \rho $ but not on $ L $, such that
\begin{equation}~\label{eq:regularityestimate_1}
	\Vert (1 - \Delta_k)^2 \hat\varphi_{X,\alpha,\beta} \Vert_1
	\le D_X \qquad \forall \alpha, \beta \in \{1,\ldots,M\} \;,
\end{equation}
with $ \hat\varphi_{X,\alpha,\beta} $ defined in~\eqref{eq:varphi}. Recall $ b_{\alpha,x} = a(v_x^{\bar\alpha}) a(u_x^\alpha) $. Then, there exists $ C > 0 $ such that
\begin{equation} \label{eq:Q2_bound}
	\bigg\vert \frac{1}{\rho^{\frac 23}} \sum_{\alpha, \beta = 1}^M \int \d x \d y \varphi_X(y-x)_{\alpha, \beta}
		\langle \xi, b_{\alpha,x} b_{\beta,y} \xi \rangle \bigg\vert
	\le C D_X \fl^{-2} \langle \xi, (\cN + L^3) \xi \rangle \;.
\end{equation}
\end{lemma}

\begin{proof}
First, note that since $ X(k)_{\alpha,\beta} $ is real, we have $ \varphi_X(x)_{\alpha, \beta} = \overline{\varphi_X(-x)_{\alpha, \beta}} $. With Cauchy--Schwarz,
\begin{equation}
\label{eq:dtflambda}
\begin{aligned}
	&\Bigg\vert \frac{1}{\rho^{\frac 23}} \sum_{\alpha, \beta = 1}^M \int \d x \d y \; \varphi_X(y-x)_{\alpha, \beta}
		\langle \xi, b_{\alpha,x} b_{\beta,y} \xi \rangle \Bigg\vert \\
	&\le \frac{1}{\rho^{\frac 23}}
		\Bigg( \sum_{\beta = 1}^M \int \d y  \Bigg\Vert \sum_{\alpha = 1}^M \int \d x \; \varphi_X(x-y)_{\alpha, \beta} b^*_{\alpha,x} \xi \Bigg\Vert^2 \Bigg)^{\frac 12}
		\Bigg( \sum_{\beta = 1}^M \int \d y \Vert b_{\beta,y} \xi \Vert^2 \Bigg)^{\frac 12} \;.
\end{aligned}
\end{equation}
The first norm is bounded by Corollary~\ref{cor:bbstar_conversion}, Lemma~\ref{lem:HS_estimate}, and the regularity estimate~\eqref{eq:regularityestimate_1}
\begin{equation}
\label{eq:trNzreduction}
\begin{aligned}
	&\Bigg\Vert \sum_{\alpha = 1}^M \int \d x \; \varphi_X(x-y)_{\alpha, \beta} b^*_{\alpha,x} \xi \Bigg\Vert
	\le 2 \Bigg\Vert \sum_{\alpha = 1}^M u^\alpha \varphi_{X, \alpha, \beta, y} v^\alpha \cW_y^{-\frac 12} \Bigg\Vert_{\HS}
		\Vert (\cN_y + 1)^{\frac 12} \xi \Vert \\
	&\le C \fl^{-2} \rho^{\frac 13} \max_{1 \le \alpha \le M} \Vert (1 - \Delta_k)^2 \hat\varphi_{X,\alpha,\beta} \Vert_1 \Vert (\cN_y + 1)^{\frac 12} \xi \Vert
	\le C D_X \fl^{-2} \rho^{\frac 13} \Vert (\cN_y + 1)^{\frac 12} \xi \Vert \;.
\end{aligned}
\end{equation}
Thus,
\begin{equation}
	\sum_{\beta = 1}^M \int \d y  \Bigg\Vert \sum_{\alpha = 1}^M \int \d x \; \varphi_X(x-y)_{\alpha, \beta} b^*_{\alpha,x} \xi \Bigg\Vert^2
	\le C D_X^2 \fl^{-4} \rho^{\frac 23} M \langle \xi, (\cN + L^3) \xi \rangle \;.
\end{equation}
For the second norm, we use $ b_{\beta,y} = a(v^{\bar\beta}_y) a(u^\beta_y) $ with $ \Vert a(v^{\bar\beta}_y) \Vert_{\op} \le \Vert \hat v^{\bar\beta} \Vert_2 \le C \rho^{\frac 13} M^{-\frac 12} $:
\begin{equation}
\label{eq:trNzreduction_2}
\begin{aligned}
	&\sum_{\beta = 1}^M \int \d y \Vert b_{\beta,y} \xi \Vert^2
	\le C \rho^{\frac 23} M^{-1} \sum_{\beta = 1}^M \int \d y \Vert a(u^\beta_y) \xi \Vert^2 
	= C \rho^{\frac 23} M^{-1} \sum_{\beta = 1}^M \langle \xi, \d \Gamma\big( (u^\beta)^2 \big) \xi \rangle \\
	&\le C \rho^{\frac 23} M^{-1} \Bigg\Vert \sum_{\beta = 1}^M (u^\beta)^2 \Bigg\Vert_{\op} \langle \xi, \cN \xi \rangle
	\le C \rho^{\frac 23} M^{-1} \langle \xi, \cN \xi \rangle \;,
\end{aligned}
\end{equation}
where we used $ \langle \xi, \d \Gamma(A) \xi \rangle \le \norm{A}_{\op} \langle \xi, \cN \xi \rangle $ and that the momentum space multiplication operator $ u^\beta $ satisfies $ 0 \le \sum_{\beta = 1}^M (u^\beta)^2 \le 1 $. Plugging into~\eqref{eq:dtflambda} yields the desired bound.
\end{proof}

\begin{proof}[Proof of Proposition~\ref{prop:gronwall}]
\underline{Expectation value of $ \cN $:}
We aim to use Grönwall's lemma; that is, defining
\[
   f_{(1)}(\lambda) := \langle \xi, T_\lambda^* (\cN + L^3) T_\lambda \xi \rangle
\]
we seek a bound of the form $ |\partial_\lambda f_{(1)}(\lambda)| \le C f_{(1)}(\lambda) $. We abbreviate $ \xi_\lambda := T_\lambda \xi $, where we recall~\eqref{eq:T} $ T_\lambda = e^{-\lambda S} $ with $ S = S_- - S_-^* $, see also~\eqref{eq:S_position}:
\begin{equation}
\begin{aligned}
	&|\partial_\lambda f_{(1)}(\lambda)|
	= |\langle \xi_\lambda, [S, \cN] \xi_\lambda \rangle|
	\le 8 |\langle \xi_\lambda, [S_-, \cN] \xi_\lambda \rangle|\\
	\le &\frac{C}{\rho^{\frac 23}} \left\vert \sum_{\alpha, \beta = 1}^M
		\int \d x \d y \; \varphi_K(y-x)_{\alpha, \bar\beta}
		\langle \xi_\lambda, b_{\alpha,x} b_{\beta,y} \xi_\lambda \rangle \right\vert \;.
\end{aligned}
\end{equation}
This is bounded by Lemma~\ref{lem:Q2_bound} with $ \xi = \xi_\lambda $, and where $ X(k)_{\alpha,\beta} = \tilde{n}_{\beta,k} \tilde{n}_{\bar\beta,k}^{-1} K(k)_{\alpha,\bar\beta} $ satisfies the regularity estimate~\eqref{eq:regularityestimate_1} with $ D_X = C \fl^{-4} |\log\fl| $ due to Proposition~\ref{prop:regularity}. Thus,
\begin{equation}
	|\partial_\lambda f_{(1)}(\lambda)|
	\le C \fl^{-6} |\log\fl| f_{(1)}(\lambda) \;.
\end{equation}
Applying Grönwall's Lemma yields the desired first inequality in \eqref{eq:gronwall}.

\medskip

\noindent \underline{Expectation value of $ \cN_z^2 $:} We compute the derivative of
\[
	f_{(2)}(\lambda) := \langle \xi_\lambda, (\cN_z + 1)^2 \xi_\lambda \rangle
\]
as
\begin{equation}
\label{eq:dtflambda2}
\begin{aligned}
	\vert \partial_\lambda f_{(2)}(\lambda) \rvert
	& \le 2 \lvert\langle \xi_\lambda, [S_-, (\cN_z + 1)^2] \xi_\lambda \rangle \rvert \\
	& = 2 \lvert \langle \xi_\lambda, \Big((\cN_z +1) [S_-, \cN_z] + [S_-, \cN_z] (\cN_z +1) \Big) \xi_\lambda \rangle \rvert \\
	&\le 2 \lVert (\cN_z +1) \xi_\lambda \rVert \Big(
		\lVert [S_-, \cN_z] \xi_\lambda \rVert
		+ \lVert [S_-, \cN_z]^* \xi_\lambda \rVert
		\Big)\;.
\end{aligned}
\end{equation}
With Lemma~\ref{lem:pullthroughbound}, we obtain the pull--through formula
\begin{equation}
\label{eq:pullthrough_Nz}
	[a(f), \cN_z]
	= \int \d x \overline{f(x)} [a_x, \cN_z]
	= \int \d x \overline{f(x)} \cW_z(x) a_x
	= a(\cW_z f) \;,
\end{equation}
which allows evaluating the commutator as
\begin{equation}
\label{eq:dtfzlambda2_split}
\begin{aligned}
	&\Vert [S_-, \cN_z] \xi_\lambda \Vert
	\le C (\I + \II + \III + \IV) \;, \\
	\I & :=\rho^{-\frac 23} \Bigg\Vert \sum_{\alpha, \beta = 1}^M
		\int \d x \d y \varphi_K(y-x)_{\alpha, \bar\beta}
		a(\cW_z v_x^{\bar\alpha}) a(u_x^\alpha) a(v_y^{\bar\beta}) a(u_y^\beta) \xi_\lambda  \Bigg\Vert \;, \\
	\II & := \rho^{-\frac 23} \Bigg\Vert \sum_{\alpha, \beta = 1}^M
		\int \d x \d y \varphi_K(y-x)_{\alpha, \bar\beta}
		a(v_x^{\bar\alpha}) a(\cW_z u_x^\alpha) a(v_y^{\bar\beta}) a(u_y^\beta) \xi_\lambda  \Bigg\Vert \;, \\
	\III & := \rho^{-\frac 23} \Bigg\Vert \sum_{\alpha, \beta = 1}^M
		\int \d x \d y \varphi_K(y-x)_{\alpha, \bar\beta}
		a(v_x^{\bar\alpha}) a(u_x^\alpha) a(\cW_z v_y^{\bar\beta}) a(u_y^\beta) \xi_\lambda  \Bigg\Vert \;, \\
	\IV & := \rho^{-\frac 23} \Bigg\Vert \sum_{\alpha, \beta = 1}^M
		\int \d x \d y \varphi_K(y-x)_{\alpha, \bar\beta}
		a(v_x^{\bar\alpha}) a(u_x^\alpha) a(v_y^{\bar\beta}) a(\cW_z u_y^\beta) \xi_\lambda  \Bigg\Vert \;.
\end{aligned}
\end{equation}
To bound $ \I $, we split $ \varphi_K(y-x) $ as follows: Recall the radial cutoff function $ \hat\chi_R \in C_c^\infty(\RRR^3) $ with $ \hat\chi_R(k) = 1 $ for $ |k| \le R $, introduced above~\eqref{eq:DB}. Since $ \hat\chi_R $ does not depend on $ \rho $, its equally denoted restriction to the momentum lattice satisfies the regularity estimate
\begin{equation} \label{eq:chi_R}
	\Vert (1 - \Delta_k)^{\frac m2} \hat\chi_R \Vert_1 
	\le C_m \;.
\end{equation}
Furthermore, since $ \supp(K_{\alpha, \beta}) \subset B_R(0) $, we have $ K(k)_{\alpha,\beta} = K(k)_{\alpha,\beta} \hat\chi_R(k) $, so with $ \chi_R $ being the periodization of the inverse Fourier transform of $ \hat\chi_R $, we may split
\begin{equation}
\label{eq:varphisplit}
\begin{aligned}
	\varphi_K(y-x)_{\alpha, \beta}
	= \overline{\varphi_K(x-y)_{\alpha, \beta}}
	&= \int \d z' \; \overline{\varphi_K(x-z')_{\alpha, \beta}}
		\; \overline{\chi_R(z'-y)} \;.
\end{aligned}
\end{equation}
Note that since $ \hat\chi_R $ is a symmetric, real function, so is $ \chi_R $. We then write as in the proof of Corollary~\ref{cor:bbstar_conversion}:
\begin{equation} \label{eq:Ibound}
\begin{aligned}
	\I
	&= \rho^{-\frac 23} \Bigg\Vert \sum_{\alpha, \beta = 1}^M
		\int \d z' \d x \d y \;
		\overline{\varphi_K(x-z')_{\alpha, \bar\beta}}
		\chi_R(y-z')
		a(\cW_z v_x^{\bar\alpha}) a(u_x^\alpha) b_{\beta,y} \xi_\lambda \Bigg\Vert \\
	&\le C \rho^{-\frac 23} \int \d z'
		\sum_{\beta = 1}^M \left\Vert \int \d x \d x' \d y \d y'
		A_\beta(x,x') B_\beta(y,y') a_x a_{x'} a_y a_{y'} \xi_\lambda \right\Vert \;, \\
	A_\beta &:= \sum_{\alpha=1}^M \cW_z v^\alpha \overline{\varphi_{K, \alpha, \bar\beta, z'}} u^\alpha \;, \qquad
	B_\beta := v^\beta \chi_{R, z'} u^\beta \;,
\end{aligned}
\end{equation}
where we recall that $ \chi_R $ is real and symmetric. We may directly apply Lemma~\ref{lem:aastar_conversion_multi}:
\begin{equation}
\begin{aligned}
	\I
	&\le C \rho^{-\frac 23} \int \d z' \sum_{\beta = 1}^M
		\norm{ A_\beta \cW_{z'}^{-\frac 12} }_{\HS}
		\norm{ B_\beta \cW_{z'}^{-\frac 12} }_{\HS}
		\Vert (\cN_{z'} + 1) \xi_\lambda \Vert \\
	&\le C \rho^{-\frac 23} \int \d z' \Bigg( \sum_{\beta = 1}^M
		\norm{ A_\beta \cW_{z'}^{-\frac 12} }_{\HS}^2 \Bigg)^{\frac 12}
		\Bigg( \sum_{\beta = 1}^M \norm{ B_\beta \cW_{z'}^{-\frac 12} }_{\HS}^2 \Bigg)^{\frac 12}
		\Vert (\cN_{z'} + 1) \xi_\lambda \Vert \;,
\end{aligned}
\end{equation}
We then estimate the Hilbert--Schmidt norms by Lemma~\ref{lem:HS_estimate} and Proposition~\ref{prop:regularity}, using that by Lemma~\ref{lem:Peetre}, $ \Vert \cW_z \cW_{z'} \Vert_{\op} \le C \cW(z-z') $:
\begin{equation}
\begin{aligned}
	&\norm{ A_\beta \cW_{z'}^{-\frac 12} }_{\HS}
	\le \Vert \cW_z \cW_{z'} \Vert_{\op}
		\Bigg\Vert \sum_{\alpha = 1}^M \cW_{z'}^{-1} v^\alpha \overline{\varphi_{K, \alpha, \bar\beta,z'}} u^\alpha \cW_{z'}^{-\frac 12} \Bigg\Vert_{\HS} \\
	&\le C \fl^{-6} \rho^{\frac 13}
		\max_{1 \le \alpha \le M} \big\Vert (1-\Delta_k)^4 \hat\varphi_{K, \alpha, \bar\beta} \big\Vert_1 \cW(z-z')
	\le C \rho^{\frac 13} \fl^{-14} |\log \fl| \cW(z-z') \;, \\
	&\sum_{\beta = 1}^M \norm{ B_\beta \cW_{z'}^{-\frac 12} }_{\HS}^2
	\le C \fl^{-4} \rho^{\frac 23} \Vert (1-\Delta_k)^2 \hat\chi_R \Vert_1^2 
	\le C \fl^{-4} \rho^{\frac 23} \;.
\end{aligned}
\end{equation}
Note that the complex conjugation of $ \overline{\varphi_{K, \alpha, \bar\beta,z'}} $ here caused a reflection of $ \hat\varphi_{K, \alpha, \bar\beta} $, which does not change the $ \ell^1 $-norm. Therefore, we arrive at
\begin{equation}
	\I
	\le C \fl^{-16} |\log \fl| M^{\frac 12} \int \d z' \cW(z-z')
	\lVert (\cN_{z'} + 1) \xi_\lambda \rVert \;.
\end{equation}
The bounds for $ \II, \III $ and $ \IV $ are analogous, so with \eqref{eq:dtfzlambda2_split} we conclude
\begin{equation}
	\lVert [S_-,\cN_z] \xi_\lambda \rVert
	\le C \fl^{-16} |\log \fl| M^{\frac 12} \int \d z' \cW(z-z')
		\lVert (\cN_{z'} + 1) \xi_\lambda \rVert \;.
\end{equation}
The bound for $ \Vert [S_-,\cN_z]^* \xi_\lambda \Vert $ is analogous, where we use the fourth bound of Lemma~\ref{lem:aastar_conversion_multi} instead of the first one. Thus,
\begin{equation}
	\lvert \partial_\lambda f_{(2)}(\lambda) \rvert
	\le C \fl^{-16} |\log \fl| M^{\frac 12} \int \d z' \cW(z-z')
		\lVert (\cN_{z'} + 1) \xi_\lambda \rVert \lVert (\cN_z + 1) \xi_\lambda \rVert\;.
\end{equation}
To close the Grönwall bound, it remains to replace $ \cN_{z'} $ by $ \cN_z $, which we achieve by translation invariance of $ \xi_\lambda $ (compare~\eqref{eq:Gamma_cTy}), so $ \langle \xi_\lambda, (\cN_{z'})^m \xi_\lambda \rangle = \langle \xi_\lambda, (\cN_z)^m \xi_\lambda \rangle $:
\begin{equation} \label{eq:close_Gronwall_f2}
	|\partial_\lambda f_{(2)}(\lambda)|
	\le C \fl^{-16} |\log \fl| M^{\frac 12} \int \d z' \cW(z-z')
		\Vert (\cN_z + 1) \xi_\lambda \Vert^2
	\le  C \fl^{-16} |\log \fl| M^{\frac 12} f_{(2)}(\lambda) \;.
\end{equation}
Applying Grönwall's lemma then gives the desired second bound in \eqref{eq:gronwall}.
\end{proof}

\section{Improved Excitation Number Bounds}
\label{sec:improved_bounds}

In this section, we will improve the bound~\eqref{eq:preliminary_bounds} on $ \cN $ entailed by Proposition~\ref{prop:gronwall}, removing the factor of $ e^{C \fl^{-7}} $.

\begin{proposition} [Extensive Excitation Number Estimate] \label{prop:improved_excitation}
Let $ \xi_\lambda := T_\lambda \Omega $, see~\eqref{eq:T}. Then, there exists $ C > 0 $ such that for $ \lambda \in [0,1]$, we have
\begin{equation} \label{eq:improved_excitation}
	\langle \xi_\lambda, \cN \xi_\lambda \rangle
	\le C L^3 \;.
\end{equation}
\end{proposition}

The proof is based on an iterated Duhamel principle, which will allow us to calculate the expectation of any quadratic pseudo-bosonic operators in $ \xi_\lambda $. Given a matrix-valued function $ X: \Lambda_L \to \RRR^{M \times M} $, recalling definition~\eqref{eq:bstar} of $ b_\alpha^*(k) $, we set
\begin{equation} \label{eq:Q1Q2}
\begin{aligned}
	Q_1(X) &:= \sum_{k \in \Lambda_L} \sum_{\alpha, \beta = 1}^M X(k)_{\alpha, \beta} b^*_\alpha(k) b_\beta(k) + \mathrm{h.c.} \;, \\
	Q_2(X) &:= \sum_{k \in \Lambda_L} \sum_{\alpha, \beta = 1}^M X(k)_{\alpha, \bar\beta} b^*_\alpha(k) b^*_\beta(-k) + \mathrm{h.c.}
\end{aligned}
\end{equation}
In position space, defining $ \varphi_{X,\alpha,\beta} $ as in~\eqref{eq:varphi}, and using~\eqref{eq:positionspace_trafo}, we have
\begin{equation} \label{eq:Q1Q2position}
\begin{aligned}
	Q_1(X) &= \frac{1}{\rho^{\frac 23}} \sum_{\alpha, \beta = 1}^M \int \d x \d y \varphi_X(x-y)_{\alpha, \beta} b^*_{\alpha,x} b_{\beta,y} + \mathrm{h.c.} \;, \\
	Q_2(X) &= \frac{1}{\rho^{\frac 23}} \sum_{\alpha, \beta = 1}^M \int \d x \d y \varphi_X(x-y)_{\alpha, \bar\beta} b^*_{\alpha,x} b^*_{\beta,y} + \mathrm{h.c.}
\end{aligned}
\end{equation}
Evaluating $ Q_1(X) $ and $ Q_2(X) $ in the state $ \xi_\lambda $ will produce matrices of the following form:
\begin{equation} \label{eq:Xtilde}
	\tilde{X}^{(n)} := \sum_{n'=0}^n {n \choose n'}
		(K I)^{n'} X (I K)^{n-n'} \;,
\end{equation}
with $ I(k) $ as in \eqref{eq:cEalphahat} and $ K(k) $ as in \eqref{eq:K}. The evaluation reads as follows.

\begin{proposition}[Expansion of Pseudo--Bosonic Expectation Values] \label{prop:evaluate_Q1Q2}
Let $ X: \Lambda_L \to \RRR^{M \times M} $ with $ \supp(X_{\alpha, \beta}) \subset B_{2R}(0) $ be a matrix-valued function with $ X(k)_{\alpha,\beta} = X(k)_{\beta,\alpha} = X(-k)_{\bar\alpha,\bar\beta} $ such that the $ \tilde{X}^{(n)}: \Lambda_L \to \RRR^{M \times M} $ associated according to~\eqref{eq:Xtilde} satisfies the following regularity estimate: There exists a $ C > 0 $ and some possibly $ \rho $-dependent $ D_X > 0 $ such that
\begin{equation} \label{eq:regularityestimate_Q1Q2}
	\Vert (1 - \Delta_k)^2 \hat\varphi_{\tilde{X}^{(n)},\alpha,\beta} \Vert_1
	\le C^n D_X \qquad \forall n \in \NNN_0, \; \alpha, \beta \in \{1,\ldots,M\} \;.
\end{equation}
Then, in the state $ \xi_\lambda := T_\lambda \Omega $ defined in~\eqref{eq:T}, for any $ \varepsilon > 0 $, there exists a $ C_\varepsilon > 0 $ such that for all $ \lambda \in [0,1]$,
\begin{equation} \label{eq:evaluate_Q1Q2}
\begin{aligned}
	\Bigg\vert \langle \xi_\lambda, Q_1(X) \xi_\lambda \rangle
	- \sum_{\substack{n=2 \\ n : \textnormal{even}}}^\infty \frac{(2 \lambda)^n}{n!} \sum_{k \in \Lambda_L} \tr \Big( I(k) X(k) \big( I(k) K(k) \big)^n \Big) \Bigg\vert
	&\le C_\varepsilon D_X \rho^{-\frac 13 + \varepsilon} L^3 \;, \\
	\Bigg\vert \langle \xi_\lambda, Q_2(X) \xi_\lambda \rangle
	- \sum_{\substack{n=1 \\ n : \textnormal{odd}}}^\infty \frac{(2 \lambda)^n}{n!} \sum_{k \in \Lambda_L} \tr \Big( I(k) X(k) \big( I(k) K(k) \big)^n \Big) \Bigg\vert
	&\le C_\varepsilon D_X \rho^{-\frac 13 + \varepsilon} L^3 \;.
\end{aligned}
\end{equation}
\end{proposition}

The proof of Proposition~\ref{prop:evaluate_Q1Q2} is based on a Duhamel expansion. The trace terms in~\eqref{eq:evaluate_Q1Q2} can be seen as the leading-order contribution to $ \langle \Omega, T_\lambda^* Q_j(X) T_\lambda \Omega \rangle $, similar to the bosonized momentum distribution evaluated in~\cite{BL23,BLN25}.

\begin{proof}[Proof of Proposition~\ref{prop:improved_excitation} from Proposition~\ref{prop:evaluate_Q1Q2}]

Recalling from~\eqref{eq:T} that $ T_\lambda = e^{- \lambda S} $, Duhamel principle yields
\begin{equation} ~\label{eq:Duhamel_1}
	\langle \xi_\lambda, \cN \xi_\lambda \rangle
	= \langle \Omega, e^{\lambda S} \cN e^{-\lambda S} \Omega \rangle
	= \langle \Omega, \cN \Omega \rangle
		+ \int_0^\lambda \d \lambda_1 \langle \Omega, e^{\lambda_1 S} [S,\cN] e^{-\lambda_1 S} \Omega \rangle \;,
\end{equation}
where $ \langle \Omega, \cN \Omega \rangle = 0 $. Comparing with~\eqref{eq:S_position}, we may then conveniently write
\begin{equation} \label{eq:SN_commutator}
	[S, \cN] = Q_2(2K) \;.
\end{equation}
According to Proposition~\ref{prop:regularity}, $ X = 2K $ satisfies the regularity estimate~\eqref{eq:regularityestimate_Q1Q2} with $ D_X = C \fl^{-4} |\log \fl| \le C_\varepsilon \rho^\varepsilon $. Moreover, since $ X \in \{h,P\} $ from~\eqref{eq:hP} satisfy the symmetry property $ X(k)_{\alpha,\beta} = X(k)_{\beta,\alpha} = X(-k)_{\bar\alpha,\bar\beta} $, so does $ K $ from~\eqref{eq:K}. Applying Proposition~\ref{prop:evaluate_Q1Q2} and bounding the trace term by Proposition~\ref{prop:trace_estimates}, we thus obtain with $ |\lambda| \le 1 $
\[
	\langle \xi_\lambda, \cN \xi_\lambda \rangle
	\le 2 \sum_{\substack{n=1 \\ n : \textnormal{odd}}}^\infty \frac{2^n}{n!} \sum_{k \in \Lambda_L} \tr \Big( \big( I(k) K(k) \big)^{n+1} \Big)
		+ C_\varepsilon \rho^{-\frac 13 + \varepsilon} L^3 
	\le \sum_{n=1}^\infty \frac{C^n}{n!} L^3 + C L^3 
	\le C L^3\;.
\]
\end{proof}
The rest of this section is dedicated to the proof of Proposition~\ref{prop:evaluate_Q1Q2}.

\subsection{Bosonization Error Bounds}

Iterating \eqref{eq:Duhamel_1} will produce commutators between pseudo-bosonic commutators. Let the transpose and reflection be given by $ X^T(k)_{\alpha,\beta} := X(k)_{\beta,\alpha} $ and $ \fR(X)(k)_{\alpha,\beta} = X(-k)_{\bar\alpha,\bar\beta} $. The following lemma holds for generic $ X $, which are not necessarily symmetric or reflection invariant.

\begin{lemma}[Bosonization error computation] \label{lem:bosonizationerror}
Recall the definitions of the matrix $I$ from \eqref{eq:cEalphahat}, of the matrix $K$ from (\ref{eq:K}), of the operator $ S $ from \eqref{eq:T}, and the Q-notation from \eqref{eq:Q1Q2}.  
Let $ X: \Lambda_L \to \RRR^{M \times M} $ with $ \supp(X_{\alpha, \beta}) \subset B_{2R}(0) $ satisfy the following regularity estimate: There exists some $ D_X > 0 $, possibly depending on $ \rho $ but not on $ L $, such that
\begin{equation}~\label{eq:regularityestimate_2}
	\Vert (1 - \Delta_k)^2 \hat\varphi_{X,\alpha,\beta} \Vert_1 , \Vert \hat\varphi_{ \fR(KIX + X^TIK)/2,\alpha,\beta} \Vert_1
	\le D_X \qquad \forall \alpha,\beta \in \{1,\ldots,M\} \;.
\end{equation}
Then 
\begin{equation} \label{eq:bosonizationerror}
\begin{split}
	[S, Q_1(X)] &= Q_2 \big( X I K + \fR(K I X) \big) + \cE_{Q_1}(X) \;, \\
	[S, Q_2(X)] &= Q_1 \big(X I K + \fR(K I X) \big)
		+ 2 \sum_{k \in \Lambda_L} \tr\big( I(k) X(k) I(k) K(k) \big) + \cE_{Q_2}(X) \;,
\end{split}
\end{equation}
where 
\begin{equation}\label{eq:error-bos} \begin{split} 
| \langle \xi_\lambda , \cE_{Q_j} (X) \xi_\lambda \rangle |  \leq C D_X \fl^{-23} M^{\frac 12} e^{C \fl^{-17} M^{\frac 12}} \rho^{-\frac 13} L^3 
\end{split} \end{equation} 
for $j=1,2$ and all $\lambda \in [-1,1]$. 
\end{lemma}

\begin{proof}
From \eqref{eq:T}, we recall $S = S_- - S_-^*$. By definition, $ Q_1(X) = Q_1(X)^* $, so $ [S,Q_1(X)] = [S_-,Q_1(X)] + \mathrm{h.c.} $. We then compute
\begin{equation} \label{eq:Commutator_S-_Q1}
\begin{aligned}
	[S_-,Q_1(X)]
	&= \frac{1}{2 \rho^{4/3}} \sum_{\alpha, \beta, \alpha', \beta' = 1}^M \int \d x \d y \d x' \d y' \; \varphi_{(X+X^T)}(x-y)_{\alpha, \beta} \varphi_K(y'-x')_{\alpha',\bar\beta'}\\
	&\quad \times \Big( b_{\alpha',x'} [b_{\beta',y'}, b_{\alpha,x}^*] b_{\beta,y}
		+ [b_{\alpha',x'}, b_{\alpha,x}^*] b_{\beta',y'} b_{\beta,y} \Big) \;.
\end{aligned}
\end{equation}
Recalling $ b_{\alpha,x} = a(v_x^{\bar\alpha}) a(u_x^\alpha) $, we find, as in Lemma~\ref{lem:approxCCR}:
\begin{equation}\label{eq:bbxx}
	[b_{\alpha,x}, b^*_{\alpha',x'}]
	= \delta_{\alpha,\alpha'} \big( (v^{\bar\alpha})^2(x-x') (u^\alpha)^2(x-x') + \cE_\alpha(x,x') \big) \;,
\end{equation}
with the error term 
\begin{equation} \label{eq:cEalpha}
	\cE_\alpha(x,y)
	:= - (u^\alpha)^2(x-y) a^*(v_y^{\bar\alpha}) a(v_x^{\bar\alpha})
		- (v^{\bar\alpha})^2(x-y) a^*(u_y^\alpha) a(u_x^\alpha) \;,
\end{equation}
where $ (u^\alpha)^2(x-y) $ is understood as the kernel of the operator $ (u^\alpha)^2 $.\\
Inserting the leading term on the r.h.s. of (\ref{eq:bbxx}) into the first term on the r.h.s. of \eqref{eq:Commutator_S-_Q1}, we find the contribution 
\begin{equation}
\begin{aligned}
	&\frac{1}{2 \rho^{4/3}} \sum_{\alpha',\beta = 1}^M \int \d y \d x' \;  
		A_{\alpha',\beta}(x',y)
		b_{\alpha',x'} b_{\beta,y} \;, \qquad \textnormal{with}\\
	A_{\alpha',\beta}(x',y)	
	&:= \sum_{\alpha = 1}^M \int \d y' \d x \; \varphi_K(x'-y')_{\bar\alpha',\alpha} (v^{\bar\alpha})^2(y'-x) (u^\alpha)^2(y'-x) \varphi_{(X+X^T)}(x-y)_{\alpha, \beta} \;.
\end{aligned}
\end{equation}
In momentum space, $ A_{\alpha',\beta} $ has Fourier multiplier (compare~\eqref{eq:varphi})
\begin{equation}
\begin{aligned}
	\hat{A}_{\alpha',\beta}(k)
	&= \sum_{\alpha = 1}^M \hat\varphi_K(k)_{\bar\alpha',\alpha} \big((\hat{v}^{\bar\alpha})^2 *(\hat{u}^\alpha)^2\big)(k) \hat\varphi_{(X+X^T)}(k)_{ \alpha, \beta} \\
	&= \sum_{\alpha = 1}^M L^6 \rho^{\frac 43} 
		\frac{K(k)_{\bar\alpha',\alpha}}{\tilde{n}_{\bar\alpha',k} \tilde{n}_{\alpha,k}}
		\frac{1}{L^3} \sum_{p \in \Lambda_L} \hat{v}^{\bar\alpha}(k-p)^2 \hat{u}^\alpha(p)^2
		\frac{(X+X^T)(k)_{\alpha,\beta}}{\tilde{n}_{\alpha,k} \tilde{n}_{\beta,k}} \;.
\end{aligned}
\end{equation}
With $ \hat{v}^{\bar\alpha}(k-p) = \hat{v}^\alpha(p-k) $, and using the definitions \eqref{eq:n} for $ n_{\alpha,k} $ and \eqref{eq:cEalphahat} for $ I(k) $, we write
\begin{equation}
	\sum_{p \in \Lambda_L} \hat{v}^{\bar\alpha}(k-p)^2 \hat{u}^\alpha(p)^2
	= n_{\alpha,k}^2
	= I(k)_{\alpha,\alpha} \tilde{n}_{\alpha,k}^2 \;,
\end{equation}
and therefore
\begin{equation}
	\hat{A}_{\alpha',\beta}(k)
	= L^3 \rho^{\frac 43} 
		\frac{(K I (X+X^T))(k)_{\bar\alpha',\beta}}{\tilde{n}_{\bar\alpha',k} \tilde{n}_{\beta,k}}
	= \rho^{\frac 23} \hat\varphi_{K I (X+X^T)}(k)_{\bar\alpha',\beta}\;.
\end{equation} 
Proceeding analogously to handle the second term on the r.h.s. of \eqref{eq:Commutator_S-_Q1} and adding the contribution of $S_-^*$, we obtain 
\[  [S, Q_1(X)] = Q_2 \big( X I K + \fR(K I X) \big) + \cE_{Q_1}(X) \]
with $\cE_{Q_1}(X) := \cE_{b \cE b}(X+X^T) + \cE_{\cE bb}(X+X^T) + \mathrm{h.c.}$ and 
\begin{equation} \label{eq:cEbbcE1}
\begin{aligned}
	\cE_{b \cE b}(X) &:= \frac{1}{2 \rho^{4/3}} \sum_{\alpha, \beta, \alpha' = 1}^M \int \d x \d y \d x' \d y' \; \varphi_X(x-y)_{\alpha, \beta} \varphi_K(y'-x')_{\alpha',\bar\alpha} b_{\alpha',x'} \cE_\alpha(y',x) b_{\beta,y} \;, \\
	\cE_{\cE bb}(X) &:= \frac{1}{2 \rho^{4/3}} \sum_{\alpha, \beta, \beta' = 1}^M \int \d x \d y \d x' \d y' \; \varphi_X(x-y)_{\alpha, \beta} \varphi_K(y'-x')_{\alpha,\bar\beta'} \cE_\alpha(x',x) b_{\beta',y'} b_{\beta,y} \;, \\
\end{aligned}
\end{equation} 
Similar direct computations lead to  
\[ [S, Q_2(X)] = Q_1 \big(X I K + \fR(K I X) \big)
		+ 2 \sum_{k \in \Lambda_L} \tr\big( I(k) X(k) I(k) K(k) \big) + \cE_{Q_2}(X) \]
where 
\[ \cE_{Q_2}(X) := \cE_{b \cE b^*}(X) + \cE_{\cE bb^*}(X) + \cE_{b^* b \cE}(X) + \cE_{b^* \cE b}(X) + \cE_\cE(Y_X) + \mathrm{h.c.} \;, \]
with $Y_X = \fR (KIX+X^TIK)/2$ and 
\begin{equation} \label{eq:cEbbcE2}
\begin{aligned}
	\cE_{b \cE b^*}(X) &:= \frac{1}{2 \rho^{4/3}} \sum_{\alpha, \beta, \alpha' = 1}^M \int \d x \d y \d x' \d y' \; \varphi_X(x-y)_{\alpha, \bar\beta} \varphi_K(y'-x')_{\alpha',\bar\alpha} b_{\alpha',x'} \cE_\alpha(y',x) b^*_{\beta,y} \;, \\
	\cE_{\cE b b^*}(X) &:= \frac{1}{2 \rho^{4/3}} \sum_{\alpha, \beta, \beta' = 1}^M \int \d x \d y \d x' \d y' \; \varphi_X(x-y)_{\alpha, \bar\beta} \varphi_K(y'-x')_{\alpha,\bar\beta'} \cE_\alpha(x',x) b_{\beta',y'} b^*_{\beta,y} \;, \\
	\cE_{b^*b \cE}(X) &:= \frac{1}{2 \rho^{4/3}} \sum_{\alpha, \beta, \alpha' = 1}^M \int \d x \d y \d x' \d y' \; \varphi_X(x-y)_{\alpha, \bar\beta} \varphi_K(y'-x')_{\alpha',\bar\beta} b^*_{\alpha,x} b_{\alpha',x'} \cE_\beta(y',y) \;, \\
	\cE_{b^* \cE b}(X) &:= \frac{1}{2 \rho^{4/3}} \sum_{\alpha, \beta, \beta' = 1}^M \int \d x \d y \d x' \d y' \; \varphi_X(x-y)_{\alpha, \bar\beta} \varphi_K(y'-x')_{\beta,\bar\beta'} b^*_{\alpha,x} \cE_\beta(x',y) b_{\beta',y'} \;, \\
	\cE_\cE(Y) &:= \frac{1}{\rho^{2/3}} \sum_{\alpha = 1}^M \int \d x \d y \; \varphi_Y(x-y)_{\alpha, \alpha} \cE_\alpha(y,x)  \;,
\end{aligned}
\end{equation}
To conclude the proof of the lemma, we have to show that all error terms in (\ref{eq:cEbbcE1}), (\ref{eq:cEbbcE2}) satisfy the bound (\ref{eq:error-bos}). 

We start with the term $ \cE_{b \cE b}(X)$ in (\ref{eq:cEbbcE1}). Inserting \eqref{eq:cEalpha} leads to the natural split
\begin{equation} \label{eq:bEb-uv}
\begin{aligned}
	\cE_{b \cE b}(X) &= \cE_{b \cE b}^{(u)}(X) + \cE_{b \cE b}^{(v)}(X) \;, \qquad \textnormal{with} \\
	\cE_{b \cE b}^{(u)}(X)
	&:= - \frac{1}{2 \rho^{4/3}} \sum_{\alpha, \beta, \alpha' = 1}^M \int \d x \d y \d x' \d y' \; \varphi_K(y'-x')_{\alpha',\bar\alpha} (u^\alpha)^2(y'-x) \varphi_X(x-y)_{\alpha, \beta}  \times  \\
	&\quad \times b_{\alpha',x'} a^*(v^{\bar\alpha}_x) a(v^{\bar\alpha}_{y'})  b_{\beta,y} \;, \\
	\cE_{b \cE b}^{(v)}(X)
	&:= - \frac{1}{2 \rho^{4/3}} \sum_{\alpha, \beta, \alpha' = 1}^M \int \d x \d y \d x' \d y' \; \varphi_K(y'-x')_{\alpha',\bar\alpha} (v^{\bar\alpha})^2(y'-x) \varphi_X(x-y)_{\alpha, \beta} \times  \\
	&\quad \times b_{\alpha',x'} a^*(u^\alpha_x) a(u^\alpha_{y'})  b_{\beta,y} \;.
\end{aligned}
\end{equation}
As in~\eqref{eq:varphisplit}, we use a radial cutoff function $ \hat\chi_{2R}(k) := \hat\chi_R(k/2) $ with $ \hat\chi_R $ introduced above~\eqref{eq:DB} and split
\begin{equation} \label{eq:varphi_A_split}
	\varphi_A(x-y)_{\alpha, \beta}
	= \int \d z \;
		\chi_{2R}(x-z)
		\varphi_A(z-y)_{\alpha, \beta} \;,
\end{equation}
first with $ A = X $ and later with $ A = K $. Note that this split is possible since $ K(k) $ and $ X(k) $ are supported in $ |k| \le 2R $. Then, by the Cauchy--Schwarz inequality, and using that $ \chi_{2R} $ is real and symmetric, that $ \overline{\varphi_K(y'-x')_{\alpha',\bar\alpha}} = \varphi_K(x'-y')_{\alpha',\bar\alpha} $ and $ \overline{(u^\alpha)^2(y'-x)} = (u^{\bar\alpha})^2(y'-x) $,
\begin{equation} \label{eq:bEb_split}
\begin{aligned}
	&| \langle \xi_\lambda, \cE_{b \cE b}^{(u)}(X) \xi_\lambda \rangle |
	\\ &\le \frac{C}{\rho^{4/3}} \int \d z
		\sum_{\alpha = 1}^M
		\norm{\sum_{\beta = 1}^M \int \d y \varphi_X(z-y)_{\alpha,\beta} b_{\beta, y} \xi_\lambda} \times \\
	& \quad \times \norm{\sum_{\alpha' = 1}^M \int \d x \d x' \d y'
		\varphi_K(x'-y')_{\alpha',\bar\alpha}
		(u^{\bar\alpha})^2(y'-x)
		\chi_{2R}(x-z)
		a^*(v^{\bar\alpha}_{y'}) a(v^{\bar\alpha}_x) b_{\alpha',x'}^* \xi_\lambda } \, .
\end{aligned}
\end{equation}
The first norm is readily bounded by Corollary~\ref{cor:bbstar_conversion} and Lemma~\ref{lem:HS_estimate} with $ a = 0 $ and $ b = 1 $, using the regularity estimate~\eqref{eq:regularityestimate_2}
\begin{equation} \label{eq:bEb_firstnorm_final}
\begin{aligned}
	\norm{\sum_{\beta = 1}^M \int \d y \varphi_X(z-y)_{\alpha,\beta} b_{\beta, y} \xi_\lambda}
	&\le C \fl^{-2} \rho^{\frac 13} \max_{1 \le \beta \le M} 
		\norm{(1-\Delta_k)^2 \hat\varphi_{X,\alpha,\beta}}_1
		\Vert \cN_z^{\frac 12} \xi_\lambda \Vert \\
	&\le C D_X \fl^{-2} \rho^{\frac 13}
		\norm{(\cN_z+1)^{\frac 12} \xi_\lambda} \;,
\end{aligned}
\end{equation}
where we used that $ \varphi_X(z-y)_{\alpha,\beta} = \varphi_{\fR(X)}(y-z)_{\bar\alpha,\bar\beta} $, where the reflection leaves the $ \ell^1 $-norm invariant.
For the second norm, we would like to apply~\eqref{eq:aastar_conversion_multi}. However, $ \varphi_K(y'-x')_{\alpha',\bar\alpha} $ prevents us from doing so, as it both contains an index of the first ($ y '$) and second fermionic operator pair ($ x' $). We therefore split it as in~\eqref{eq:varphi_A_split} and obtain
\begin{equation}
\begin{aligned} \label{eq:AB_split}
	&\norm{\sum_{\alpha' = 1}^M \int \d x \d x' \d y'
		\varphi_K(x'-y')_{\alpha',\bar\alpha}
		(u^{\bar\alpha})^2(y'-x)
		\chi_{2R}(x-z)
		a^*(v^{\bar\alpha}_{y'}) a(v^{\bar\alpha}_x) b_{\alpha',x'}^* \xi_\lambda } \\
	&\le \int \d z'
		 \Bigg\Vert \sum_{\alpha' = 1}^M \int \d x \d x' \d y'
		\chi_{2R}(y'-z')
		(u^{\bar\alpha})^2(y'-x)
		\chi_{2R}(x-z)
		a^*(v^{\bar\alpha}_{y'}) a(v^{\bar\alpha}_x) \times \\
	&\quad \times 
		\varphi_K(x'-z')_{\alpha',\bar\alpha} b_{\alpha',x'}^* \xi_\lambda \Bigg\Vert \\
	&= \int \d z' \left\Vert \int \d x \d x' \d y \d y' A_\alpha(x,x') B_\alpha(y,y') a^*_x a_{x'} a^*_y a^*_{y'} \xi_\lambda \right\Vert \;,
\end{aligned}
\end{equation}
where we introduced
\begin{equation} \label{eq:AB_definition}
	A_\alpha := v^{\bar\alpha} \chi_{2R,z'} (u^{\bar\alpha})^2 \chi_{2R,z} v^{\bar\alpha} \;, \qquad
	B_\alpha := \sum_{\alpha' = 1}^M u^{\alpha'} \varphi_{K,\alpha',\bar\alpha,z'} v^{\alpha'} \;,
\end{equation}
with $ \chi_{2R,z}(x) = \chi_{2R}(x-z) $. Applying Lemma~\ref{lem:aastar_conversion_multi} we obtain
\begin{equation} \label{eq:bEb_secondnorm_final}
\begin{aligned}
	&\Bigg\Vert \sum_{\alpha' = 1}^M \int \d x \d x' \d y'
		\chi_{2R}(y'-z')
		(u^{\bar\alpha})^2(y'-x)
		\chi_{2R}(x-z)
		a^*(v^{\bar\alpha}_{y'}) a(v^{\bar\alpha}_x)
		\varphi_K(x'-z')_{\alpha',\bar\alpha} b_{\alpha',x'}^* \xi_\lambda \Bigg\Vert \\
	&\le C \norm{ A_\alpha \cW_{z'}^{-\frac 12} }_{\HS}
		\norm{ B_\alpha \cW_{z'}^{-\frac 12} }_{\HS}
		\Vert (\cN_{z'} + 1) \xi_\lambda \Vert \;,
\end{aligned}
\end{equation}
where we estimate the second Hilbert--Schmidt norm via Lemma~\ref{lem:HS_estimate} and Proposition~\ref{prop:regularity}:
\begin{equation}
	\norm{ B_\alpha \cW_{z'}^{-\frac 12} }_{\HS}
	\le C \fl^{-2} \rho^{\frac 13}
		\max_{1 \le \alpha' \le M} \big\Vert (1-\Delta_k)^2 \hat\varphi_{K,\alpha',\bar\alpha} \big\Vert_1 \\
	\le C \rho^{\frac 13} \fl^{-6} |\log \fl| \;.
\end{equation}
So with~\eqref{eq:bEb_firstnorm_final} and~\eqref{eq:bEb_secondnorm_final}, the bound~\eqref{eq:bEb_split} for $ \cE_{b \cE b}^{(u)}(X) $ becomes
\begin{equation} \label{eq:I_bEb_u_intermediate_estimate}
\begin{aligned}
	& | \langle \xi_\lambda, \cE_{b \cE b}^{(u)}(X) \xi_\lambda \rangle | \\
	&\le C D_X \fl^{-8} |\log \fl| \rho^{-\frac 23}
		\int \d z' \d z
		\Vert (\cN_z+1)^{\frac 12} \xi_\lambda \Vert 
		\sum_{\alpha = 1}^M \norm{ A_\alpha \cW_{z'}^{-\frac 12} }_{\HS}
		\Vert (\cN_{z'} + 1) \xi_\lambda \Vert \;.
\end{aligned}
\end{equation}
To estimate $\| A_\alpha \cW_{z'}^{-1/2}\|_\HS$, we first apply Lemma \ref{lem:HS_uv_removal}. From (\ref{eq:AB_definition}), reversing $ \alpha \mapsto \bar\alpha $ for convenience, we find for any $ \alpha \in \{1,\ldots,M\} $
\[ \| A_{\bar\alpha}  \cW_{z'}^{-1/2} \|_\text{HS} \leq C \fl^{-2} \| v^\alpha \chi_{2R,z'} (u^\alpha)^2 \chi_{2R,z} \cW_{z'}^{-1/2} \|_\HS \;. \]
To extract decay in $(z-z')$, we notice that, for arbitrary $a,b \in \mathbb{Z}$, $\| \cW_{z'}^{-b/2} \chi_{2R,z'} \cW_z^{a/2} \|_\infty \leq C_{ab} (1 + |z-z'|)^{-2a}$, since $\chi_{2R}$ exhibits arbitrarily fast polynomial decay. Hence, we can control  
\[  \begin{split}  \| &A_{\bar\alpha}  \cW_{z'}^{-1/2} \|_\text{HS} \\ &\leq C \fl^{-2} \| v^\alpha  \cW_{z'}^{1/2} \|_\HS \| \cW_{z'}^{-1/2} \chi_{2R,z'} \cW_z^{a/2} \|_\text{op} \| \cW_z^{-a/2} (u^\alpha)^2 \cW^{a/2}_z \|_\text{op} \| \cW_z^{-a/2} \chi_{2R,z} \cW_{z'}^{-1/2} \|_\text{op} \\ &\leq C \rho^{1/3}\fl^{-2} M^{-1/2} (1+ |z-z'|)^{-2(a-1)}  \| \cW_z^{-a/2} (u^\alpha)^2 \cW^{a/2}_z \|_\text{op} \;, \end{split}  \]
where we bounded $\| v^\alpha  \cW_{z'}^{1/2} \|_\HS \leq \| \hat{v}^\alpha \|_2 \| \cW_{z'}^{1/2} \|_2 \leq C \rho^{1/3} / \sqrt{M}$. Writing 
\[ \begin{split} \cW_z^{-a/2}  (u^\alpha)^2 \cW_z^{a/2} = \; &(u^{\alpha})^2 +u^\alpha [ \cW_z^{-a/2}, u^\alpha] \cW_z^{a/2} + [\cW_z^{-a/2}, u^\alpha ] \cW_z^{a/2} u^\alpha \\ &+ [\cW_z^{-a/2}, u^\alpha] \cW_z^{a/2} [\cW_z^{-a/2}, u^\alpha ] \cW_z^{a/2} \;, \end{split} \]
and using Lemma \ref{lem:ualpha_derivative_bound}, we obtain
\[ \| W_z^{-a/2}  (u^\alpha)^2 W_z^{a/2}  \|_\text{op} \leq C \fl^{-4a} \;. \]
Hence
\[  \| A_{\bar\alpha}  \cW_{z'}^{-1/2} \|_\text{HS}  \leq C \fl^{-4a-2} \rho^{\frac 13} M^{-\frac 12} (1+|z-z'|)^{-2(a-1)} \;. \]
Inserting in \eqref{eq:I_bEb_u_intermediate_estimate} and applying Prop. \ref{prop:gronwall}, we arrive at 
\[  | \langle \xi_\lambda, \cE_{b \cE b}^{(u)}(X) \xi_\lambda \rangle | \leq C D_X \fl^{-10-4a} |\log \fl | M^{\frac 12} \rho^{-\frac 13} e^{C \fl^{-17} M^{\frac 12} }  \int \d z \d z' \frac{1}{(1+|z-z'|)^{2(a-1)}} \;. \]
Choosing $a=3$, we conclude that 
\[  | \langle \xi_\lambda, \cE_{b \cE b}^{(u)}(X) \xi_\lambda \rangle |  \leq C D_X \fl^{-23}  M^{\frac 12} e^{C \fl^{-17} M^{\frac 12}} \rho^{-\frac 13}  L^3 \;. \]

The estimate for $ \cE_{b \cE b}^{(v)}(X) $ in (\ref{eq:bEb-uv}) can be shown analogously, swapping the roles of $ u^\alpha $ and $ v^{\bar\alpha} $.
The bound for $ | \langle \xi_\lambda, \cE_{\cE bb}(X) \xi_\lambda \rangle | $ in (\ref{eq:cEbbcE1}) works similarly, where in~\eqref{eq:AB_split}, we have to replace $ a_x^* a_{x'} a_y^* a_{y'}^* $ with $ a_x^* a_{x'}^* a_y^* a_{y'} $ and~\eqref{eq:AB_definition} becomes
\[
	A_\alpha
	:= \sum_{\beta'=1}^M u^{\beta'} \overline{\varphi_{K,\alpha,\bar\beta',z'}} v^{\beta'} \;, \qquad 
	B_\alpha
	:= v^{\bar\alpha} \chi_{2R,z'} (u^{\bar\alpha})^2 \chi_{2R,z} v^{\bar\alpha} \;.
\]
The bounds for $ \cE_{b \cE b^*}(X) $ and $ \cE_{\cE b b^*}(X) $ in (\ref{eq:cEbbcE2}) are analogous to those for $ \cE_{b \cE b}(X) $ and $ \cE_{\cE b b}(X) $, replacing $ b_{\beta,y} $ by $ b_{\beta,y}^* $ and $ \varphi_{X,\alpha,\beta} $ by $ \varphi_{X,\alpha,\bar\beta} $. Also the bound for $ \cE_{b^* b \cE}(X) $ works similarly; as in~\eqref{eq:bEb_split}, we split 
\begin{equation}
\begin{aligned}
	& | \langle \xi_\lambda, \cE_{b^* b \cE}^{(u)}(X) \xi_\lambda \rangle | \\
	&\le \frac{C}{\rho^{4/3}} \int \d z
		\sum_{\beta = 1}^M
		\norm{\sum_{\alpha = 1}^M \int \d x \; \overline{\varphi_X(x-z)_{\alpha,\bar\beta}} b_{\alpha, x} \xi_\lambda} \\
	& \quad \times \norm{\sum_{\alpha' = 1}^M \int \d y \d x' \d y' \;
		\varphi_K(y'-x')_{\alpha',\bar\beta}
		(u^\beta)^2(y'-y)
		\chi_{2R}(y-z)
		b_{\alpha',x'} a^*(v^{\bar\beta}_y) a(v^{\bar\beta}_{y'}) \xi_\lambda } \;.
\end{aligned}
\end{equation}
The first norm is bounded as in~\eqref{eq:bEb_firstnorm_final} and below, and the second as in~\eqref{eq:AB_split} with $ a_x^* a_{x'} a_y^* a_{y'}^* $ replaced by $ a_x a_{x'} a_y^* a_{y'} $ and $ A_\alpha, B_\alpha $ in~\eqref{eq:AB_definition} replaced by
\[
	A_\beta
	:= \sum_{\alpha'=1}^M v^{\alpha'} \overline{\varphi_{K,\alpha',\bar\beta,z'}} u^{\alpha'} \;, \qquad
	B_\beta
	:= v^{\bar\beta} \chi_{2R,z} (u^{\bar\beta})^2 \chi_{2R,z'} v^{\bar\beta} \;.
\]
The bound for $ \cE_{b^* \cE b}(X) $ is again analogous to the one for $ \cE_{b^* b \cE}(X) $. 

Finally, to estimate $ \cE_\cE(Y) $ in (\ref{eq:cEbbcE2}), we write $\cE_\cE(Y) = \cE_\cE^{(u)}(Y) + \cE_\cE^{(v)}(Y)$, with 
\begin{equation}
\begin{aligned}
\cE_\cE^{(u)}(Y) &:= -\frac{1}{\rho^{2/3}} \d \Gamma \Big( \sum_{\alpha = 1}^M v^{\bar\alpha} \big( \hat\varphi_{Y,\alpha,\alpha} * (\hat{u}^{\bar\alpha})^2 \big) v^{\bar\alpha} \Big) \;, \\
	\cE_\cE^{(v)}(Y)
	&:= -\frac{1}{\rho^{2/3}} \d \Gamma \Big( \sum_{\alpha = 1}^M u^\alpha \big(\hat\varphi_{Y,\alpha,\alpha} * (\hat{v}^\alpha)^2 \big) u^\alpha \Big) \;,
\end{aligned}
\end{equation}
where $ \hat\varphi_{Y,\alpha,\alpha} * (\hat{u}^{\bar\alpha})^2 $ is understood as multiplication operator in momentum space. With $ \langle \xi, \d \Gamma(A) \xi \rangle \le \norm{A}_{\op} \langle \xi, \cN \xi \rangle $, and substituting $ \alpha \mapsto \bar\alpha $, we conclude
\begin{equation}
	|\langle \xi_\lambda, \cE_\cE^{(u)}(Y) \xi_\lambda \rangle|
	\le C \rho^{-\frac 23}
		|\langle \xi_\lambda, \cN \xi_\lambda \rangle|
		\norm{\sum_{\alpha = 1}^M v^\alpha \big( \hat\varphi_{Y,\bar\alpha,\bar\alpha} * (\hat{u}^\alpha)^2 \big) v^\alpha}_{\op} \;.
\end{equation} 
With $\sum_\alpha (v^\alpha)^2 \leq 1$, we obtain, for any one-particle $\eta \in L^2 (\Lambda_L)$,  
\[  \begin{split} \sum_{\alpha=1}^M \langle \eta, v^\alpha \big( \hat\varphi_{Y,\bar\alpha,\bar\alpha} * (\hat{u}^\alpha)^2 \big) v^\alpha \eta \rangle &\leq \sum_{\alpha=1}^M \| v^\alpha \eta \|^2  \|  \hat\varphi_{Y,\bar\alpha,\bar\alpha} * (\hat{u}^\alpha)^2 \|_\infty \\ &\leq \max_{1 \le \alpha \le M} \norm{\hat\varphi_{Y,\bar\alpha,\bar\alpha}}_1 \norm{\hat{u}^\alpha}_{\infty}^2 \| \eta \|^2 \leq\max_{1 \le \alpha \le M} \norm{\hat\varphi_{Y,\bar\alpha,\bar\alpha}}_1  \| \eta \|^2 \end{split} \;. \]
From the assumption (\ref{eq:regularityestimate_2}), we find $\| \sum_{\alpha=1}^M  v^\alpha \big( \hat\varphi_{Y_X, \bar\alpha,\bar\alpha} * (\hat{u}^\alpha)^2 \big) v^\alpha \|_\text{op} \leq D_X$ and thus, with Prop. \ref{prop:gronwall}, 
\begin{equation}
	|\langle \xi_\lambda, \cE_\cE^{(u)}(Y_X) \xi_\lambda \rangle|
	\le C D_X e^{C \fl^{-7} } \rho^{-\frac 23} L^3 \;.
\end{equation}
Bounding $ \cE_\cE^{(v)}(Y_X) $ analogously, we conclude (\ref{eq:error-bos}). 
\end{proof}

\subsection{Proof of Proposition~\ref{prop:evaluate_Q1Q2}}

\begin{proof}
Recall $ \xi_\lambda = e^{-\lambda S} \Omega $. We iterate the Duhamel formula as in~\eqref{eq:Duhamel_1} with $ \cN $ replaced by $ Q_1(X) $. Using $ \langle \Omega, Q_1(X) \Omega \rangle = 0 $ and Lemma~\ref{lem:bosonizationerror}, recalling the symmetry $ X = X^T = \fR(X) $, and the same for $ I $ and $ K $, the first expansion reads
\begin{equation} \label{eq:Duhamel_1bis}
\begin{aligned}
	\langle \xi_\lambda, Q_1(X) \xi_\lambda \rangle
	&= \int_0^\lambda \d \lambda_1 \langle \Omega, e^{\lambda_1 S} [S,Q_1(X)] e^{-\lambda_1 S} \Omega \rangle \\
	&= \int_0^\lambda \d \lambda_1 \langle \xi_{\lambda_1}, \cE_{Q_1}(X) \xi_{\lambda_1} \rangle 
		+ \int_0^\lambda \d \lambda_1 \langle \xi_{\lambda_1}, Q_2(KIX + XIK) \xi_{\lambda_1} \rangle \;.
\end{aligned}
\end{equation}
We retain the bosonization error $ \cE_{Q_1}(X) $ and only expand the $ Q_2 $-contribution. Recall the notation $ \tilde{X}^{(n)} $~\eqref{eq:Xtilde} with $ \tilde{X}^{(0)} = X $ and $ \tilde{X}^{(n+1)} = K I \tilde{X}^{(n)} + \tilde{X}^{(n)} I K $. 
Then, introducing the integral over the $ n $-dimensional rescaled simplex $ \Delta_\lambda^n $
\begin{equation}
	\int_{\Delta_\lambda^n} \d^n \ulambda
	:= \int_0^\lambda \d \lambda_1 \int_0^{\lambda_1} \d \lambda_2 \ldots \int_0^{\lambda_{n-1}} \d \lambda_n \;, \qquad
	\ulambda := (\lambda_1, \ldots, \lambda_n) \;,
\end{equation}
and using cyclicity of the trace and Lemma~\ref{lem:bosonizationerror}, the second order Duhamel expansion gives 
\begin{equation} \label{eq:Duhamel_2}
\begin{aligned}
	\langle \xi_\lambda, Q_1(X) \xi_\lambda \rangle
	&= \frac{\lambda^2}{2} \sum_{k \in \Lambda_L} \tr \big( I(k) \tilde{X}^{(2)}(k) \big)
		+ \int_{\Delta_\lambda^1} \d^1 \ulambda \langle \xi_{\lambda_1}, \cE_{Q_1}(\tilde{X}^{(0)}) \xi_{\lambda_1} \rangle \\
	&\quad + \int_{\Delta_\lambda^2} \d^2 \ulambda \langle \xi_{\lambda_2}, \cE_{Q_2}(\tilde{X}^{(1)}) \xi_{\lambda_2} \rangle
		+ \int_{\Delta_\lambda^2} \d^2 \ulambda \langle \xi_{\lambda_2}, Q_1( \tilde{X}^{(2)} ) \xi_{\lambda_2} \rangle \;.
\end{aligned}
\end{equation}
In further iterations, we retain the error terms $ \cE_{Q_1}, \cE_{Q_2} $ and only expand $ Q_1, Q_2 $. After an odd number $ n_* $ of iterations, the result is then
\begin{equation} \label{eq:Duhamel_2n}
\begin{aligned}
	\langle \xi_\lambda, Q_1(X) \xi_\lambda \rangle
	&= \sum_{\substack{n=2 \\ n : \textnormal{even}}}^{n_*} \frac{\lambda^n}{n!} \sum_{k \in \Lambda_L} \tr \big( I(k) \tilde{X}^{(n)}(k) \big)
		+ \sum_{n=1}^{n_*} \int_{\Delta_\lambda^n} \d^n \ulambda \langle \xi_{\lambda_n}, \cE_{Q_{\sigma(n)}}(\tilde{X}^{(n-1)}) \xi_{\lambda_n} \rangle \\
	& \quad + \int_{\Delta_\lambda^{n_*}} \d^{n_*} \ulambda \langle \xi_{\lambda_{n_*}}, Q_2( \tilde{X}^{(n_*)} ) \xi_{\lambda_{n_*}} \rangle \;,
\end{aligned}
\end{equation}
where we used $ \int_{\Delta_\lambda^n} 1 \; \d^n \ulambda = \frac{\lambda^n}{n!} $ and where
\begin{equation}
	\sigma(n) := \begin{cases}
		1 \quad &\textnormal{for } n : \textnormal{odd} \;, \\
		2 \quad &\textnormal{for } n : \textnormal{even} \;.
	\end{cases}
\end{equation}
The trace term, by cyclicity of the trace and $ \sum_{n'=0}^n {n \choose n'} = 2^n $, agrees with the trace term in~\eqref{eq:evaluate_Q1Q2}. So it remains to bound the two error terms with $ \cE_{Q_{\sigma(n)}} $ and $ Q_2( \tilde{X}^{(n_*)} ) $.

\medskip

From \eqref{eq:regularityestimate_Q1Q2}, we conclude that (\ref{eq:regularityestimate_2}) holds true for $X = \tilde{X}^{(n-1)}$, with $D_X$ replaced by $C^n D_X $. Hence, Lemma \ref{lem:bosonizationerror} implies that 
\begin{equation} \label{eq:bound_EQ_N}
\begin{aligned}
	&\sum_{n=1}^{n_*} \int_{\Delta_\lambda^n} \d^n \ulambda |\langle \xi_{\lambda_n}, \cE_{Q_{\sigma(n)}}(\tilde{X}^{(n-1)}) \xi_{\lambda_n} \rangle|
	\le \sum_{n=1}^{\infty} \frac{\lambda^n}{n!} 
		\sup_{\lambda \in [0,1]}
		|\langle \xi_\lambda, \cE_{Q_{\sigma(n)}}(\tilde{X}^{(n-1)}) \xi_\lambda \rangle| \\
	&\le C D_X \sum_{n=1}^{\infty} \frac{C^n}{n!}
		\big( \fl^{-23} M^{\frac 12} e^{C \fl^{-17} M^{\frac 12}} \rho^{-\frac 13}
		+ e^{C \fl^{-7}} \rho^{-\frac 23} \big) L^3
	\le C_\varepsilon D_X \rho^{-\frac 13 + \varepsilon} L^3 \;,
\end{aligned}
\end{equation}
uniformly in $ n_* \in \NNN $, where we used that $ \fl^{-23} M^{\frac 12} \ll M \sim \log(\rho^\eta) $, so given any $ \varepsilon > 0 $, for $ \eta > 0 $ small enough we have $ \fl^{-23} M^{\frac 12} e^{C \fl^{-17} M^{\frac 12}} \ll \rho^{\varepsilon} $.

\medskip

Finally, the $ Q_2 $-term is controlled by Lemma~\ref{lem:Q2_bound}: With~\eqref{eq:Q1Q2position} and the triangle inequality, we obtain
\begin{equation} \label{eq:Q2_bound_tail}
	|\langle \xi_{\lambda_{n_*}}, Q_2( \tilde{X}^{(n_*)} ) \xi_{\lambda_{n_*}} \rangle|
	\le \bigg\vert \frac{2}{\rho^{\frac 23}} \sum_{\alpha, \beta = 1}^M \int \d x \d y \varphi_{\tilde{X}^{(n_*)}}(y-x)_{\alpha, \bar\beta}
		\langle \xi_{\lambda_{n_*}}, b_{\alpha,x} b_{\beta,y} \xi_{\lambda_{n_*}} \rangle \bigg\vert \;,
\end{equation}
which is twice the left-hand side of~\eqref{eq:Q2_bound} with $ \varphi_{X,\alpha,\beta} = \varphi_{\tilde{X}^{(n_*)},\alpha,\bar\beta} $ and $ \xi = \xi_{\lambda_{n_*}} $ plugged in. By~\eqref{eq:regularityestimate_Q1Q2}, $ \tilde{X}^{(n_*)} $ satisfies the regularity estimate~\eqref{eq:regularityestimate_1}, where $ D_X $ of~\eqref{eq:regularityestimate_1} gets replaced by $ C^{n_*} D_X $ of~\eqref{eq:regularityestimate_Q1Q2}. So Lemma~\ref{lem:Q2_bound} applies, and we obtain
\begin{equation}
	\int_{\Delta_\lambda^{n_*}} \d^{n_*} \ulambda |\langle \xi_{\lambda_{n_*}}, Q_2( \tilde{X}^{(n_*)} ) \xi_{\lambda_{n_*}} \rangle|
	\le \frac{C^{n_*}}{n_*!} D_X \fl^{-2}
		\sup_{\lambda \in [0,1]} \langle \xi_\lambda, (\cN + L^3) \xi_\lambda \rangle \;.
\end{equation}
Bounding the number operator by Proposition~\ref{prop:gronwall}, we conclude
\begin{equation} \label{eq:bound_Q_N}
	\int_{\Delta_\lambda^{n_*}} \d^{n_*} \ulambda |\langle \xi_{\lambda_{n_*}}, Q_2( \tilde{X}^{(n_*)} ) \xi_{\lambda_{n_*}} \rangle| \\
	\le \frac{C^{n_*}}{n_*!} D_X \fl^{-2} e^{C \fl^{-7}} L^3 \;.
\end{equation}
Since the bound~\eqref{eq:bound_EQ_N} is uniform in $ n_* \in \NNN $, we can choose $ n_* $ as large as we like, rendering the r.\,h.\,s.\ of~\eqref{eq:bound_Q_N} as small as we wish. This concludes the evaluation of $ \langle \xi_\lambda, Q_1(X) \xi_\lambda \rangle $ in~\eqref{eq:evaluate_Q1Q2}.

The evaluation of $ \langle \xi_\lambda, Q_2(X) \xi_\lambda \rangle $ is analogous: After an even number of iterations, we arrive at an expansion like~\eqref{eq:Duhamel_2n}, with the only difference that $ \cE_{Q_1} $ and $ \cE_{Q_2} $ are swapped, and the trace term is summed over odd $ n \ge 1 $.
\end{proof}

\section{Reduction to the Effective Hamiltonian}
\label{sec:reduction_Heff}

With our improved bound on $ \cN $ from Proposition~\ref{prop:improved_excitation}, we can proceed with estimating the errors of the approximation
\[
	R^* \cH_N R - E_{\FB}
	\simeq \QQQ_{\B}^{\cR} + \DDD_{\B} \;,
\]
compare~\eqref{eq:RHNR} and~\eqref{eq:Bosonizedoperator} in the trial state $ \xi = e^{-S} \Omega $, compare~\eqref{eq:T}. For this, the following observation will be useful: Since $ S $ only excites pairs in the patches $ B_\alpha $, we can ignore annihilation of particles outside patches:
\begin{equation} \label{eq:ignore_outside_patches}
	[\hat{a}_p,S] = 0 \qquad \Rightarrow \qquad
	\hat{a}_p e^{-S} \Omega
	= e^{-S} \hat{a}_p \Omega
	= 0
	\qquad \textnormal{if } p \notin B_\alpha \; \forall \alpha \in \{1, \ldots, M\} \;.
\end{equation} 
We can therefore restrict our attention to excitations in (unregularized) patches, which we describe by the functions
\begin{equation} \label{eq:u0v0}
\begin{aligned}
	\hat{u}_{(0)}^\alpha(p)
	&:= \chi(p \in \BF^c \cap B_\alpha) \;, \qquad
	&\hat{v}_{(0)}^\alpha(h)
	&:= \chi(h \in \BF \cap B_\alpha) \;, \\
	\hat{u}_{(0)}(p)
	&:= \sum_{\alpha = 1}^M \hat{u}_{(0)}^\alpha(p) \;, \qquad 
	&\hat{v}_{(0)}(h)
	&:= \sum_{\alpha = 1}^M \hat{v}_{(0)}^\alpha(h) \;.
\end{aligned}
\end{equation}
Let us also denote the sums of the regularized functions $ \hat{u}^\alpha $, $ \hat{v}^\alpha $~\eqref{eq:uhat_vhat} over all patches by
\begin{equation} \label{eq:uv_allpatches}
	\hat{u}(p)
	:= \sum_{\alpha = 1}^M \hat{u}^\alpha(p) \;, \qquad 
	\hat{v}(h)
	:= \sum_{\alpha = 1}^M \hat{v}^\alpha(h) \;,
\end{equation}
and the regularization errors by 
\begin{equation} \label{eq:uv_edge}
	\hat{u}_{(\edge)}(p)
	:= \hat{u}_{(0)}(p) - \hat{u}(p) \;, \qquad 
	\hat{v}_{(\edge)}(h)
	:= \hat{v}_{(0)}(h) - \hat{v}(h) \;.
\end{equation}
Note that $ \hat{u}_{(\edge)} $ and $ \hat{v}_{(\edge)} $ are supported close to the boundary of $ \bigcup_\alpha (B_\alpha \cap \overline{B_{k_{\F}}(0)}^c) $ and $ \bigcup_\alpha (B_\alpha \cap \overline{B_{k_{\F}}(0)}) $ in a region of thickness $ \fl $, hence the name ``edge''. Also, $ \hat{u} $ is smooth, while $ \hat{u}_{(0)} $ and $ \hat{u}_{(\edge)} $ are not. Furthermore, for any function $ \hat{f} $ introduced in~\eqref{eq:u0v0},~\eqref{eq:uv_allpatches},~\eqref{eq:uv_edge}, we denote again by $ f $ its inverse Fourier transform as in~\eqref{eq:uxvx}, and by $ f $ the integral operator with kernel $ f(x-y) = f_y(x) $, which is a momentum space multiplication by $ \hat{f} $.\\

\subsection{Bounding Non-Bosonizable Terms}

We start with bounding the non-bosonizable terms in~\eqref{eq:RHNR}. Recall that they have to be compared to a leading term of order $ \hbar L^3 $.

\begin{lemma}[Non-bosonizable terms] \label{lem:nonbosonizable}
Within the trial state $ \xi = T \Omega $, compare~\eqref{eq:T}, the terms $ \cE_1 $, $ \cE_2 $, and $ \XXX $ in~\eqref{eq:RHNR} are bounded as
\begin{equation} \label{eq:nonbosonizable}
	\langle \xi, \cE_2 \xi \rangle
	= 0 \;, \qquad
	|\langle \xi, \XXX \xi \rangle|
	\le C \rho^{-\frac 23} \hbar L^3 \;, \qquad
	|\langle \xi, \cE_1 \xi \rangle|
	\le C \fl^{\frac 12} \hbar L^3 \;.
\end{equation}
\end{lemma}

\begin{proof}
First, note that $ S $ according to \eqref{eq:T} generates and annihilates two pairs of particle--hole excitations, so $ T $ contains an even number of such excitations. Since $ \cE_2 $ changes this number by 1, we immediately conclude $ \langle \xi, \cE_2 \xi \rangle = 0 $.

\medskip

To bound $ \XXX $, note that on the last line of~\eqref{eq:RHNR} we have $ \sum_{p \in \BF^c \cap (\BF + k) \cap \Lambda_L} \hat{a}_p^* \hat{a}_p \le \sum_{p \in \Lambda_L} \hat{a}_p^* \hat{a}_p = \cN $, and the same is true for the sum over $ h $. Thus, with Proposition~\ref{prop:improved_excitation},
\[
	|\langle \xi, \XXX \xi \rangle|
	\le \frac{C}{\rho L^3} \sum_{k \in \Lambda_L} \hat{V}(k)
		|\langle \xi, \cN \xi \rangle|
	\le C \rho^{-1} \Vert \hat{V} \Vert_\infty L^3 \;,
\]
which concludes the bound on $ \XXX $ with $ \hbar = \rho^{-1/3} $.

\medskip

It remains to bound $ \cE_1 $. Plugging~\eqref{eq:fDk} into~\eqref{eq:RHNR} and applying Cauchy-Schwarz, we obtain 
\begin{equation} \label{eq:cE1_fourcontributions}
\begin{aligned}
	|\langle \xi, \cE_1 \xi \rangle|
	&\le \frac{1}{\rho L^3} \sum_{k \in \Lambda_L} \hat{V}(k) \sum_{p, p' \in \BF^c \cap (\BF^c + k)} \langle \xi, \hat{a}_{p'}^* \hat{a}_{p'-k} \hat{a}_{p-k}^* \hat{a}_p \xi \rangle \\
	&\quad + \frac{1}{\rho L^3} \sum_{k \in \Lambda_L} \hat{V}(k) \sum_{h, h' \in \BF \cap (\BF - k)} \langle \xi, \hat{a}_{-h'}^* \hat{a}_{-h'-k} \hat{a}_{-h-k}^* \hat{a}_{-h} \xi \rangle \\
\end{aligned}
\end{equation}
We show how to bound the first contribution, the other works analogously. First, note that in $ \xi = e^{-S} \Omega $, $ S $ only excites pairs within patches, see~\eqref{eq:ignore_outside_patches}, so for $ p \in \BF^c $,
\begin{equation} \label{eq:ignore_outside_patches_2}
	\hat{a}_p \xi
	= \sum_{\alpha = 1}^M \chi(p \in B_\alpha) \hat{a}_p \xi
	= \hat{u}_{(0)}(p) \hat{a}_p \xi \;,
\end{equation}
compare definition~\eqref{eq:u0v0} of $ \hat{u}_{(0)} $. This immediately allows restricting $ p',p $ to patches. To also restrict $ p'-k $ and $ p-k $, we apply the CAR and then convert in position space: 
\begin{equation} \label{eq:cE1_firstcontribution}
\begin{aligned}
	&\frac{1}{\rho L^3} \Bigg\vert \sum_{k \in \Lambda_L} \hat{V}(k) \sum_{p, p' \in \BF^c \cap (\BF^c + k)} \langle \xi, \hat{a}_{p'}^* \hat{a}_{p'-k} \hat{a}_{p-k}^* \hat{a}_p \xi \rangle \Bigg\vert \\
	&\le \frac{1}{\rho L^3} \Bigg\vert \sum_{k \in \Lambda_L} \hat{V}(k) \sum_{p, p' \in \BF^c} \langle \xi, \hat{a}_{p'}^* \hat{a}_{p-k}^* \hat{a}_{p'-k} \hat{a}_p \xi \rangle \Bigg\vert
		+ \frac{1}{\rho L^3} \Bigg\vert \sum_{k \in \Lambda_L} \hat{V}(k) \sum_{p \in \BF^c} \langle \xi, \hat{a}_p^* \hat{a}_p \xi  \rangle \Bigg\vert \\
	&\le \frac{1}{\rho L^3} \Bigg\vert \sum_{k \in \Lambda_L} \hat{V}(k) \sum_{p, p' \in \Lambda_L}
		\hat{u}_{(0)}(p') \hat{u}_{(0)}(p'-k) \hat{u}_{(0)}(p) \hat{u}_{(0)}(p-k)
		\langle \xi, \hat{a}_{p'}^* \hat{a}_{p-k}^* \hat{a}_{p'-k} \hat{a}_p \xi \rangle \Bigg\vert \\
	&\quad + \rho^{-1} \Vert \hat{V} \Vert_1 \langle \xi, \cN \xi \rangle \\
	&\le \rho^{-1} \Big\vert \int \d x \d y \; V(x-y) \langle \xi, a^*(u_{(0),x}) a^*(u_{(0),y}) a(u_{(0),x}) a(u_{(0),y}) \xi \rangle \Big\vert 
	+ C \rho^{-\frac 23} \hbar L^3\;,
\end{aligned}
\end{equation}
where in the last step we wrote the first contribution in position space and bounded $ \langle \xi, \cN \xi \rangle $ by Proposition~\eqref{eq:improved_excitation}. To bound the first contribution, we split $ u_{(0),x} = u_{(\edge),x} + u_x  $, compare~\eqref{eq:uv_edge}, resulting in four terms:
\begin{equation} \label{eq:cE1_firstcontribution_split}
\begin{aligned}
	&\rho^{-1} \Big\vert \int \d x \d y \; V(x-y) \langle \xi, a^*(u_{(0),x}) a^*(u_{(0),y}) a(u_{(0),x}) a(u_{(0),y}) \xi \rangle \Big\vert 
	\le \I + \II + \III + \IV \;, \\
	\I &:= \rho^{-1} \Big\vert \int \d x \d y \; V(x-y) \langle \xi, a^*(u_{(\edge),x}) a^*(u_{(0),y}) a(u_{(\edge),x}) a(u_{(0),y}) \xi \rangle \Big\vert \;, \\
	\II &:= \rho^{-1} \Big\vert \int \d x \d y \; V(x-y) \langle \xi, a^*(u_x) a^*(u_{(0),y}) a(u_{(\edge),x}) a(u_{(0),y}) \xi \rangle \Big\vert \;, \\
	\III &:= \rho^{-1} \Big\vert \int \d x \d y \; V(x-y) \langle \xi, a^*(u_{(\edge),x}) a^*(u_{(0),y}) a(u_x) a(u_{(0),y}) \xi \rangle \Big\vert \;, \\
	\IV &:= \rho^{-1} \Big\vert \int \d x \d y \; V(x-y) \langle \xi, a^*(u_x) a^*(u_{(0),y}) a(u_x) a(u_{(0),y}) \xi \rangle\Big\vert \;.
\end{aligned}
\end{equation}
The terms $ \I $, $ \II $, and $ \III $ are bounded the same way, using that $ \hat{u}_{(\edge)} $ is supported in a region of thickness $ \fl $ around the patch boundaries, so $ \Vert \hat{u}_{(\edge)} \Vert_2 \le C \rho^{\frac 13} \fl^{\frac 12} $, which is stronger than $ \Vert \hat{u} \Vert_2 \le C \rho^{\frac 13} $, compare~\eqref{eq:uv_edge}. E.g., for $ \II $, recalling Lemma~\ref{lem:FT_bounds},
\begin{equation} \label{eq:cE1_edgeterm_final}
\begin{aligned}
	\II &\le \rho^{-1} \int \d x \d y \; |V(x-y)|
		\Vert a(u_x) a(u_{(0),y}) \xi \Vert
		\Vert a(u_{(\edge),x}) a(u_{(0),y}) \xi \Vert \\
	&\le \rho^{-1} \int \d x \d y \; |V(x-y)|
		\Vert u_x \Vert_2
		\Vert u_{(\edge),x} \Vert_2
		\Vert a(u_{(0),y}) \xi \Vert^2 \\
	&\le C \fl^{\frac 12} \rho^{-\frac 13} \Vert V \Vert_1 \langle \xi, \cN \xi \rangle 
	\le C \fl^{\frac 12} \hbar L^3 \;,
\end{aligned}
\end{equation}
where we used
\[
	\int \d y \langle \xi, a^*(u_{(0),y}) a(u_{(0),y}) \xi \rangle
	= \sum_{p \in \Lambda_L} |\hat{u}_{(0)}(p)|^2  \langle \xi, \hat{a}_p^* \hat{a}_p \xi \rangle
	\le \langle \xi, \cN \xi \rangle \;.
\]
The terms $ \I $ and $ \III $ are bounded analogously.
As for $\IV$, we write 
\begin{equation}\label{eq:IV1} \IV = \rho^{-1} \Big| \int \d y \, \langle a (u_{(0),y}) \xi, \d \Gamma (uV_y u) a (u_{(0),y}) \xi \rangle \Big| \;. \end{equation}
We observe that (the following bound would not be correct if we replaced $u$ by $u_{(0)}$, because $\hat{u}_{(0)}$ is not continuous; this is why we first had to handle the $u_\text{(edge)}$ separately):
\begin{equation}
\begin{aligned}
	\Vert \cW_y^{-\frac 12} u V_y u \cW_y^{-\frac 12} \Vert_{\op}
	&\le \Vert u \cW_y^{-\frac 12} V_y \cW_y^{-\frac 12} u \Vert_{\op}
		+ \Vert [u, \cW_y^{-\frac 12}] V_y \cW_y^{-\frac 12} u \Vert_{\op} \\
	&\quad + \Vert u \cW_y^{-\frac 12} V_y [\cW_y^{-\frac 12}, u] \Vert_{\op}
		+ \Vert [u, \cW_y^{-\frac 12}] V_y [\cW_y^{-\frac 12}, u] \Vert_{\op} \\
	&\le \big( \Vert [u, \cW_y^{-\frac 12}] \cW_y^{\frac 12} \Vert_{\op}
		+ \Vert \hat{u} \Vert_{\infty} \big)^2
		\Vert \cW_y^{-\frac 12} V_y \cW_y^{-\frac 12} \Vert_{\infty} \;.
\end{aligned}
\end{equation}
From $ \Vert \hat{u} \Vert_{\infty} = 1 $, $ \Vert \cW_y^{-\frac 12} V_y \cW_y^{-\frac 12} \Vert_{\infty} \le C $ and since, proceeding as in 
Lemma~\ref{lem:ualpha_derivative_bound},
\[ \Vert [\cW_y^{-\frac 12}, u] \cW_y^{\frac 12} \Vert_{\op} \le C \fl^{-2}, \] we find 
\[
	\Vert \cW_y^{-\frac 12} u V_y u \cW_y^{-\frac 12} \Vert_{\op}
	\le C \fl^{-4} \;,
\]
and therefore $\pm \d \Gamma (uV_y u) \leq C \fl^{-4} \mathcal{N}_y$, where we recall that $\mathcal{N}_y =  \d \Gamma (\cW_y)$. Inserting in (\ref{eq:IV1}), we obtain 
\[ \begin{split} \IV &\leq C \fl^{-4} \rho^{-1} \int \d y \, \langle \xi, a^* (u_{(0),y}) \cN_y a (u_{(0),y}) \xi \rangle \\ &= C \fl^{-4} \rho^{-1} \int \d y \, \langle \cN_y \xi, a^* (u_{(0),y})  a (u_{(0),y}) \xi \rangle
 - C \fl^{-4} \rho^{-1} \int \d y \, \langle a (\cW_y u_{(0),y}) \xi , a (u_{(0),y}) \xi \rangle \\  &\leq C \fl^{-4} \rho^{-1} \int \d y \, \| \cN_y \xi \| \| a^* (u_{(0),y})  a (u_{(0),y}) \xi \| 
 + C \fl^{-4} \rho^{-1} \| \widehat{\cW u_{(0)}} \|_\infty \| \hat{u}_{(0)} \|_\infty \langle \xi, \cN \xi \rangle \;,
 \end{split} \] 
where in the second term, the operators $ a(\cW_y u_{(0),y}) $ and $ a(u_{(0),y}) $ were converted into $ \d \Gamma $-operators and estimated against $ \cN $ as in~\eqref{eq:trNzreduction_2}.
With Prop. \ref{prop:gronwall} and Prop. \ref{prop:improved_excitation} and estimating $\| \cW u_{(0)} \|_1 \leq \| \cW \|_2 \| u_{(0)} \|_2 \leq C \rho^{1/3}$, we find 
\[ \IV \leq C \fl^{-4} e^{C \fl^{-17} M^{\frac 12}} \rho^{-1} L^{\frac 32} \left( \int \d y \, \| a^* (u_{(0),y})  a (u_{(0),y}) \xi \|^2 \right)^{1/2} + C \fl^{-4} \rho^{-\frac 23} L^3  \;. \]
With $ \Vert a^*(u_{(0),y}) \Vert_{\op} \le \Vert u_{(0),y} \Vert_2 = \Vert \hat{u}_{(0)} \Vert_2 \le C \rho^{\frac 13} $ and Proposition~\ref{prop:improved_excitation}, we can also bound
\begin{equation} \label{eq:cE1_secondbracket}
	\int \d y \Vert a^*(u_{(0),y}) a(u_{(0),y}) \xi \Vert^2
	\le C \rho^{\frac 23} \int \d y \Vert a(u_{(0),y}) \xi \Vert^2
	\le C \rho^{\frac 23} \langle \xi, \cN \xi \rangle
	\le C \rho^{\frac 23} L^3 \;.
\end{equation}
We conclude that 
\[ \IV \leq C \fl^{-4} e^{C \fl^{-17} M^{\frac 12}} \rho^{-\frac 23} L^3 \;.  \]

\end{proof}

\subsection{Restriction to Patches}

Next, we bound the difference $ \QQQ_{\B} - \QQQ_{\B}^{\cR} $. For this, it will be useful to introduce the unregularized and regularized bosonized creation operators acting on all patches:
\begin{equation} \label{eq:b0b}
\begin{aligned}
	b^*_{(0)}(k) 
	&:= \sum_{p \in \Lambda_L} \hat{u}_{(0)}(p) \hat{v}_{(0)}(p-k) \hat{a}_p^* \hat{a}_{-p+k}^* \;, \\
	b^*(k) 
	&:= \sum_{p \in \Lambda_L} \hat{u}(p) \hat{v}(p-k) \hat{a}_p^* \hat{a}_{-p+k}^* 
	= \sum_{\alpha = 1}^M \tilde{n}_{\alpha,k} b^*_\alpha(k)\;,
\end{aligned}
\end{equation}
where we recall that $ \hat{u}_{(0)} $, $ \hat{v}_{(0)} $ were defined in~\eqref{eq:u0v0} and $ \hat{u} $, $ \hat{v} $ in~\eqref{eq:uv_allpatches}. We will further use the following bound.

\begin{lemma}[Bounds on Bosonized Interacting Hamiltonian] \label{lem:bbstar_bound}
For $ \xi \in \cF $, we have
\begin{equation} \label{eq:bbstar_bound}
	\frac{1}{L^3} \sum_{k \in \Lambda_L} \hat{V}(k) \Vert b(k) \xi \Vert^2
	\le C \rho^{\frac 23} \langle \xi, \cN \xi \rangle \;, \quad
	\frac{1}{L^3} \sum_{k \in \Lambda_L} \hat{V}(k) \Vert b^*(k) \xi \Vert^2
	\le C \rho^{\frac 23} \langle \xi, (\cN + L^3) \xi \rangle \;,
\end{equation}
and the same is true if $ b^\sharp(k) $ is replaced by $ b_{(0)}^\sharp(k) $.
\end{lemma}

\begin{proof}
Writing in position space, and using the Cauchy--Schwarz inequality and the bound $ \Vert a(u_x) \Vert_{\op} \le \Vert u_x \Vert_2 = \Vert \hat{u} \Vert_2 \le C \rho^{\frac 13} $, we get
\begin{equation}
\begin{aligned}
	&\frac{1}{L^3} \sum_{k \in \Lambda_L} \hat{V}(k) \Vert b(k) \xi \Vert^2
	= \int \d x \d y \; V(x-y) \langle \xi, a^*(u_x) a^*(v_x) a(v_y) a(u_y) \xi \rangle \\
	&\le \int \d x \d y \; |V(x-y)| \Vert a(v_x) a(u_x) \xi \Vert \Vert a(v_y) a(u_y) \xi \Vert \\
	&\le C \rho^{\frac 23} \int \d x \d y \; |V(x-y)| (\Vert a(u_x) \xi \Vert^2 + \Vert a(u_y) \xi \Vert^2)
	\le C \rho^{\frac 23} \Vert V \Vert_1 \langle \xi, \cN \xi \rangle \;,
\end{aligned}
\end{equation}
which, with $ \Vert V \Vert_1 \le C $ gives the first desired bound. For the second bound, we note that
\[
	\Vert b^*(k) \xi \Vert^2
	= \Vert b(k) \xi \Vert^2
		+ \langle \xi, [b(k),b^*(k)] \xi \rangle \;.
\]
The sum over $ \Vert b(k) \xi \Vert^2 $ was already bounded above. Using the CAR, the commutator amounts to
\[
	[b(k),b^*(k)]
	= \sum_{p \in \Lambda_L} \hat{u}(p)^2 \hat{v}(p-k)^2 (- \hat{a}^*_{-p+k} \hat{a}_{-p+k} - \hat{a}^*_p \hat{a}_p + 1)
	\le \sum_{p \in \Lambda_L} \hat{u}(p)^2
	\le C \rho^{\frac 23} L^3 \;.
\]
We conclude
\begin{equation}
	\frac{1}{L^3} \sum_{k \in \Lambda_L} \hat{V}(k) \Vert b^*(k) \xi \Vert^2
	\le C \rho^{\frac 23} \langle \xi, \cN \xi \rangle
		+ \frac{C}{L^3} \sum_{k \in \Lambda_L} \hat{V}(k) \rho^{\frac 23} L^3 \Vert \xi \Vert^2
	\le C \rho^{\frac 23} \langle \xi, (\cN + L^3) \xi \rangle \;.
\end{equation}
The proof with $ b^\sharp(k) $ replaced by $ b_{(0)}^\sharp(k) $ is analogous.
\end{proof}

The bound on the patch restriction error is the following.

\begin{lemma}[Reduction to Pair Operators on Patches] \label{lem:QBQBR_bound}
Within the trial state $ \xi = T \Omega $, compare~\eqref{eq:T}, the difference between $ \QQQ_{\B} $~\eqref{eq:RHNR} and $ \QQQ_{\B}^{\cR} $~\eqref{eq:QQQBcR} is bounded as
\begin{equation} \label{eq:QBQBR_bound}
	|\langle \xi, (\QQQ_{\B} - \QQQ_{\B}^{\cR}) \xi \rangle |
	\le C \fl^{\frac 12} \hbar L^3 \;.
\end{equation}
\end{lemma}

\begin{proof}
Recall the definition~\eqref{eq:RHNR} of $ \QQQ_{\B} $ with $ B^*(k) $ given in~\eqref{eq:bstar_full}. Using~\eqref{eq:ignore_outside_patches}, we may restrict all momenta to patches, that is, replacing $ B^\sharp(k) $ by $ b^\sharp_{(0)}(k) $, compare~\eqref{eq:bstar_full} and~\eqref{eq:b0b}:
\begin{equation}
	\langle \xi, \QQQ_{\B} \xi \rangle
	= \frac{1}{2 \rho L^3} \sum_{k \in \Lambda_L} \hat{V}(k)
	\langle \xi, \Big( 2 b^*_{(0)}(k) b_{(0)}(k)
		+ b^*_{(0)}(k) b^*_{(0)}(-k)
		+ b_{(0)}(-k) b_{(0)}(k)\Big) \xi \rangle \;.
\end{equation}
By comparison, with $ b^\sharp_\alpha(k) = \hat\chi_{>,\alpha}(k) b^\sharp_\alpha(k) $ according to~\eqref{eq:bstar_chi_insertion} and $ \sum_{\alpha=1}^M \tilde{n}_{\alpha,k} b_\alpha^\sharp(k) = b^\sharp(k) $ according to~\eqref{eq:b0b}, we may write $ \QQQ_{\B}^{\cR} $ from~\eqref{eq:QQQBcR} as
\begin{equation}
	\QQQ_{\B}^{\cR}
	= \frac{1}{2 \rho L^3} \sum_{k \in \Lambda_L} \hat{V}(k)
		\Big( 2 b^*(k) b(k)
		+ b^*(k) b^*(-k)
		+ b(-k) b(k)\Big) \;.
\end{equation}
So, introducing
\begin{equation} \label{eq:r}
	r^*(k) 
	:= b^*_{(0)}(k) - b^*(k)  \;,
\end{equation}
we obtain
\begin{equation}
\begin{aligned}
	\langle \xi, (\QQQ_{\B} - \QQQ_{\B}^{\cR}) \xi \rangle
	& = \frac{1}{2 \rho L^3} \sum_{k \in \Lambda_L} \hat{V}(k)
		\langle \xi, \Big( 2 r^*(k) b(k) + 2 b_{(0)}^*(k) r(k)\\
		&\quad + r^*(k) b^*(-k) + b_{(0)}^*(k) r^*(-k)
		+ r(-k) b(k) + b_{(0)}(-k) r(k) \Big) \xi \rangle \;.
\end{aligned}
\end{equation}
With the Cauchy--Schwarz inequality, and replacing $ -k \mapsto k $, we then get
\begin{equation} \label{eq:QB_estimate_CS}
\begin{aligned}
	&|\langle \xi, (\QQQ_{\B} - \QQQ_{\B}^{\cR}) \xi \rangle| \\
	&\le \frac{C}{\rho}
		\Bigg( \frac{1}{L^3} \sum_{k \in \Lambda_L}
		\hat{V}(k) \Vert r(k) \xi \Vert^2 \Bigg)^{\frac 12} \\
	&\quad \times \Bigg( \frac{1}{L^3} \sum_{k \in \Lambda_L} \hat{V}(k)
		\big( \Vert b(k) \xi \Vert^2
		+ \Vert b^*(k) \xi \Vert^2
		+ \Vert b_{(0)}(k) \xi \Vert^2
		+ \Vert b_{(0)}^*(k) \xi \Vert^2
		+ \Vert r(k) \xi \Vert^2 \big) \Bigg)^{\frac 12} \;.
\end{aligned}
\end{equation}
To bound the first bracket, we observe that with~\eqref{eq:b0b} and~\eqref{eq:uv_edge}, we can write $ r(k) =r_1(k) + r_2(k) + r_3(k) $, where
\begin{equation}
\begin{aligned}
	 r_1(k)
	 &:= \sum_{p \in \Lambda_L} \hat{u}_{(\edge)}(p) \hat{v}(p-k) \hat{a}_{-p+k} \hat{a}_p \;, \\
	 r_2(k)
	 &:= \sum_{p \in \Lambda_L} \hat{u}(p) \hat{v}_{(\edge)}(p-k) \hat{a}_{-p+k} \hat{a}_p \;, \\
	 r_3(k)
	 &:= \sum_{p \in \Lambda_L} \hat{u}_{(\edge)}(p) \hat{v}_{(\edge)}(p-k) \hat{a}_{-p+k} \hat{a}_p \;.
\end{aligned}
\end{equation}
Correspondingly, $ \Vert r(k) \xi \Vert^2 \le 3 \big( \Vert r_1(k) \xi \Vert^2 + \Vert r_2(k) \xi \Vert^2 + \Vert r_3(k) \xi \Vert^2 \big) $. We then bound in position space, as in the proof of Lemma~\ref{lem:bbstar_bound} and using $ \Vert u_{(\edge),x} \Vert_2^2 = \Vert \hat{u}_{(\edge)} \Vert_2^2 \le C \fl \rho^{\frac 23} $:
\begin{equation}
\begin{aligned}
	&\frac{1}{L^3} \sum_{k \in \Lambda_L}
		\hat{V}(k) \Vert r_1(k) \xi \Vert^2 
	= \int \d x \d y \; V(x-y) \langle a(v_x) a(u_{(\edge),x}) \xi, a(v_y) a(u_{(\edge),y}) \xi \rangle \\
	&\le 	\int \d x \d y |V(x-y)|
		\Vert u_{(\edge),x} \Vert_2 \Vert u_{(\edge),y} \Vert_2 \big( \Vert a(v_x) \xi \Vert^2 + \Vert a(v_y) \xi \Vert^2 \big) 
	\le C \fl \rho^{\frac 23} \langle \xi, \cN \xi \rangle \;,
\end{aligned}
\end{equation}
where we used $ \Vert V \Vert_1 \le C $. The bounds with $ r_2(k) $  and $ r_3(k) $ are analogous, so the first bracket in~\eqref{eq:QB_estimate_CS} is bounded as
\[
	\frac{1}{L^3} \sum_{k \in \Lambda_L}
		\hat{V}(k) \Vert r(k) \xi \Vert^2
	\le C \fl \rho^{\frac 23} \langle \xi, \cN \xi \rangle \;.
\]
Within the second bracket, the sums over $ \Vert b(k) \xi \Vert^2 $, $ \Vert b^*(k) \xi \Vert^2 $, $ \Vert b_{(0)}(k) \xi \Vert^2 $, and $ \Vert b_{(0)}^*(k) \xi \Vert^2 $ are controlled by Lemma~\ref{lem:bbstar_bound}, resulting in
\begin{equation}
	|\langle \xi, (\QQQ_{\B} - \QQQ_{\B}^{\cR}) \xi \rangle|
	\le C \rho^{-1} \fl^{\frac 12} \rho^{\frac 23} \langle \xi, (\cN + L^3) \xi \rangle
	= C \fl^{\frac 12} \hbar \langle \xi, (\cN + L^3) \xi \rangle \;.
\end{equation}
Bounding $ \langle \xi, \cN \xi \rangle \le C L^3 $ by Proposition~\ref{prop:improved_excitation} then concludes the proof.
\end{proof}

\subsection{Linearization of the Kinetic Energy}

In this section, we bound the difference $ \HHH_0 - \DDD_B $, compare~\eqref{eq:RHNR} and~\eqref{eq:DB}. This operator difference is generally large, but small when taking commutators with bosonized operators.\\

\begin{lemma}[Bosonization of Kinetic Energy] \label{lem:H0DB_bound}
Within the trial state $ \xi = T \Omega $, see~\eqref{eq:T}, we have
\begin{equation} \label{eq:H0DB_bound}
	|\langle \xi, (\HHH_0 - \DDD_{\B}) \xi \rangle |
	\le C \fl^{\frac 12} \hbar L^3 \;.
\end{equation}
\end{lemma}

\begin{proof}
Using the fundamental theorem of calculus and $ \xi_\lambda := e^{-\lambda S} \Omega $,
\begin{equation}
\begin{aligned}
	&|\langle \xi, (\HHH_0 - \DDD_{\B}) \xi \rangle|
	= |\langle \Omega, e^S (\HHH_0 - \DDD_{\B}) e^{-S} \Omega \rangle| \\
	&\le |\langle \Omega, (\HHH_0 - \DDD_{\B}) \Omega \rangle|
		+ \int_0^1 \d \lambda |\partial_\lambda \langle \Omega, e^{\lambda S} (\HHH_0 - \DDD_{\B}) e^{-\lambda S} \Omega \rangle| \\
	&= \int_0^1 \d \lambda |\langle \Omega, e^{\lambda S} [S, \HHH_0 - \DDD_{\B}] e^{-\lambda S} \Omega \rangle|
	\le \sup_{\lambda \in [0,1]} |\langle \xi_\lambda, [S,\HHH_0 - \DDD_{\B}] \xi_\lambda \rangle| \;,
\end{aligned}
\end{equation}
where we used that $ \hat{a}_p \Omega = 0 $ implies $ (\HHH_0 - \DDD_{\B}) \Omega = 0 $. We will compare both $ [S,\HHH_0] $ and $ [S,\DDD_{\B}] $ to their bosonized approximation: First, recalling the definition of $ h(k) $~\eqref{eq:hP} and $ Q_1(A) $~\eqref{eq:Q1Q2}, we write $ \DDD_{\B} $~\eqref{eq:DB} as
\begin{equation}
	\DDD_{\B} = \hbar^2 k_{\F} Q_1(h) \;.
\end{equation}
Then, by means of Lemma~\ref{lem:bosonizationerror}, with $ Y := K I h + h I K $:
\begin{equation}
	[S,\DDD_{\B}]
	= \hbar^2 k_{\F} Q_2(Y)
		+ \hbar^2 k_{\F} \cE_{Q_1}(h) \;.
\end{equation}
Here, $ \hbar^2 k_{\F} Q_2(Y) $ is the bosonized approximation. The bosonization error $ \cE_{Q_1}(h) $ can be bounded by Lemma~\ref{lem:bosonizationerror},
using that $ X = h $ satisfies the regularity estimate~\eqref{eq:regularityestimate_2} with $ D_X = C M $ by means of Proposition~\ref{prop:regularity}:
\begin{equation}
\begin{aligned}
	|\langle \xi_\lambda, ([S,\DDD_{\B}] - \hbar^2 k_{\F} Q_2(Y)) \xi_\lambda \rangle|
	&\le \hbar^2 k_{\F} |\langle \xi_\lambda, \cE_{Q_1}(h) \xi_\lambda \rangle|
	\\ & \le C \hbar M^{\frac 32} \fl^{-23} e^{C \fl^{-17} M^{\frac 12}} \rho^{-\frac 13} L^3
	\le C_\varepsilon \rho^{-\frac 13 + \varepsilon} \hbar L^3 \;,
\end{aligned}
\end{equation}
where we used in the last line that $ \fl^{-23} M^{\frac 12} \ll M \sim \log(\rho^\eta) $ with $ \eta $ arbitrarily small. The other commutator can be computed with~\eqref{eq:S} as
\begin{equation}
\begin{aligned}
	[S,\HHH_0]
	&= \sum_{k \in \Lambda_L} \sum_{\alpha, \beta = 1}^M
		\frac{K(k)_{\alpha, \bar\beta}}{\tilde{n}_{\alpha, k} \tilde{n}_{\bar\beta, k}}
		\sum_{p, q \in \Lambda_L} \frac {\hbar^2}{2} (|p|^2 - |p-k|^2 + |q|^2 - |q+k|^2) \\
	&\quad \times \hat{u}^\alpha(p) \hat{v}^\alpha(p-k) \hat{u}^\beta(q) \hat{v}^\beta(q+k) 
		\hat{a}_{-p+k} \hat{a}_p \hat{a}_{-q-k} \hat{a}_q + \mathrm{h.c.}
\end{aligned}
\end{equation}
By comparison, the bosonized approximation with $ I $ replaced by $ 1 $ would be
\begin{equation}
\begin{aligned}
	\hbar^2 k_{\F} Q_2(Kh+hK)
	&= \sum_{k \in \Lambda_L} \sum_{\alpha, \beta = 1}^M
		\frac{K(k)_{\alpha, \bar\beta}}{\tilde{n}_{\alpha, k} \tilde{n}_{\bar\beta, k}}
		\sum_{p, q \in \Lambda_L} \hbar^2 k_{\F} (\tilde\lambda_{\alpha,k} + \tilde\lambda_{\bar\beta,k}) \\
	&\quad \times \hat{u}^\alpha(p) \hat{v}^\alpha(p-k) \hat{u}^\beta(q) \hat{v}^\beta(q+k) 
		\hat{a}_{-p+k} \hat{a}_p \hat{a}_{-q-k} \hat{a}_q + \mathrm{h.c.}
\end{aligned}
\end{equation}
Recalling $ \hat{u}^\alpha, \hat{v}^\alpha $~\eqref{eq:uhat_vhat}, we only obtain nonzero contributions for $ |p| - |p-k| \ge 2 \fl $. Since $ |p| - |p-k| = k \cdot \hat{p} + \cO(\rho^{-\frac 13}) = k \cdot \hat\omega_\alpha + \cO(M^{-\frac 12}) $, we get $ k \cdot \hat\omega_\alpha \ge \fl $ for large enough $ \rho $, so $ \tilde\lambda_{\alpha, k} = k \cdot \hat\omega_\alpha $, and likewise $ \tilde\lambda_{\bar\beta, k} = - k \cdot \hat\omega_\beta $.
We then write the difference of the two kinetic energy functions as follows:
\begin{equation}
	\frac{|p|^2 - |p-k|^2}{2} - k \cdot (k_{\F} \hat\omega_\alpha)
	= k \cdot p - \frac{|k|^2}{2} - k \cdot (k_{\F} \hat\omega_\alpha)
	= k \cdot (p - k_{\F} \hat\omega_\alpha)
		- \frac{|k|^2}{2} \;.
\end{equation}
Then, denoting $ K_2(k) := |k|^2 K(k) $, and for $ j \in \{1,2,3\} $, $ K_{1,j}(k) := k_j K(k) $ and $ \hat{u}^\alpha_{(\err),j}(p) := \hat{u}^\alpha(p) \hbar (p_j - k_{\F} \hat\omega_{\alpha,j}) $, and writing in position space as in Section~\ref{sec:positionspace}:
\begin{equation}
\begin{aligned}
	&|\langle \xi_\lambda, ([S,\HHH_0] - \hbar^2 k_{\F} Q_2(Y)) \xi_\lambda \rangle|
	\le 2 \I_k + 2 \I_{kp} + \I_Y \;, \\
	\I_k
	&:= \hbar^2 \Bigg\vert \frac{1}{\rho^{\frac 23}} \sum_{\alpha,\beta = 1}^M \int \d x \d y \; \varphi_{K_2}(y-x)_{\alpha,\bar\beta} \langle \xi_\lambda, a(v_x^{\bar\alpha}) a(u_x^\alpha) a(v_y^{\bar\beta}) a(u_y^\beta) \xi_\lambda \rangle \Bigg\vert \;, \\
	\I_{kp}
	&:= \hbar \sum_{j=1}^3 \Bigg\vert \frac{1}{\rho^{\frac 23}} \sum_{\alpha,\beta = 1}^M \int \d x \d y \; \varphi_{K_{1,j}}(y-x)_{\alpha,\bar\beta} \langle \xi_\lambda, a(v_x^{\bar\alpha}) a(u_{(\err),j,x}^\alpha) a(v_y^{\bar\beta}) a(u_y^\beta) \xi_\lambda \rangle + \mathrm{c.c.} \Bigg\vert \;, \\
		\I_Y
	&:= \hbar^2 k_{\F} |\langle \xi_\lambda, Q_2(Kh+hK-Y) \xi_\lambda \rangle | \;.
\end{aligned}
\end{equation}
Here, $ K_2(k) $ and $ K_{1,j}(k) $ satisfy the regularity estimate~\eqref{eq:regularityestimate_1} with $ D_X = C \fl^{-4} |\log\fl| $ by Proposition~\ref{prop:regularity}. Thus, we can bound $ \I_k $ by Lemma~\ref{lem:Q2_bound}, using $ \langle \xi_\lambda, \cN \xi_\lambda \rangle \le C L^3 $ by means of Proposition~\ref{prop:improved_excitation}:
\begin{equation}
	\I_k
	\le C \hbar^2 \fl^{-6} |\log\fl| L^3
	\le C_\varepsilon \rho^{-\frac 13 + \varepsilon} \hbar L^3 \;.
\end{equation}
The bound of $ \I_{kp} $ is analogous up to the following modification in the proof of Lemma~\ref{lem:Q2_bound}: We replace $ u_x^\alpha $ by $ u_{(\err),j,x}^\alpha $, which is supported in the same patch but satisfies the improved bound $ \Vert \hat{u}_{(\err),j}^\alpha \Vert_\infty \le C M^{-\frac 12} $. This is since $ (p - k_{\F} \hat\omega_\alpha) $ is the distance to the middle point of patch $ B_\alpha $, which is bounded by the patch diameter $ \le C \rho^{\frac 13} M^{-\frac 12} $. Exchanging the role of the variables $ x $ and $ y $ in the proof of Lemma~\ref{lem:Q2_bound}, we may eliminate $ u_x^\alpha $ through the bound $ 0 \le \sum_{\alpha = 1}^M (u^\alpha)^2 \le 1 $. The replacement of $ u_x^\alpha $ by $ u_{(\err),j,x}^\alpha $ improves this bound to $ 0 \le \sum_{\alpha = 1}^M (u_{(\err),j}^\alpha)^2 \le C M^{-1} $, yielding a final bound
\begin{equation}
	\I_{kp}
	\le C \hbar \fl^{-6} |\log\fl| M^{-\frac 12} L^3
	\le C \fl^{\frac 12} \hbar L^3 \;.
\end{equation}
Finally, to bound $ \I_Y $, we write
\begin{equation}
	\I_Y
	= \hbar^2 k_{\F} |\langle \xi_\lambda, Q_2(K(1-I)h+h(1-I)K) \xi_\lambda \rangle |
\end{equation}
and apply Proposition~\ref{prop:evaluate_Q1Q2} with $ X = K(1-I)h+h(1-I)K $, where Proposition~\ref{prop:regularity} implies the regularity estimate~\eqref{eq:regularityestimate_Q1Q2} with $ D_X = C \fl^{-4} |\log \fl| $:
\begin{equation}
\begin{aligned}
	\I_Y
	&\le \hbar^2 k_\F \bigg\lvert \sum_{\substack{n=1\\n:\text{odd}}}^\infty \frac{(2 \lambda)^n}{n!} \sum_{k \in \Lambda_L} \tr \Big( I(k) K(k) (1-I)(k) h(k)
		(I(k) K(k))^n \Big) \bigg\rvert \\
	&\quad + \hbar^2 k_\F \bigg\lvert \sum_{\substack{n=1\\n:\text{odd}}}^\infty \frac{(2 \lambda)^n}{n!} \sum_{k \in \Lambda_L} \tr \Big( I(k) h(k) (1-I)(k) K(k)
		(I(k) K(k))^n \Big) \bigg\rvert
	+ C_\varepsilon D_X \rho^{-\frac 13 + \varepsilon} \hbar L^3 \\
	&\le C \fl |\log \fl| \hbar L^3 \;,
\end{aligned}
\end{equation}
where in the last step we estimated the trace by Proposition~\ref{prop:trace_estimates}.
\end{proof}

\section{Proof of Theorem~\ref{thm:main}}
\label{sec:proofmain}

\begin{proof}

Recall the trial state~\eqref{eq:trialstate}
\begin{equation*}
	\psi := R T \Omega \;,
\end{equation*}
with $ \Omega $ being the Fock vacuum, $ R $ the particle--hole transformation \eqref{eq:R} and $ T $ the pseudo-bosonic Bogoliubov transformation \eqref{eq:T}. By the variational principle, recalling~\eqref{eq:RHNR}, with $ \xi = T \Omega $, we get
\begin{equation}
\begin{aligned}
	E_{\GS}
	& \le \langle \xi, R^* H_N R \xi \rangle \\
	&\le E_{\FB}
		+ \langle \xi, (\DDD_{\B} + \QQQ_{\B}^{\cR}) \xi \rangle
		+ |\langle \xi, (\QQQ_{\B} - \QQQ_{\B}^{\cR}) \xi \rangle| \\
		&\quad
		+ |\langle \xi, (\HHH_0 - \DDD_{\B}) \xi \rangle|
		+ |\langle \xi, (\cE_1 + \cE_2 + \XXX) \xi \rangle| \;,
\end{aligned}
\end{equation}
with $ \QQQ_{\B}^{\cR} $ and $ \DDD_{\B} $ given in~\eqref{eq:QQQBcR} and~\eqref{eq:DB}, respectively. Estimating $ (\cE_1 + \cE_2 + \XXX) $ by Lemma~\ref{lem:nonbosonizable}, $ (\QQQ_{\B} - \QQQ_{\B}^{\cR}) $ by Lemma~\ref{lem:QBQBR_bound} and $ (\HHH_0 - \DDD_{\B}) $ by Lemma~\ref{lem:H0DB_bound}, we get
\begin{equation}
	E_{\GS}
	\le  E_{\FB}
		+ \langle \xi, (\DDD_{\B} + \QQQ_{\B}^{\cR}) \xi \rangle
		+ C \fl^{\frac 12} \hbar L^3 \;.
\end{equation}
Next, we note that in the notation~\eqref{eq:Q1Q2}, the effective Hamiltonian $ \DDD_{\B} + \QQQ_{\B}^{\cR} $ from~\eqref{eq:Bosonizedoperator} can be written as
\begin{equation}
	\DDD_{\B} + \QQQ_{\B}^{\cR}
	= \hbar^2 k_{\F} (Q_1(h+P) + Q_2(P)) \;.
\end{equation}
By definition, $ X \in \{h,P\} $ satisfies $ X(k)_{\alpha,\beta} = X(k)_{\beta,\alpha} = X(-k)_{\bar\alpha,\bar\beta} $ and is supported for $ |k| \le 2 R $, and by Proposition~\ref{prop:regularity}, $ X \in \{h,P\} $ satisfies the regularity estimate~\eqref{eq:regularityestimate_Q1Q2} with $ D_X = M \fl^{-4} |\log \fl| $. So Proposition~\ref{prop:evaluate_Q1Q2} applies, and with $ I(0) = 0 $ according to~\eqref{eq:cEalphahat}, we obtain
\begin{equation}
\begin{aligned}
	\langle \xi, (\DDD_{\B} + \QQQ_{\B}^{\cR}) \xi \rangle
	&\le \hbar^2 k_{\F} \sum_{\substack{n=2 \\ n : \textnormal{even}}}^\infty \frac{2^n}{n!} \sum_{k \in \Lambda_L \setminus \{0\}} \tr \Big( I(k) (h+P)(k) \big( I(k) K(k) \big)^n \Big) \\
	&\quad + \hbar^2 k_{\F} \sum_{\substack{n=1 \\ n : \textnormal{odd}}}^\infty \frac{2^n}{n!} \sum_{k \in \Lambda_L \setminus \{0\}} \tr \Big( I(k) P(k) \big( I(k) K(k) \big)^n \Big)
		+ C_\varepsilon \rho^{-\frac 13 + \varepsilon} \hbar L^3
\end{aligned}
\end{equation}
Next, we replace the approximate identities $ I(k) $ by true identities using Proposition~\ref{prop:trace_estimates}:
\begin{equation}
\begin{aligned}
	&\langle \xi, (\DDD_{\B} + \QQQ_{\B}^{\cR}) \xi \rangle \\
	&\le \hbar^2 k_{\F} \sum_{\substack{n=2 \\ n : \textnormal{even}}}^\infty \frac{2^n}{n!} \sum_{k \in \Lambda_L \setminus \{0\}} \tr \big( (h+P)(k) K(k)^n \big) \\
	& \quad + \hbar^2 k_{\F} \sum_{\substack{n=1 \\ n : \textnormal{odd}}}^\infty \frac{2^n}{n!} \sum_{k \in \Lambda_L \setminus \{0\}} \tr \big( P(k)K(k)^n \big)
		+ C \fl |\log \fl| \hbar L^3 \\
	&= \hbar^2 k_{\F} \sum_{k \in \Lambda_L \setminus \{0\}} \tr \big( (h+P)(k) (\cosh(2K(k)) - 1)
		+ P(k) \sinh(2K(k)) \big)
		+ C \fl |\log \fl| \hbar L^3 \;.
\end{aligned}
\end{equation}
To further evaluate the trace, we introduce the matrix-valued function $ E: \Lambda_L \to \RRR^{M \times M} $, $ E := e^K (h+2P) e^K $, where with~\cite[Proposition~3.5]{CHN21}, using $ A = h+P $ and $ B = P $, we conclude
\begin{equation} \label{eq:E}
	E = e^K (h+2P) e^K = e^{-K} h e^{-K} \;,
\end{equation}
and find, with cyclicity of the trace and omitting the $ (k) $ for convenience:
\begin{equation*}
\begin{aligned}
	&\tr \big( (h+P) (\cosh(2K) - 1) + P \sinh(2K) \big)
	= \frac 12 \tr \big( h (e^{2K} + e^{-2K}) + 2P e^{2K} \big)
		- \tr (h+P) \\
	&= \frac 12 \tr \big( e^K (h+2P) e^K + e^{-K} h e^{-K} \big)
		- \tr (h+P)
	= \tr (E - h - P) \;.
\end{aligned}
\end{equation*}
So up to this point, we have proven
\begin{equation} \label{eq:EGS_evaluation_intermediate}
	E_{\GS}
	\le E_{\FB} + \hbar^2 k_{\F} \sum_{k \in \Lambda_L \setminus \{0\}} \tr \big( E(k) - h(k) - P(k) \big)
	+ C \fl^{\frac 12} \hbar L^3 \;.
\end{equation}
To finish the proof, it remains to extract the RPA energy from the trace. First, by the same algebraic manipulations as in~\cite[Prop.~7.6]{CHN21}, we obtain for $ k \in \supp(\hat V) $
\begin{equation}
	\tr \big( E(k) - h(k) - P(k) \big)
	= \frac{1}{\pi} \int_0^\infty \log \Bigg( 1 + 2 \sum_{\alpha=1}^M \frac{\tilde\lambda_{\alpha,k}} {\tilde\lambda_{\alpha,k}^2 + \mu^2} \tilde{v}_{\alpha,k}^2 \Bigg) \d \mu
		- \sum_{\alpha=1}^M \tilde{v}_{\alpha,k}^2 \;,
\end{equation}
and with the identity
\[
	1 = \frac{2}{\pi} \int_0^\infty \frac{\tilde\lambda_{\alpha,k}} {\tilde\lambda_{\alpha,k}^2 + \mu^2} \d \mu \;,
\]
with $ F(x) = \log(1+x) - x $, and after a change of variables $ \mu \mapsto |k| \mu $, we obtain
\begin{equation} \label{eq:trEhP_integralformula}
	\tr \big( E(k) - h(k) - P(k) \big)
	= \frac{|k|}{\pi} \int_0^\infty F(Q_k(\mu)) \d \mu \;, \quad
	Q_k(\mu)
	:= 2 \sum_{\alpha=1}^M \frac{|k|^{-2} \tilde\lambda_{\alpha,k}} {|k|^{-2} \tilde\lambda_{\alpha,k}^2 + \mu^2} \tilde{v}_{\alpha,k}^2 \;,
\end{equation}
which already agrees with the expression in~\eqref{eq:Ecorr} up to $ Q^{(0)}_k(\mu) $ being replaced by $ Q_k(\mu) $. To extract $ Q^{(0)}_k(\mu) $, we first plug in~\eqref{eq:hP} and~\eqref{eq:ntilde}, yielding
\begin{equation} \label{eq:Qmu_bis}
	Q_k(\mu)
	= \frac{4 \pi}{M} \sum_{\alpha=1}^M \frac{|k|^{-2} \tilde\lambda_{\alpha,k}^2} {|k|^{-2} \tilde\lambda_{\alpha,k}^2 + \mu^2} \frac{\hat V(k) k_{\F}}{(2 \pi)^3 \rho^{\frac 13}} \hat\chi_{>,\alpha}(k)^2 \;.
\end{equation}
The function $ \hat\chi_{>,\alpha} $ restricts the sum to patches $ \alpha $ inside $ \cI_k := \{\alpha \in \{1,\ldots,M\} : k \cdot \hat\omega_\alpha \ge -2 \fl \} $. We split this summation region into $ \cI_{k,\ge} := \{\alpha \in \{1,\ldots,M\} : k \cdot \hat\omega_\alpha \ge \fl \} $ and $ \cI_{k,<} := \{\alpha \in \{1,\ldots,M\} : -2 \fl \le k \cdot \hat\omega_\alpha < \fl \} $. Defining the angle $ \theta_{\alpha,k} \in [0,\pi] $ such that $ \cos \theta_{\alpha,k} = k \cdot \hat\omega_\alpha |k|^{-1} $, we notice that the set $ \cI_{k,<} $ comprises patches with associated angles $ \theta_{\alpha,k} \in [\pi/2 - C \fl |k|^{-1}, \pi/2 + C \fl |k|^{-1}] $ for $ C > 0 $ large enough. Since each patch covers an angle of $ M^{-\frac 12} \ll \fl $ on the sphere, we conclude
\[
	|\cI_{k,<}| \le C M |k|^{-1} \fl \;.
\]
On the other hand, for $ \alpha \in \cI_{k,\ge} $, we have $ \tilde\lambda_{\alpha,k} |k|^{-1} = \cos \theta_{\alpha,k} $ and $ \hat\chi_{>,\alpha}(k) = 1 $, so
\begin{equation}
\begin{aligned}
	&\bigg\lvert Q_k(\mu)
		- \frac{4 \pi}{M}
		\sum_{\alpha \in \cI_{k,\ge}} \frac{\cos^2 \theta_{\alpha,k}} {\cos^2 \theta_{\alpha,k} + \mu^2}
		\frac{\hat V(k) k_{\F}}{(2 \pi)^3 \rho^{\frac 13}} \bigg\rvert \\
	&\le \frac{4 \pi}{M} \sum_{\alpha \in \cI_{k,<}} \frac{|k|^{-2} \tilde\lambda_{\alpha,k}^2} {|k|^{-2} \tilde\lambda_{\alpha,k}^2 + \mu^2} \frac{\hat V(k) k_{\F}}{(2 \pi)^3 \rho^{\frac 13}} \hat\chi_{>,\alpha}(k)^2
	\le C M^{-1} |\cI_{k,<}| 
	\le C |k|^{-1} \fl \;.
\end{aligned}
\end{equation}
The above sum over $ \alpha \in \cI_{k,\ge} $ is evaluated by a Riemann sum argument as in~\cite[(5.15)--(5.18)]{BNPSS20}, using $ \cos \theta_{\alpha,k} \ge |k|^{-1} \fl $ with the result
\begin{equation}
\begin{aligned}
	\bigg\lvert \frac{4 \pi}{M}
		\sum_{\alpha \in \cI_{k,\ge}} \frac{\cos^2 \theta_{\alpha,k}} {\cos^2 \theta_{\alpha,k} + \mu^2}
		- 2 \pi \big( 1 - \mu \arctan(\mu^{-1}) \big) \bigg\rvert
	\le C \big( |k| \fl^{-1} M^{-\frac 12}
		+ |k|^{-1} \fl \big)
	\le  C |k|^{-1} \fl \;.
\end{aligned}
\end{equation}
Comparing with $ Q^{(0)}_k(\mu) $ from~\eqref{eq:Ecorr}, we obtain
\begin{equation}
\begin{aligned}
	\big\vert Q_k(\mu) - Q^{(0)}_k(\mu) \big\vert
	&\le C |k|^{-1} \fl \;.
\end{aligned}
\end{equation}
For $ x \ge 0 $ we have $ F'(x) \in [-1,0] $, so $ F $ is Lipschitz continuous with Lipschitz constant $1$, and we obtain
\begin{equation} \label{eq:FQ_diff_bound}
\begin{aligned}
	\big\vert F(Q_k(\mu)) - F(Q^{(0)}_k(\mu)) \big\vert
	&\le C |k|^{-1} \fl \;.
\end{aligned}
\end{equation}
This is a strong bound to be used for small $ \mu $.
For large $ \mu $, we need decay with respect to $ \mu $, which we obtain by estimating ~\eqref{eq:Qmu_bis} with $ \tilde\lambda_{\alpha,k} \le C $:
\begin{equation}
	|F(Q_k(\mu))|
	\le |Q_k(\mu)|
	\le C M^{-1} \sum_{\alpha \in \cI_k} \frac{|k|^{-2} \tilde\lambda_{\alpha,k}^2} {\mu^2}
	\le C |k|^{-2} \mu^{-2} \;.
\end{equation}
Furthermore, by~\eqref{eq:Ecorr} with Taylor expansion of $ \arctan(x) $ around $ x = 0 $,
\begin{equation}
	|F(Q^{(0)}_k(\mu))|
	\le |Q^{(0)}_k(\mu)|
	\le C \mu^{-2} 
	\le C |k|^{-2} \mu^{-2} \;.
\end{equation}
Therefore, splitting between ``small'' and ``large'' $ \mu $ at some scale $ \mu_* > 0 $,
\begin{equation}
\begin{aligned}
	&\int_0^\infty \big\vert F(Q_k(\mu)) - F(Q^{(0)}_k(\mu)) \big\vert \d \mu
	\le C \int_0^{\mu_*} |k|^{-1} \fl \; \d \mu
		+ \int_{\mu_*}^\infty |k|^{-2} \mu^{-2} \d \mu \\
	&\le C |k|^{-1} \fl \mu_* + C |k|^{-2} \mu_*^{-1}
	\le C |k|^{-\frac 32} \fl^{\frac 12} \;,
\end{aligned}
\end{equation}
where we optimized $ \mu_* = |k|^{-\frac 12} \fl^{-\frac 12} $. Comparing~\eqref{eq:Ecorr} and~\eqref{eq:trEhP_integralformula}, where $ \frac{(2 \pi)^3}{L^3} \sum_k $ is replaced by $ \int \d k $ through a Riemann sum argument, up to an error of relative size $ o_L(1) $, we observe
\begin{equation}
\begin{aligned}
	&\Bigg\vert \hbar^2 k_{\F} \sum_{k \in \Lambda_L \setminus \{0\}} \tr \big( E(k) - h(k) - P(k) \big)
		- e_{\RPA} L^3 \Bigg\vert \\
	&\le C \hbar \sum_{k \in \Lambda_L \setminus \{0\}} |k| 
		\int_0^\infty \big\vert F(Q_k(\mu)) - F(Q^{(0)}_k(\mu)) \big\vert \d \mu 
		+ o_L(1) \\
	&\le C \fl^{\frac 12} \hbar
		\sum_{k \in \Lambda_L \setminus \{0\} \cap \supp(\hat{V})} |k|^{-\frac 12} 
		+ o_L(1)
	\le C \fl^{\frac 12} \hbar L^3 \;,
\end{aligned}
\end{equation}
where we used that both $ F(Q_k(\mu)) $ and $ F(Q^{(0)}_k(\mu)) $ vanish if $ \hat{V}(k) = 0 $.
Plugging this into~\eqref{eq:EGS_evaluation_intermediate}, we obtain $ E_{\GS} \le E_{\FB} + e_{\RPA} L^3 + C \fl^{\frac 12} \hbar L^3 $, and dividing by $ L^3 $ concludes the proof.
\end{proof}

\appendix

\section{Regularity Estimates for Matrix-Valued Functions}
\label{sec:regularity_K}

The goal of this section is to derive estimates of derivatives on momentum space quantities, ensuring their decay in position space.
Recall that $ -\Delta_k $ is the discrete Laplacian.

\begin{proposition}[Regularity estimates]
\label{prop:regularity}
Let $ V $ satisfy Assumption \ref{as:V}. Recall the definition~\eqref{eq:varphi} of $ \hat\varphi_{X,\alpha,\beta} $ for $ X: \Lambda_L \to \RRR^{M \times M} $. Then, with $ X \in \{K, h, P\} $, see~\eqref{eq:K} and~\eqref{eq:hP}, given $ m \in \NNN_0 $ even, there exist constants $ C_{V,m} > 0 $, $ C_m > 0 $ such that for all $ \alpha,\beta \in \{1,\ldots,M\} $,
\begin{equation} \label{eq:regularity_KhP_final}
\begin{aligned}
	\Vert (1 - \Delta_k)^{\frac m2} \hat\varphi_{K,\alpha,\beta} \Vert_1
	&\le C_{V,m} \fl^{-m} |\log \fl| \;, \\
	\Vert (1 - \Delta_k)^{\frac m2} \hat\varphi_{h,\alpha,\beta} \Vert_1
	&\le C_m M \;, \\
	\Vert (1 - \Delta_k)^{\frac m2} \hat\varphi_{P,\alpha,\beta} \Vert_1
	&\le C_{V,m} \fl^{-m} \;.
\end{aligned}
\end{equation}
Furthermore, if $ K_j(k) := p_j(k) K(k) $ with $ p_j $ being a polynomial of degree $ j $, then there exists $ C_{V,m,p} > 0 $ such that
\begin{equation} \label{eq:regularity_Kj_final}
	\Vert (1 - \Delta_k)^{\frac m2} \hat\varphi_{K_j,\alpha,\beta} \Vert_1
	\le C_{V,m,p} \fl^{-m} |\log \fl| \qquad \forall \alpha,\beta \in \{1,\ldots,M\}\;.
\end{equation}
Moreover, for $ \tilde{X}^{(n)}$ defined as in~\eqref{eq:Xtilde} for $ n \in \NNN_0 $, we have
\begin{equation} \label{eq:regularity_KhP_tilde_final}
\begin{aligned}
		\Vert (1 - \Delta_k)^{\frac m2} \hat\varphi_{\tilde{X}^{(n)},\alpha,\beta} \Vert_1
	&\le (C_{V,m})^{n+1} \fl^{-m} |\log \fl| \;, \qquad
	&&X \in \{K, P, h(1-I)K, K(1-I)h \} \;, \\
	\Vert (1 - \Delta_k)^{\frac m2} \hat\varphi_{\tilde{h}^{(n)},\alpha,\beta} \Vert_1
	&\le (C_{V,m})^{n+1} \fl^{-m} |\log \fl| \qquad
	&&\mathrm{if} \; n \ge 1 \;, \\
	\Vert (1 - \Delta_k)^{\frac m2} \hat\varphi_{\tilde{h}^{(0)},\alpha,\beta} \Vert_1
	&\le C_m M\;.
\end{aligned}
\end{equation}
\end{proposition}
Proposition~\ref{prop:regularity} will be proven in Appendix~\ref{subsec:proof_regularitybounds}.

We will also need bounds on certain trace terms, as well as on the error from replacing the approximate identity $ I(k) $ introduced in~\eqref{eq:cEalphahat} by the identity matrix.

\begin{proposition}[Trace estimates]
\label{prop:trace_estimates}
Let $ V $ satisfy Assumption~\ref{as:V} and recall $ K $ from~\eqref{eq:K}. Then, there exists $ C_V > 0 $ such that
\begin{equation} \label{eq:trace_estimates}
	\sum_{k \in \Lambda_L} \Big\vert \tr \Big( \big( I(k) K(k) \big)^{n+1} \Big) \Big\vert
	\le (C_V)^n L^3 \qquad \forall n \in \NNN \;.
\end{equation}
Moreover, for $ X \in \{h, P\}$ (see \eqref{eq:hP}), we have the estimate
\begin{equation} \label{eq:trace_replacements}
\begin{aligned}
	\sum_{k \in \Lambda_L} \Big\vert \tr \Big( I(k) X(k) \big( I(k) K(k) \big)^n \Big)
		- \tr \big( X(k) K(k)^n \big)
	\Big\vert
	&\le C_V C^n \fl |\log \fl| L^3 \quad 
	&&\forall n \in \NNN \;.
\end{aligned}
\end{equation}
and for $ X \in \{h(1-I)K, K(1-I)h\} $, we have
\begin{equation} \label{eq:trace_estimates_X}
\begin{aligned}
	\sum_{k \in \Lambda_L} \Big\vert \tr \Big( I(k) X(k) \big( I(k) K(k) \big)^n \Big)
	\Big\vert
	&\le C_V C^n \fl |\log \fl| L^3 \quad 
	&&\forall n \in \NNN \;.
\end{aligned}
\end{equation}
\end{proposition}

We will first derive pointwise bounds for derivatives of $ K(k)$, $h(k)$, and $P(k) $, then proceed to $ \tilde{K}^{(n)}(k)$, $\tilde{h}^{(n)}(k)$, and $\tilde{P}^{(n)}(k) $, and finally sum over the momenta $ k \in \Lambda_L $. Proposition~\ref{prop:trace_estimates} will then be proven in Appendix~\ref{subsec:proof_trace_estimates}.

\subsection{Deriving Pointwise Bounds for $K$}
\label{subsec:K_Pointwise}

The proof strategy is as follows: We will first derive regularity estimates for the matrix-valued function
\begin{equation} \label{eq:B}
	B: \Lambda_L \to \RRR^{M \times M} \;, \qquad
	B(k) := (1-e^{2K(k)}) \;,
\end{equation}
since it has a particularly convenient integral representation: Following the same steps as around \cite[(7.27)]{CHN22}, we obtain
\begin{equation}
\label{eq:Bintegral}
	B(k)_{\alpha, \beta}
	= \frac{4}{\pi} \int_0^\infty
		\frac{\mu^2 \tilde\lambda_{\alpha,k}^{-1}
		(\tilde\lambda_{\alpha,k}^{-2} + \mu^2)^{-1}
		\tilde\lambda_{\beta,k}^{-1}
		(\tilde\lambda_{\beta,k}^{-2} + \mu^2)^{-1}
		\tilde{v}_{\alpha,k} \tilde{v}_{\beta,k}
		}
		{1 + 2 \mu^2 \sum_{\gamma = 1}^M \tilde\lambda_{\gamma,k}^{-1} (\tilde\lambda_{\gamma,k}^{-2} + \mu^2)^{-1} \tilde{v}_{\gamma,k}^2 }
		\d \mu \;,
\end{equation}
with $ \tilde\lambda_{\alpha,k} $ and $ \tilde{v}_{\alpha,k} $ defined in \eqref{eq:lambdatilde} and \eqref{eq:hP}. This will allow us to derive regularity estimates for $ K $ and quantities depending on $ K $.\\
To establish those regularity estimates, recall~\eqref{eq:Delta_k} that $ -\Delta_k = -\sum_j \bar\partial_{k_j} \partial_{k_j} $ with $ \bar\partial_{k_j}, \partial_{k_j} $ being the left/right discrete derivative, satisfying $ [\bar\partial_{k_j},\partial_{k_{j'}}] = 0 $. To track higher discrete derivatives, we use multi-indices $ \um = (m_1,m_2,m_3) \in \NNN_0^3 $, $ |\um| := m_1 + m_2 + m_3 $ and the notation
\begin{equation}
	\partial_k^{\um} \hat{f}(k)
	:= \partial_{k_1}^{m_1} \partial_{k_2}^{m_2} \partial_{k_3}^{m_3} \hat{f}(k) \;, \qquad
	\bar\partial_k^{\um} \hat{f}(k)
	:= \bar\partial_{k_1}^{m_1} \bar\partial_{k_2}^{m_2} \bar\partial_{k_3}^{m_3} \hat{f}(k) \;.
\end{equation}
We will also make use of repeated product and chain rules, which, for discrete derivatives, produce discretization errors $ O_L(L^{-1}) $.

\begin{lemma} [Repeated discrete product and chain rules] \label{lem:product_chain_rule}
Let $ \hat{f}, \hat{g}: \Lambda_L \to \CCC $. Then, for any multi-index $ \um \in \NNN_0^3 $ with $ |\um| = m $, there exists some $ C_m > 0 $ such that
\begin{equation} \label{eq:product_rule_repeated}
\begin{aligned}
	|\partial_k^{\um} (\hat{f} \hat{g})(k)|
	&\le C_m \sum_{\substack{\um_1, \um_2 \in \NNN_0^3 \\ m \le |\um_1 + \um_2| \le 2 m}}
		L^{m-|\um_1 + \um_2|} |\partial_k^{\um_1} \hat{f}(k)| |\partial_k^{\um_2} \hat{g}(k)| \;.
\end{aligned}
\end{equation}
Furthermore, if $ h \in C^\infty(\RRR) $, and $ \hat{f} $ extends to a smooth function $ \hat{f}: \RRR^3 \to \RRR $, then\footnote{Observe that $ k + k'_1 $ does not need to be on the lattice $ \Lambda_L $. We may nevertheless take discrete derivatives $ \partial_k^{\um_1} \hat{f}(k + k'_1) $ as defined in~\eqref{eq:partial_k}.\\ Also, we denote concatenation of functions by $ (h \circ \hat{f})(k) = h(\hat{f}(k)) $}
\begin{equation} \label{eq:chain_rule_repeated}
\begin{aligned}
	|\partial_k^{\um} (h \circ \hat{f})(k)|
	&\le C_m \sum_{\ell=0}^m \sum_{\substack{\um_1, \ldots, \um_\ell \in \NNN_0^3 \\ m \le |\um_1 + \ldots + \um_\ell| \le (\ell+1) m}} 
		L^{m - |\um_1 + \ldots + \um_\ell|} \\
		& \times \sup_{\substack{ k'_1, \ldots, k'_{\ell+1} \in \RRR^3 \\ |k'_1|, \ldots, |k'_{\ell+1}| \le 2 \pi \ell L^{-1}}}
		|\partial_k^{\um_1} \hat{f}(k + k'_1)| \cdots |\partial_k^{\um_\ell} \hat{f}(k + k'_\ell)|
		 |h^{(\ell)}(\hat{f}(k + k'_{\ell+1}))| \;,
\end{aligned}
\end{equation}
where $ h^{(\ell)} $ denotes the $ \ell $-th derivative of $ h $.
All bounds remain true if any number of right discrete derivatives $ \partial_{k_j} $ is replaced by left discrete derivatives $ \bar\partial_{k_j} $.
\end{lemma}

\begin{proof}
Recall the product rule~\eqref{eq:productrule}, which can be written as
\begin{equation} \label{eq:product_rule_bis}
	\partial_{k_j}(\hat{f} \hat{g})
	= (\partial_{k_j} \hat{f}) \hat{g}
		+ \hat{f} (\partial_{k_j} \hat{g})
		+ 2 \pi L^{-1} (\partial_{k_j} \hat{f}) (\partial_{k_j} \hat{g}) \;,
\end{equation}
i.\,e., we obtain the usual product rule for smooth functions, plus an error term that carries an excess derivative and a factor of $ L^{-1} $. Iterating this formula $ m $ times will render $ 3^m $ contributions with at most $ 2m $ derivatives, where each excess derivative comes with a factor of $ L^{-1} $. This establishes~\eqref{eq:product_rule_repeated}.\\

To prove~\eqref{eq:chain_rule_repeated}, first observe that we can write
\begin{equation}
\begin{aligned}
	\partial_{k_j}(h \circ \hat{f})(k)
	&= \frac{h(\hat{f}(k + 2 \pi L^{-1} e_j)) - h(\hat{f}(k))}{\hat{f}(k + 2 \pi L^{-1} e_j) - \hat{f}(k)}
		\frac{\hat{f}(k + 2 \pi L^{-1} e_j) - \hat{f}(k)}{2 \pi L^{-1}} \\
	&= (\partial_{k_j} \hat{f})(k) \int_0^1 h'(f^{[t]}_j(k)) \d t
\end{aligned}
\end{equation}
with interpolating function
\begin{equation}
	f^{[t]}_j(k) := t \hat{f}(k + 2 \pi L^{-1} e_j) + (1-t) \hat{f}(k) \;.
\end{equation}
For higher derivatives, we iterate this interpolation expansion: In the case of 2 derivatives, using the product rule~\eqref{eq:product_rule_bis}, we get
\begin{equation} \label{eq:expansion_2_derivatives}
\begin{aligned}
	\partial_{k_{j_2}} \partial_{k_{j_1}} (h \circ \hat{f})(k)
	&= (\partial_{k_{j_2}} \partial_{k_{j_1}} \hat{f})(k) \int_0^1 h'(f^{[t_1]}_{j_1}(k)) \d t_1
		+ (\partial_{k_{j_1}} \hat{f})(k) \int_0^1 (\partial_{k_{j_2}} h'(f^{[t_1]}_{j_1}(k))) \d t_1 \\
	&\quad + 2 \pi L^{-1} (\partial_{k_{j_2}} \partial_{k_{j_1}} \hat{f})(k) \int_0^1 (\partial_{k_{j_2}} h'(f^{[t_1]}_{j_1}(k))) \d t_1 \;,
\end{aligned}
\end{equation}
where we note that the third term on the right-hand side of~\eqref{eq:expansion_2_derivatives} carries again an excess derivative, together with a factor of $ L^{-1} $. We expand $ (\partial_{k_{j_2}} h'(f^{[t_1]}_{j_1}(k))) $ by iteratively defining
\begin{equation} \label{eq:interpolated_f}
	f^{[t_1,\ldots,t_{n+1}]}_{j_1,\ldots,j_{n+1}}(k)
	:= t_{n+1} f^{[t_1,\ldots,t_n]}_{j_1,\ldots,j_n}(k + 2 \pi L^{-1} e_{j_{n+1}}) + (1-t_{n+1}) f^{[t_1,\ldots,t_n]}_{j_1,\ldots,j_n}(k) \;,
\end{equation}
and expanding
\begin{equation}
	\partial_{k_{j_{n+1}}} h^{(n)}(f^{[t_1,\ldots,t_n]}_{j_1,\ldots,j_n})
	= \big( \partial_{k_{j_{n+1}}} f^{[t_1,\ldots,t_n]}_{j_1,\ldots,j_n}(k) \big)
		\int_0^1 h^{(n+1)}(f^{[t_1,\ldots,t_{n+1}]}_{j_1,\ldots,j_{n+1}}) \d t_{n+1} \;,
\end{equation}
where $ h^{(n)} $ denotes the $ n $-th derivative of $ h $. Note that each $ f^{[t_1,\ldots,t_n]}_{j_1,\ldots,j_n}(k) $ is a convex combination of $ \{\hat{f}(k + k') : |k'| \le 2 \pi L^{-1} n \} $.\\
Iterating the product rule and expansion, we find that $ \partial_k^{\um}(h \circ \hat{f})(k) $ is a finite sum of terms of the form
\begin{equation}
	\int_0^1 \d t_1 \ldots \int_0^1 \d t_\ell \;
		\big(\partial_k^{\um_1} f^{\sharp_1}(k) \big) 
		\ldots
		\big(\partial_k^{\um_\ell} f^{\sharp_\ell}(k) \big)
		\big(h^{(\ell)} (f^{\sharp_{\ell+1}}(k)) \big)
		L^{m - |\um_1 + \ldots + \um_\ell|} \;,
\end{equation}
where $ 0 \le \ell \le m $, and where $ f^{\sharp_1}, \ldots, f^{\sharp_{\ell+1}} $ are any of the interpolating functions $ f^{[t_1,\ldots,t_n]}_{j_1,\ldots,j_n} $ with $ n \le \ell $. Furthermore, note that the product rule~\eqref{eq:product_rule_bis} produced $ 0 \le |\um_1 + \ldots + \um_\ell| - m \le \ell m $ excess derivatives, where each came with a factor of $ L^{-1} $. Since the interpolating functions $ f^{\sharp}(k) = f^{[t_1,\ldots,t_n]}_{j_1,\ldots,j_n}(k) $ are convex combinations of $ \{ \hat{f}(k+k') : |k'| \le 2 \pi n L^{-1} \} $, where we recall $ n \le \ell $, we conclude with smoothness of $ h $:
\begin{equation}
	|\partial_k^{\um_j} f^{\sharp_j}(k)|
	\le \sup_{|k'| \le 2 \pi \ell L^{-1}} |\partial_k^{\um_j} \hat{f}(k + k')| \;, \qquad
	|h^{(\ell)} (f^{\sharp}(k))|
	\le \sup_{|k'| \le 2 \pi \ell L^{-1}} |h^{(\ell)}(\hat{f}(k + k'))| \;.
\end{equation}
Summing over all $ \ell $ and all choices of multi-indices $ \um_1, \ldots, \um_\ell $, we then obtain~\eqref{eq:chain_rule_repeated}.

The proofs after replacing right- by left-derivatives $ \bar\partial_{k_j} $ are analogous.
\end{proof}

To control discrete derivatives of $ B(k) $ as in~\eqref{eq:Bintegral} and $ I(k) $ as in~\eqref{eq:cEalphahat}, we start by controlling discrete derivatives of $ \tilde\lambda_{\alpha,k} $, $ \tilde{n}_{\alpha,k}^2 $, and $ n_{\alpha,k}^2 $. 

\begin{lemma}[Pointwise bounds on $ \tilde\lambda_{\alpha,k} $, $ \tilde{n}_{\alpha,k}^2 $, and $ n_{\alpha,k}^2 $]
\label{lem:regularity_lambda}
Let $ |k| \le 2 R $ and recall the definitions of $ \tilde\lambda_{\alpha,k} $~\eqref{eq:lambdatilde}, $ \tilde{n}_{\alpha,k}^2 $~\eqref{eq:ntilde} and $ n_{\alpha,k} $~\eqref{eq:n}.
Given $ m \in \NNN_0 $, any multi-index $ \um \in \NNN_0^3 $ with $ |\um| = m $, and any $ a \in \RRR $, there exist $ C_m, C_{m,a} > 0 $ such that for $ L $ large enough,
\begin{equation} \label{eq:regularity_lambda_n}
\begin{aligned}
	|\partial_k^{\um} \tilde\lambda_{\alpha,k}^a|
	&\le C_{m,a} \fl^{-m} \tilde\lambda_{\alpha,k}^a \qquad
	&&\forall k \in \Lambda_L, \; \alpha \in \{1, \ldots, M\} \;, \\
	|\partial_k^{\um} \tilde{n}_{\alpha,k}^a|
	&\le C_{m,a} \fl^{-m} \tilde{n}_{\alpha,k}^a \qquad
	&&\forall k \in \Lambda_L, \; \alpha \in \{1, \ldots, M\} \;, \\
	|\partial_k^{\um} n_{\alpha,k}^2 |
	&\le C_m \fl^{-m} L^3 \rho^{\frac 23} M^{-1} \tilde\lambda_{\alpha,k} \qquad
	&&\forall k \in \Lambda_L, \; \alpha \in \{1, \ldots, M\} \;,
\end{aligned}
\end{equation}
and the same bounds are true if any number of right-derivatives $ \partial_{k_j} $ is replaced by left-derivatives $ \bar\partial_{k_j} $. Moreover, given $ a \in \RRR $ and $ C_m > 0 $ there exist $ c_{m,a}, C_{m,a} > 0 $ such that
\begin{equation} \label{eq:lambda_diff_control}
	c_{m,a} 	\tilde\lambda_{\alpha,k}^a
	\le \tilde\lambda_{\alpha,k+k'}^a
	\le C_{m,a} \tilde\lambda_{\alpha,k}^a \qquad \forall |k'| \le C_m L^{-1} \;.
\end{equation}
\end{lemma}

\begin{proof}
We will repeatedly make use of the following interpolation argument: If $ \hat{f}: \Lambda_L \to \CCC $ extends to a smooth function, equally denoted $ \hat{f} $, then
\begin{equation} \label{eq:interpolation_kprime}
	|\hat{f}(k + k') - \hat{f}(k)|
	\le |k'| \sup_{|k''| \le |k'|} |\nabla \hat{f}(k+k'')| \;.
\end{equation}
Repeatedly applying this interpolation argument allows bounding discrete derivatives $ \partial_k^{\um} $ by smooth derivatives $ \partial^{\um} $ as
\begin{equation} \label{eq:interpolation}
	|\partial_k^{\um} \hat{f}(k)|
	\le C_m \sup_{|k'| \le 2 \pi m L^{-1}} | \partial^{\um} \hat{f}(k+k')| \;.
\end{equation}
Recall~\eqref{eq:lambdatilde} that $ \tilde\lambda_{\alpha,k} = \fl f_>( \fl^{-1} k \cdot \hat{\omega}_\alpha) $ with a smooth, non-decreasing function $ f_>(x) $ as in~\eqref{eq:f}. So each smooth derivative produces a factor of $ \fl^{-1} $ and for $ (\alpha,k) $ such that $ x = \fl^{-1} k \cdot \hat\omega_\alpha \in [0,1] $ we obtain
\begin{equation}
	|\partial^{\um} \tilde\lambda_{\alpha,k}|
	\le C_m \fl^{-m+1} \;.
\end{equation}
For $ x > 1 $, we simply have $|\nabla \tilde\lambda_{\alpha,k}| \le 1 $, and higher derivatives even vanish, since $ k \mapsto \tilde\lambda_{\alpha,k} $ is linear.
Finally, for $ x < 0 $, $ f_>(x) $ is constant, so all derivatives vanish. Thus, smooth derivatives are uniformly bounded by $ |\partial^{\um} \tilde\lambda_{\alpha,k}| \le C_m \fl^{-m+1} $, and with~\eqref{eq:interpolation} and $ \tilde\lambda_{\alpha,k} \ge \fl/2 $ we conclude
\begin{equation}
	|\partial_k^{\um} \tilde\lambda_{\alpha,k}|
	\le C_m \sup_{|k'| \le 2 \pi m L^{-1}} | \partial^{\um} \tilde\lambda_{\alpha,k+k'}|
	\le C_m \fl^{-m+1}
	\le C_m \fl^{-m} \tilde\lambda_{\alpha,k}
\end{equation}
for $ m \ge 1 $, where the analogous estimate for $ m = 0 $ is trivial. This establishes the first bound in~\eqref{eq:regularity_lambda_n} for $ a = 1 $.\\
To extend the result to $ a \neq 1 $, we first prove the shift control bound~\eqref{eq:lambda_diff_control}: Since we assumed $ |k'| \le C_m L^{-1} $, from~\eqref{eq:interpolation_kprime} we get
\begin{equation} \label{eq:lambdatilde_shift_bound}
	\tilde\lambda_{\alpha,k+k'}
	\le \tilde\lambda_{\alpha,k} + C |k'|
	\le \tilde\lambda_{\alpha,k} (1 + C_m L^{-1} \fl^{-1})
	\le C_m \tilde\lambda_{\alpha,k} \;.
\end{equation}
Analogously, $ \tilde\lambda_{\alpha,k+k'}^{-1} \le \tilde\lambda_{\alpha,k}^{-1} (1 - C_m L^{-1} \fl^{-1})^{-1} \le C_m \tilde\lambda_{\alpha,k}^{-1} $, and taking the $ a $-th power readily implies~\eqref{eq:lambda_diff_control}.
Let us return to the bound on $ |\partial_k^{\um} \tilde\lambda_{\alpha,k}^a| $. Recall that $ k \mapsto \tilde\lambda_{\alpha,k} $ extends to a smooth function on $ \RRR^3 $, as initially defined in~\eqref{eq:lambdatilde}. We may therefore apply the chain rule bound from Lemma~\ref{lem:product_chain_rule} with $ h \in C^\infty $ such that $ h(x) = x^a $ for $ x \in [\fl/2,2R] $:
\begin{equation} \label{eq:derivatives_lambda_a_expansion}
\begin{aligned}
	|\partial_k^{\um} \tilde\lambda_{\alpha,k}^a|
	&\le C_{m,a} \sum_{\ell=0}^m \sum_{\substack{\um_1, \ldots, \um_\ell \in \NNN_0^3 \\ m \le |\um_1 + \ldots + \um_\ell| \le (\ell+1) m}} 
		L^{m - |\um_1 + \ldots + \um_\ell|} \\
		& \times \sup_{\substack{ k'_1, \ldots, k'_{\ell+1} \in \RRR^3 \\ |k'_1|, \ldots, |k'_{\ell+1}| \le 2 \pi \ell L^{-1}}}
		|\partial_k^{\um_1} \tilde\lambda_{\alpha,k+k'_1}| \cdots |\partial_k^{\um_\ell} \tilde\lambda_{\alpha,k+k'_\ell}|
		 |\tilde\lambda_{\alpha,k+k'_{\ell+1}}^{a-\ell}| \;.
\end{aligned}
\end{equation}
We now use the interpolation~\eqref{eq:interpolation_kprime}, as well as the shift bound~\eqref{eq:lambda_diff_control},
\begin{equation}
\begin{aligned}
	&|\partial_k^{\um} \tilde\lambda_{\alpha,k+k'}|
	\le |\partial_k^{\um} \tilde\lambda_{\alpha,k}| + C_m L^{-1} \sup_{|k''| \le C_m L^{-1}} | \nabla \partial_k^{\um} \tilde\lambda_{\alpha,k+k''} | \\
	&\le C_m \fl^{-m} \tilde\lambda_{\alpha,k}
		+ C_m L^{-1} \fl^{-m-1} \sup_{|k''| \le C_m L^{-1}} \tilde\lambda_{\alpha,k+k''}
	\le C_m \fl^{-m} \tilde\lambda_{\alpha,k} \;.
\end{aligned}
\end{equation}
Moreover,~\eqref{eq:lambda_diff_control} yields $ |\tilde\lambda_{\alpha,k+k_{\ell+1}'}^{a-\ell}| \le C_{m,a} \tilde\lambda_{\alpha,k}^{a-\ell} $ so plugging all bounds into~\eqref{eq:derivatives_lambda_a_expansion} and using $ L^{-1} \ll \fl $, we get
\begin{equation}
	|\partial_k^{\um} \tilde\lambda_{\alpha,k}^a|
	\le C_{m,a} \sum_{\ell=0}^m \sum_{\substack{\um_1, \ldots, \um_\ell \in \NNN_0^3 \\ m \le |\um_1 + \ldots + \um_\ell| \le (\ell+1) m}} 
		L^{m - |\um_1 + \ldots + \um_\ell|} \fl^{-|\um_1 + \ldots + \um_\ell| } \tilde\lambda_{\alpha,k}^{\ell} \tilde\lambda_{\alpha,k}^{a-\ell}
	\le C_{m,a} \fl^{-m} \tilde\lambda_{\alpha,k}^a \;,
\end{equation}
which is the first bound claimed in~\eqref{eq:regularity_lambda_n}.\\
The second bound then immediately follows, since $ \tilde{n}_{\alpha,k}^2 $ is proportional to $ \tilde\lambda_{\alpha,k} $, see~\eqref{eq:ntilde}.\\
For the third bound, we recall the definition of $ n_{\alpha,k} $~\eqref{eq:n}. For $ m = 0 $, we have
\begin{equation}
	|n_{\alpha,k}^2|
	= \sum_{p \in \Lambda_L} \hat{u}^\alpha(p)^2 \hat{v}^\alpha(p-k)^2
	\le C |B_\alpha \cap \BF^c \cap (\BF + k)| \;.
\end{equation}
Since the patches $ B_\alpha $ are regular with a surface area $ \sim \rho^{\frac 23} M^{-1} $, and the surface normal vector varies by an amount of $ M^{-\frac 12} \ll \fl $ over the patch, the number of points in the above intersection is bounded as
\begin{equation} \label{eq:Vol_bound}
\begin{aligned}
	|B_\alpha \cap \BF^c \cap (\BF + k)|
	&\le C L^3 \mathrm{Vol}(B_\alpha \cap B_{k_{\F}}(0)^c \cap (B_{k_{\F}}(0) + k)) + \cO_L(L^2) \\
	&\le C \rho^{\frac 23} M^{-1} |k| (\max\{\hat{k} \cdot \hat\omega_\alpha + C M^{-\frac 12}, 0 \} ) L^3 \\
	&\le C \rho^{\frac 23} M^{-1} \tilde\lambda_{\alpha,k} L^3 \;,
\end{aligned}
\end{equation}
where we used that~\eqref{eq:lambdatilde} $ \tilde\lambda_{\alpha,k} \ge k \cdot \hat\omega_\alpha $ and that $ |k| M^{-\frac 12} \le C \fl/2 \le C \tilde\lambda_{\alpha,k} $. This gives the desired bound for $ m = 0 $. For $ m \ge 1 $, we observe that $ |\partial^{\um}_k \hat{v}^\alpha(p-k)^2| \le C_m \fl^{-m} $, while $ m $ discrete derivatives enlarge the support of the function by at most $ 2 \pi m L^{-1} $, so
\begin{equation}
	|\partial^{\um}_k n_{\alpha,k}^2|
	\le \sum_{p \in \Lambda_L} |\hat{u}^\alpha(p)^2| |\partial^{\um}_k \hat{v}^\alpha(p-k)^2|
	\le C_m \fl^{-m} |B_{\alpha,k}^{(+m)}| \;.
\end{equation}
where $ B_{\alpha,k}^{(+m)} := \{ p \in \Lambda_L : \mathrm{dist}(p, B_\alpha \cap \BF^c \cap (\BF + k)) \le 2 \pi m L^{-1} \} $. Evidently,
\begin{equation}
	|B_{\alpha,k}^{(+m)}|
	\le |B_\alpha \cap \BF^c \cap (\BF + k)| + \cO_L(L^2) 
	\le C \rho^{\frac 23} M^{-1} \tilde\lambda_{\alpha,k} L^3 \;,
\end{equation}
which readily yields the third bound in~\eqref{eq:regularity_lambda_n}.\\
The proof after replacing right-derivatives $ \partial_{k_j} $ by left-derivatives $ \bar\partial_{k_j} $ is completely analogous.
\end{proof}

\begin{lemma}[Pointwise bound on $ B $]
\label{lem:regularity_B}
Recall the definition of $ B(k) $~\eqref{eq:B} and the notation $ \hat\varphi_B(k)_{\alpha,\beta} $~\eqref{eq:varphi}. Given $ m \in \NNN_0 $ and any multi-index $ \um \in \NNN_0^3 $ with $ |\um| = m $, there exists $ C_{V,m} > 0 $ such that
\begin{equation} \label{eq:regularity_B}
	| \partial_k^{\um} \hat\varphi_B(k)_{\alpha,\beta} |
	\le C_{V,m} \fl^{-m} (\tilde\lambda_{\alpha,k} + \tilde\lambda_{\beta,k})^{-1} \qquad
	\forall k \in \Lambda_L , \; \alpha, \beta \in \{1, \ldots, M\} \;,
\end{equation}
and the same bound is true if any number of right-derivatives $ \partial_{k_j} $ is replaced by left-derivatives $ \bar\partial_{k_j} $.
Moreover, $ B(k)_{\alpha,\beta} = \hat\varphi_B(k)_{\alpha,\beta} = 0 $ if $ |k| > R $ or $ k \notin \supp(\hat\chi_{>,\alpha}) \cap \supp(\hat\chi_{>,\beta}) $. 
\end{lemma}

\begin{proof}[Proof of Lemma \ref{lem:regularity_B}]
We only show the proof for all derivatives being right-derivatives, since the proof including left-derivatives is again analogous.\\
Recall the integral representation~\eqref{eq:Bintegral} of $ B(k) $, where the integrand is positive for all $ \mu \ge 0 $. From~\eqref{eq:hP} $ \tilde{v}_{\alpha,k}^2 \lesssim \hat\chi_{>,\alpha}(k)^2 \hat{V}(k) \tilde\lambda_{\alpha,k} M^{-1} $, we readily see that this integrand vanishes if $ k \notin \supp(\hat\chi_{>,\alpha}) \cap \supp(\hat\chi_{>,\beta}) $ or if $ |k| > R $, in which case $ \hat{V}(k) = 0 $.
Applying the product rule from Lemma~\ref{lem:product_chain_rule}, we then obtain
\begin{equation} \label{eq:Bderivativebound}
\begin{aligned}
	&| \partial_k^{\um} \hat\varphi_B(k)_{\alpha,\beta} |
	= \left\vert \partial_k^{\um} \frac{B(k)_{\alpha,\beta}}{\tilde{n}_{\alpha,k} \tilde{n}_{\beta,k}} L^3 \rho^{\frac 23} \right\vert \\
	&\le C_m \sum_{\substack{\um_1, \ldots, \um_8 \in \NNN_0^3 \\ m \le |\um_1 + \ldots + \um_8| \le 8 m}} L^{m-|\um_1 + \ldots + \um_8|} \int_0^\infty \mu^2
		|\partial_k^{\um_1} \tilde\lambda_{\alpha,k}^{-1}|
		|\partial_k^{\um_2} (\tilde\lambda_{\alpha,k}^{-2} + \mu^2)^{-1}| 
		|\partial_k^{\um_3} \tilde\lambda_{\beta,k}^{-1}| \\
	& \qquad \times  
		|\partial_k^{\um_4} (\tilde\lambda_{\beta,k}^{-2} + \mu^2)^{-1}|
		|\partial_k^{\um_5} \hat{V}(k) |
		|\partial_k^{\um_6} \hat\chi_{>,\alpha}(k)|
		|\partial_k^{\um_7} \hat\chi_{>,\beta}(k)|
		|\partial_k^{\um_8} (1 + \mu^2 \tilde Q_k(\mu))^{-1}|
		\d \mu \;,
\end{aligned}
\end{equation}
where we abbreviated
\begin{equation} \label{eq:Qkmu}
	\tilde Q_k(\mu)
	= 2 \sum_{\gamma = 1}^M \tilde\lambda_{\gamma,k}^{-1} (\tilde\lambda_{\gamma,k}^{-2} + \mu^2)^{-1} \tilde{v}_{\gamma,k}^2 \;.
\end{equation}
Lemma~\ref{lem:regularity_lambda} yields $ |\partial_k^{\um_1} \tilde\lambda_{\alpha,k}^{-1}| \le C_m \fl^{-|\um_1|} \tilde\lambda_{\alpha,k}^{-1} $. To bound $ |\partial_k^{\um_2} (\tilde\lambda_{\alpha,k}^{-2} + \mu^2)^{-1}| $, observe that $ k \mapsto \tilde\lambda_{\alpha,k}^{-2} $ extends to a smooth function on $ \RRR^3 $, so by the chain rule in Lemma~\ref{lem:product_chain_rule} with $ h(x) = (x + \mu^2)^{-1} $,
\begin{equation}
\begin{aligned}
	&|\partial_k^{\um_2} (\tilde\lambda_{\alpha,k}^{-2} + \mu^2)^{-1}| \\
	&\le C_m \sum_{\ell=0}^{|\um_2|} \sum_{\substack{\um_{2,1}, \ldots, \um_{2,\ell} \in \NNN_0^3 \\ |\um_2| \le |\um_{2,1} + \ldots + \um_{2,\ell}| \le (\ell+1) |\um_2|}} 
		L^{|\um_2| - |\um_{2,1} + \ldots + \um_{2,\ell}|} \\ & \times \sup_{\substack{ k'_1, \ldots, k'_{\ell+1} \in \RRR^3 \\ |k'_1|, \ldots, |k'_{\ell+1}| \le 2 \pi \ell L^{-1}}}
		|\partial_k^{\um_{2,1}} \tilde\lambda_{\alpha,k + k'_1}^{-2}| \cdots |\partial_k^{\um_{2,\ell}} \tilde\lambda_{\alpha,k + k'_\ell}^{-2}|
		 |(\tilde\lambda_{\alpha,k + k'_{\ell+1}}^{-2} + \mu^2)^{-\ell-1}| \;.
\end{aligned}
\end{equation}
With Lemma~\ref{lem:regularity_lambda} and~\eqref{eq:lambda_diff_control} we estimate
\begin{equation}
	|\partial_k^{\um_{2,j}} \tilde\lambda_{\alpha,k + k'_j}^{-2}|
	\le C_m \fl^{-|\um_{2,j}|} \tilde\lambda_{\alpha,k + k'_j}^{-2}
	\le C_m \fl^{-|\um_{2,j}|} \tilde\lambda_{\alpha,k}^{-2} \;.
\end{equation}
Moreover, with~\eqref{eq:lambda_diff_control},
\begin{equation}
	|(\tilde\lambda_{\alpha,k + k'_{\ell+1}}^{-2} + \mu^2)^{-\ell-1}|
	\le |((C_m \tilde\lambda_{\alpha,k})^{-2} + \mu^2)^{-\ell-1}|
	\le C_m |(\tilde\lambda_{\alpha,k}^{-2} + \mu^2)^{-\ell-1}| \;.
\end{equation}
Thus,
\begin{equation} \label{eq:control_partial_um2}
\begin{aligned}
	&|\partial_k^{\um_2} (\tilde\lambda_{\alpha,k}^{-2} + \mu^2)^{-1}| \\
	&\le C_m \sum_{\ell=0}^{|\um_2|} \sum_{\substack{\um_{2,1}, \ldots, \um_{2,\ell} \in \NNN_0^3 \\ |\um_2| \le |\um_{2,1} + \ldots + \um_{2,\ell}| \le (\ell+1) |\um_2|}} 
		L^{|\um_2| - |\um_{2,1} + \ldots + \um_{2,\ell}|} \fl^{- |\um_{2,1} + \ldots + \um_{2,\ell}|} \frac{(\tilde\lambda_{\alpha,k}^{-2})^\ell}{(\tilde\lambda_{\alpha,k}^{-2} + \mu^2)^{\ell+1}} \\
	&\le C_m \fl^{-|\um_2|} (\tilde\lambda_{\alpha,k}^{-2} + \mu^2)^{-1} \;,
\end{aligned}
\end{equation}
where we used $ L^{-1} \ll \fl $, so terms with $ |\um_{2,1} + \ldots + \um_{2,\ell}| > |\um_2| $ are subleading.
The terms with $ \partial_k^{\um_3} $ and $ \partial_k^{\um_4} $ in~\eqref{eq:Bderivativebound} are treated analogously. Moreover, $ |\partial_k^{\um_5} \hat{V}(k)| \le C_V $ since $ \hat{V} \in C_c^\infty $. Similarly, with~\eqref{eq:chi>}, we get $ |\partial_k^{\um_6} \hat\chi_{>,\alpha}(k)| \le C \fl^{-|\um_6|} $ and $ |\partial_k^{\um_7} \hat\chi_{>,\beta}(k)| \le C \fl^{-|\um_7|} $.
To control the $ \partial_k^{\um_8} $-term, we use definition~\eqref{eq:hP} of $ \tilde{v}_{\gamma,k}^2 $ and positivity of all contributions to the sum in~\eqref{eq:Qkmu}, and apply the product rule in Lemma~\ref{lem:product_chain_rule}:
\begin{equation} \label{eq:control_Qkmu}
\begin{aligned}
	|\partial_k^{\um_8} \tilde Q_k(\mu)| 
	&\le C_m \sum_{\substack{\um_{8,1}, \ldots \um_{8,5} \in \NNN_0^3 \\ |\um_8| \le |\um_{8,1} + \ldots + \um_{8,5}| \le 5 |\um_8|}}
		L^{|\um_8| - |\um_{8,1} + \ldots + \um_{8,5}|}
		\sum_{\gamma = 1}^M
		|\partial_k^{\um_{8,1}} \tilde\lambda_{\gamma,k}^{-1}| \\
	&\quad \times |\partial_k^{\um_{8,2}} (\tilde\lambda_{\gamma,k}^{-2} + \mu^2)^{-1} |
		|\partial_k^{\um_{8,3}} \hat{V}(k)|
		|\partial_k^{\um_{8,4}} (\hat\chi_{>,\gamma}(k))^2|
		\left\vert \partial_k^{\um_{8,5}} \tilde{n}_{\gamma,k}^2 \rho^{-\frac 23} L^{-3} \right\vert \\
	&\le C_{V,m} \sum_{\substack{\um_{8,1}, \ldots \um_{8,5} \in \NNN_0^3 \\ |\um_8| \le |\um_{8,1} + \ldots + \um_{8,5}| \le 5 |\um_8|}}
		L^{|\um_8| - |\um_{8,1} + \ldots + \um_{8,5}|} \\
	&\quad \times \sum_{\gamma = 1}^M
		\fl^{-|\um_{8,1}|} \tilde\lambda_{\gamma,k}^{-1}
		\fl^{-|\um_{8,2}|} (\tilde\lambda_{\gamma,k}^{-2} + \mu^2)^{-1}
		\fl^{-|\um_{8,4}|}
		\fl^{-|\um_{8,5}|} M^{-1} \tilde\lambda_{\gamma,k} \\
	&\le C_{V,m} \mu^{-2} \fl^{-|\um_8|}
		\sum_{\gamma = 1}^M M^{-1}
	\le C_{V,m} \mu^{-2} \fl^{-|\um_8|}  \;,
\end{aligned}
\end{equation}
Here, we controlled the $ \partial_k^{\um_{8,5}} $-contribution by Lemma \ref{lem:regularity_lambda} and $ \tilde{n}_{\gamma,k}^2 \le C \tilde\lambda_{\gamma,k} \rho^{\frac 23} L^3 M^{-1} $, see~\eqref{eq:ntilde}, and all other contributions as above, and then used $ L^{-1} \ll \fl $ to conclude that all contributions with $ |\um_{8,1} + \ldots + \um_{8,5}| > |\um_8| $ are subleading. Next, we apply the chain rule in Lemma~\ref{lem:product_chain_rule} with $ h(x) = (1 + \mu^2 x)^{-1} $, where we observe that $ k \mapsto \tilde Q_k(\mu) $ naturally extends to a smooth function on $ \RRR^3 $, since $ \hat{V}(k) $ and $ k \mapsto \tilde\lambda_{\alpha,k}^{-1} $ extend to smooth functions:
\begin{equation} \label{eq:control_partial_um6}
\begin{aligned}
	&|\partial_k^{\um_8} (1 + \mu^2 \tilde Q_k(\mu))^{-1}| \\
	&\le C_m \sum_{\ell=0}^{|\um_8|} \sum_{\substack{\um_{8,1}, \ldots, \um_{8,\ell} \in \NNN_0^3 \\ |\um_8| \le |\um_{8,1} + \ldots + \um_{8,\ell}| \le (\ell+1) |\um_8|}} 
		L^{|\um_8| - |\um_{8,1} + \ldots + \um_{8,\ell}|} \\
	&\quad \times \sup_{\substack{ k'_1, \ldots, k'_{\ell+1} \in \RRR^3 \\ |k'_1|, \ldots, |k'_{\ell+1}| \le 2 \pi \ell L^{-1}}}
		|\partial_k^{\um_{8,1}} \tilde Q_{k + k'_1}(\mu)|
		\ldots 
		|\partial_k^{\um_{8,\ell}} \tilde Q_{k + k'_\ell}(\mu)|
		\mu^{2 \ell} |(1 + \mu^2 \tilde Q_{k + k'_{\ell+1}}(\mu))^{-\ell-1}| \;.
\end{aligned}
\end{equation}
We again remove shifts $ |k'| \le C_m L^{-1} $ by interpolating
\begin{equation}
\begin{aligned}
	&\mu^2 |\partial_k^{\um_{8,j}} \tilde Q_{k + k'}(\mu)|
	\le \mu^2 |\partial_k^{\um_{8,j}} \tilde Q_k(\mu)|
		+ C_m L^{-1} \sup_{|k''| \le |k'|} \mu^2 |\nabla \partial_k^{\um_{8,j}} \tilde Q_{k+k''}(\mu)| \\
	&\le C_{V,m} \fl^{-|\um_{8,j}|} + C_{V,m} \fl^{-|\um_{8,j}|-1} L^{-1}
	\le C_{V,m} \fl^{-|\um_{8,j}|} \;,
\end{aligned}
\end{equation}
as well as
\begin{equation}
\begin{aligned}
	&|(1 + \mu^2 \tilde Q_{k+k'}(\mu))^{-1}|
	\le \Big(1 + \mu^2 \tilde Q_{k}(\mu) - C_m L^{-1} \sup_{|k''| \le |k'|} \mu^2 |\nabla \tilde Q_{k+k''}(\mu)| \Big)^{-1} \\
	&\le \Big(1 - C_{V,m} L^{-1} \fl^{-1} + \mu^2 \tilde Q_{k}(\mu) \Big)^{-1}
	\le C_{V,m} (1 + \mu^2 \tilde Q_{k}(\mu))^{-1} \;.
\end{aligned}
\end{equation}
Hence,
\begin{equation}
\begin{aligned}
	&|\partial_k^{\um_8} (1 + \mu^2 \tilde Q_k(\mu))^{-1}| \\
	&\le C_{V,m} \sum_{\ell=0}^{|\um_8|} \sum_{\substack{\um_{8,1}, \ldots, \um_{8,\ell} \in \NNN_0^3 \\ |\um_8| \le |\um_{8,1} + \ldots + \um_{8,\ell}| \le (\ell+1) |\um_8|}} 
		L^{|\um_8| - |\um_{8,1} + \ldots + \um_{8,\ell}|} 
		\fl^{- |\um_{8,1} + \ldots + \um_{8,\ell}|}
		(1 + \mu^2 \tilde Q_k(\mu))^{-1} \\
	&\le C_{V,m} \fl^{-|\um_8|} \;.
\end{aligned}
\end{equation}
Plugging all derivative bounds into \eqref{eq:Bderivativebound}, and using that terms with $ |\um_1 + \ldots + \um_8| > m $ derivatives are subleading due to $ L^{-1} \ll \fl $, we finally get
\begin{equation}
\begin{aligned}
	| \partial_k^{\um} \hat\varphi_B(k)_{\alpha,\beta} |
	&\le C_{V,m} \fl^{-m} \int_0^\infty \mu^2
		\tilde\lambda_{\alpha,k}^{-1} (\tilde\lambda_{\alpha,k}^{-2} + \mu^2)^{-1}
		\tilde\lambda_{\beta,k}^{-1} (\tilde\lambda_{\beta,k}^{-2} + \mu^2)^{-1}
		\d \mu \\
	&\le C_{V,m} \fl^{-m}
		(\tilde\lambda_{\alpha,k} + \tilde\lambda_{\beta,k})^{-1} \;,
\end{aligned}
\end{equation}
where the $ \mu $-integral was evaluated explicitly.
\end{proof}

\begin{lemma}[Pointwise bounds on $ K $ and $ K_j $]
\label{lem:regularity_K}
Recall the definition \eqref{eq:K} of $ K(k) $ and the notation $ K_j(k) = K(k) p_j(k) $, where $ p_j $ is any polynomial of degree $ j $, as well as $ \hat\varphi_K(k)_{\alpha,\beta} $ from~\eqref{eq:varphi}. Then, given $ m,j \in \NNN_0 $ and any multi-index $ \um \in \NNN_0^3 $ with $ |\um| = m $, there exists some $ C_{V,m}, C_{V,m,p} > 0 $ such that
\begin{equation} \label{eq:regularity_K}
\begin{aligned}
	| \partial_k^{\um} \hat\varphi_K(k)_{\alpha,\beta} |
	&\le C_{V,m} \fl^{-m} \tilde\lambda_{\alpha,k}^{-\frac 12} \tilde\lambda_{\beta,k}^{-\frac 12} \qquad
	\forall k \in \Lambda_L , \; \alpha, \beta \in \{1, \ldots, M\} \;, \\
	| \partial_k^{\um} \hat\varphi_{K_j}(k)_{\alpha,\beta} |
	&\le C_{V,m,p} \fl^{-m} \tilde\lambda_{\alpha,k}^{-\frac 12} \tilde\lambda_{\beta,k}^{-\frac 12} \qquad
	\forall k \in \Lambda_L , \; \alpha, \beta \in \{1, \ldots, M\} \;,
\end{aligned}
\end{equation}
and the same bounds are true if any number of right-derivatives $ \partial_{k_j} $ is replaced by left-derivatives $ \bar\partial_{k_j} $.
Moreover,
\begin{equation}
K(k)_{\alpha,\beta} = \hat\varphi_K(k)_{\alpha,\beta} = 0 \qquad \text{if } |k| > R \text{ or } k \notin \supp(\hat\chi_{>,\alpha}) \cap \supp(\hat\chi_{>,\beta}) \;.
\end{equation}
\end{lemma}

\begin{proof}
We only show the proof for right-derivatives, the case with left-derivatives is analogous. Recall~\eqref{eq:B} and expand in a Taylor series
\begin{equation}
	K(k)_{\alpha, \beta}
	= \frac 12 \log (1-B(k))_{\alpha,\beta}
	= - \sum_{n=1}^\infty \frac{1}{2n} (B(k)^n)_{\alpha, \beta} \;.
\end{equation}
By operator monotonicity of the square root, as well as~\eqref{eq:K} and~\eqref{eq:B}
\begin{equation}
\begin{aligned}
	(h^2 + 2 h^{\frac 12} P h^{\frac 12})^{\frac 12}
	&\ge h
	\quad \Rightarrow \quad
	h^{-\frac 12} (h^2 + 2 h^{\frac 12} P h^{\frac 12})^{\frac 12} h^{-\frac 12}
	\ge 1 \\
	\quad \Rightarrow \quad
	-K(k)
	&\ge 0
	\quad \Rightarrow \quad
	0 \le B(k) < 1 \;,
\end{aligned}
\end{equation}
so the log-series converges absolutely. Next, we show that $ \frac{1}{1 - \Vert B(k) \Vert} $ is uniformly bounded in $ (\rho,L) $: By spectral calculus and~\eqref{eq:B},~\eqref{eq:K}
\begin{equation}
	\frac{1}{1 - \Vert B(k) \Vert}
	= (\min \sigma(1-B(k)) )^{-1}
	= \max \sigma(e^{-2 K(k)})
	= \Vert h^{-\frac 12} (h^2 + 2 h^{\frac 12} P h^{\frac 12})^{\frac 12} h^{-\frac 12} \Vert \;.
\end{equation}
Now, recall~\eqref{eq:hP} that $ P $ is a rank-1 matrix $ P(k) = |\tilde{v}_k \rangle \langle \tilde{v}_k | $ with $ \tilde{v}_k = (\tilde{v}_{\alpha,k})_{\alpha=1}^M $. Thus, using the Sherman--Morrison formula, the integral identity $ A^{\frac 12} = \frac{2}{\pi} \int_0^\infty \frac{A \; \d t}{A + t^2} $, and omitting indices $ k $ for convenience,
\begin{equation}
\begin{aligned}
	&\Vert h^{-\frac 12} (h^2 + 2 h^{\frac 12} P h^{\frac 12})^{\frac 12} h^{-\frac 12} \Vert
	= \Bigg\Vert 1 + \frac{4}{\pi} \int_0^\infty \frac{t^2 (h^2+t^2)^{-1} |\tilde{v} \rangle \langle \tilde{v}| (h^2+t^2)^{-1}}{1 + 2 \langle \tilde{v}, (h^2+t^2)^{-1} h \tilde{v} \rangle} \d t \Bigg\Vert \\
	&\le 1 + C \int_0^\infty t^2 \Vert (h^2+t^2)^{-1} \tilde{v} \Vert^2 \d t
	= 1 + C \int_0^\infty t^2 \sum_{\alpha=1}^M (\tilde\lambda_{\alpha,k}^2 + t^2)^{-2} \tilde{v}_{\alpha,k}^2 \; \d t \\
	&\le 1 + C_V \sum_{\alpha=1}^M M^{-1} \int_0^\infty t^2  (\tilde\lambda_{\alpha,k}^2 + t^2)^{-2} \tilde\lambda_{\alpha,k} \; \d t
	\le 1 + C_V \;.
\end{aligned}
\end{equation}
where we used $ \tilde{v}_{\alpha,k}^2 \le C \hat{V}(k) \tilde{n}_{\alpha,k}^2 \rho^{-\frac 23} L^{-3} \le C_V \tilde\lambda_{\alpha,k} M^{-1} $ by~\eqref{eq:hP} and~\eqref{eq:ntilde}, as well as the identity $ \int_0^\infty t^2 (t^2 + \lambda^2)^{-2} \; \d t = \frac{\pi}{4 \lambda} $. Thus, $ \frac{1}{1-\Vert B(k) \Vert} $ is uniformly bounded in $ (\rho,L) $.\\
Moreover, for $ |k| > R $ or $ k \notin \supp(\hat\chi_{>,\alpha}) \cap \supp(\hat\chi_{>,\beta}) $, Lemma~\ref{lem:regularity_B} yields $ B(k)_{\alpha,\beta} = 0 $, so also $ K(k)_{\alpha,\beta} = 0 $. Hence, also setting $ K(k)_{\alpha,\beta} = 0 $ for $ k \notin \supp(\hat\chi_R) $~\eqref{eq:DB} does not affect smoothness. Next, we write
\begin{equation}
	\hat\varphi_K(k)_{\alpha,\beta}
	= - \sum_{n=1}^\infty \frac{1}{2n}
	\frac{L^{\frac 32} \rho^{\frac 13}}{\tilde{n}_{\alpha,k}}	
		B^n(k)_{\alpha,\beta}
		\frac{L^{\frac 32} \rho^{\frac 13}}{\tilde{n}_{\beta,k}} \;,
\end{equation}
so for each $ n $ we have $ n + 2 $ factors: $ \frac{L^{\frac 32} \rho^{\frac 13}}{\tilde{n}_{\alpha,k}} $, $ \frac{L^{\frac 32} \rho^{\frac 13}}{\tilde{n}_{\beta,k}} $, and $ n $ factors of $ B $. Evaluating $ \partial_k^{\um} \hat\varphi_K(k)_{\alpha,\beta} $ by the product rule from Lemma~\ref{lem:product_chain_rule}, each $ n $ will produce contributions with at least $ m $ and at most $ (n+2)m $ derivatives, where each excess derivative comes with a factor of $ L^{-1} $. We partition those contributions according to the number $ 0 \le r \le n $ of factors $ B(k) $ that are hit by at least one derivative, as well as the numbers $ s_1, \ldots, s_{r+1} $ of non-hit factors $ B(k) $ between those ones hit. For the $ r+2 $ factors hit by derivatives, let $ \cM_r $ denote the set of multi-indices $ (\um_0,\ldots,\um_{r+1}) $ tracking the derivatives on each factor, and observe that $ |\cM_r| \le 2^{(r+2)m} $, since when applying the product rule from Lemma~\ref{lem:product_chain_rule}, we can choose for each of the $ (r+2) $ terms and each of the $ m $ derivatives, if the derivative is applied to the term or not. Then,
\begin{equation}
	| \partial_k^{\um} \hat\varphi_K(k)_{\alpha,\beta} | 
	\le \sum_{n=1}^\infty \sum_{r=0}^n \I_{n,r} \;,
\end{equation}
where for $ r \ge 2 $,
\begin{equation}
\begin{aligned}
	\I_{n,r} &:= C
		\sum_{(\um_0,\ldots,\um_{r+1}) \in \cM_r}
		\sum_{\substack{s_1, \ldots, s_{r+1} \in \NNN_0 \\ s_1 + \ldots + s_{r+1} = n-r}}
		L^{m-|\um_0 +\ldots + \um_{r+1}|} 		 
		\Bigg\vert \partial_k^{\um_0} \frac{L^{\frac 32} \rho^{\frac 13}}{\tilde{n}_{\alpha,k}} \Bigg\vert \\
	&\quad \times \Bigg\vert
		\sum_{\alpha_1,\ldots,\alpha_r = 1}^M (B^{s_1} \partial_k^{\um_1} B)(k)_{\alpha,\alpha_1}
		\prod_{j=2}^r
		(B^{s_j} \partial_k^{\um_j} B)(k)_{\alpha_{j-1},\alpha_j}
		\; B^{s_{r+1}}(k)_{\alpha_r,\beta}
		\Bigg\vert
		\Bigg\vert \partial_k^{\um_{r+1}} \frac{L^{\frac 32} \rho^{\frac 13}}{\tilde{n}_{\beta,k}} \Bigg\vert \\
	&\le C_m \sum_{(\um_0,\ldots,\um_{r+1}) \in \cM_r}
		\sum_{\substack{s_1, \ldots, s_{r+1} \in \NNN_0 \\ s_1 + \ldots + s_{r+1} = n-r}} \!\!
		L^{m-|\um_0 +\ldots + \um_{r+1}|} 
		\Bigg(\sum_{\alpha_1=1}^M
		|(B^{s_1} \partial_k^{\um_1} B)(k)_{\alpha,\alpha_1}|^2 \Bigg)^{\frac 12}\\
	&\quad \times
		\prod_{j=2}^{r-1}
		\Bigg(\sum_{\alpha_{j-1}, \alpha_j=1}^M |(B^{s_j} \partial_k^{\um_j} B)(k)_{\alpha_{j-1},\alpha_j}|^2 \Bigg)^{\frac 12}
		\Bigg(\sum_{\alpha_{r-1}=1}^M |(B^{s_r} (\partial_k^{\um_r} B) B^{s_{r+1}})(k)_{\alpha_{r-1},\beta}|^2 \Bigg)^{\frac 12} \\
	&\quad \times
		M \fl^{-|\um_0 + \um_{r+1}|} \tilde\lambda_{\alpha,k}^{-\frac 12} \tilde\lambda_{\beta,k}^{-\frac 12} \\
\end{aligned}
\end{equation}
where we controlled $ |\partial_k^{\um} \tilde{n}_{\alpha,k}^{-1} L^{\frac 32} \rho^{\frac 13}| $ by Lemma~\ref{lem:regularity_lambda} and $ \frac{L^{\frac 32} \rho^{\frac 13}}{\tilde{n}_{\alpha,k}}  \le C M^{\frac 12} \tilde\lambda_{\alpha,k}^{-\frac 12} $ by~\eqref{eq:ntilde}, and used the Cauchy--Schwarz inequality. We now bound, using Lemmas~\ref{lem:regularity_lambda} and~\ref{lem:regularity_B}:
\begin{equation} \label{eq:Bsj_estimate}
\begin{aligned}
	&\sum_{\alpha, \beta=1}^M |(B^{s_j} \partial_k^{\um_j} B)(k)_{\alpha,\beta}|^2
	\le \Vert B \Vert^{2 s_j} \sum_{\alpha, \beta=1}^M |(\partial_k^{\um_j} B)(k)_{\alpha,\beta}|^2 \\
	&= \Vert B \Vert^{2 s_j} \sum_{\alpha, \beta=1}^M
		\Big\vert \partial_k^{\um_j} \frac{\tilde{n}_{\alpha,k} \tilde{n}_{\beta,k}}{L^3 \rho^{\frac 23}} \hat\varphi_B(k)_{\alpha,\beta} \Big\vert^2
	\le C_{V,m} \fl^{-2|\um_j|} \Vert B \Vert^{2 s_j} \sum_{\alpha, \beta=1}^M
		\Big\vert \frac{\tilde{n}_{\alpha,k} \tilde{n}_{\beta,k}}{L^3 \rho^{\frac 23}} \frac{1}{\tilde\lambda_{\alpha,k} + \tilde\lambda_{\beta,k}} \Big\vert^2 \\
	&\le C_{V,m} \fl^{-2|\um_j|} \Vert B \Vert^{2 s_j} \sum_{\alpha, \beta=1}^M M^{-2}
	\le C_{V,m} \fl^{-2|\um_j|} \Vert B \Vert^{2 s_j} \;.
\end{aligned}
\end{equation}
Next, we consider the term with $ \partial_k^{\um_r} B $. If $ s_{r+1} = 0 $, proceeding as above yields
\begin{equation}
	\sum_{\alpha_{r-1}=1}^M |(B^{s_r} \partial_k^{\um_r} B)(k)_{\alpha_{r-1},\beta}|^2
	\le C_{V,m} \fl^{-2|\um_r|} \Vert B \Vert^{2 s_r} M^{-1} \;.
\end{equation}
If $ s_{r+1} \ge 1 $, we split by Cauchy--Schwarz:
\begin{equation} \label{eq:Bsr_estimate}
\begin{aligned}
	&\Bigg(\sum_{\alpha_{r-1}=1}^M |(B^{s_r} (\partial_k^{\um_r} B) B^{s_{r+1}})(k)_{\alpha_{r-1},\beta}|^2 \Bigg)^{\frac 12} \\
	&\le \Bigg(\sum_{\alpha_{r-1}, \alpha_r=1}^M |(B^{s_r} \partial_k^{\um_r} B)(k)_{\alpha_{r-1},\alpha_r}|^2 \Bigg)^{\frac 12}
		\Bigg(\sum_{\alpha_r=1}^M |B^{s_{r+1}}(k)_{\alpha_r,\beta}|^2 \Bigg)^{\frac 12} \\
	&\le C_{V,m} \fl^{-|\um_r|} \Vert B \Vert^{s_r} 
		\Bigg(\sum_{\alpha_r=1}^M |(B^{s_{r+1}-1} \partial_k^0 B)(k)_{\alpha_r,\beta}|^2 \Bigg)^{\frac 12} \\
	&\le C_{V,m} \fl^{-|\um_r|} \Vert B \Vert^{s_r + s_{r+1} - 1} M^{-\frac 12} \;. 
\end{aligned}
\end{equation}
The term with $ \partial_k^{\um_1} B $ is then controlled as the term with $ \partial_k^{\um_r} B $ for $ s_r = 0 $. Overall, using also $ \Vert B \Vert \le 1 $, we obtain
\begin{equation} \label{eq:Inr_estimate}
\begin{aligned}
	\I_{n,r}
	&\le C_{V,m}^{r+1}
		\!\! \sum_{(\um_0,\ldots,\um_{r+1}) \in \cM_r} \!\!
		L^{m-|\um_0 +\ldots + \um_{r+1}|}
		\frac{\fl^{-|\um_0 +\ldots + \um_{r+1}|}}{\tilde\lambda_{\alpha,k}^{\frac 12} \tilde\lambda_{\beta,k}^{\frac 12}} 
		\!\! \sum_{\substack{s_1, \ldots, s_{r+1} \in \NNN_0 \\ s_1 + \ldots + s_{r+1} = n-r}} \!\!
		\min\{1,\Vert B \Vert^{n-r-2}\} \\
	&\le C_{V,m}^{r+1}
		\sum_{(\um_0,\ldots,\um_{r+1}) \in \cM_r}
		L^{m-|\um_0 +\ldots + \um_{r+1}|}
		\frac{\fl^{-|\um_0 +\ldots + \um_{r+1}|}}{\tilde\lambda_{\alpha,k}^{\frac 12} \tilde\lambda_{\beta,k}^{\frac 12}}
		{n \choose r}
		\min\{1,\Vert B \Vert^{n-r-2}\} \;,
\end{aligned}
\end{equation}
where we used that there are $ {n \choose r} $ ways to select the $ r $ out of $ n $ factors $ B(k) $ which are hit by derivatives. The cases $ r \in \{0,1\} $ have to be treated separately: For $ r = 1 $:
\begin{equation}
\begin{aligned}
	\I_{n,1} &:=
	 C_m \sum_{(\um_0,\um_1,\um_2) \in \cM_1}
		\sum_{s_1=0}^{n-1} 
		L^{m-|\um_0 + \um_1 + \um_2|} 
		|(B^{s_1} (\partial_k^{\um_1} B) B^{n-1-s_1})(k)_{\alpha,\beta}| \\
	&\quad \times 
		M \fl^{-|\um_0 + \um_2|} \tilde\lambda_{\alpha,k}^{-\frac 12} \tilde\lambda_{\beta,k}^{-\frac 12} \;, \\
\end{aligned}
\end{equation}
where the term with $ \partial_k^{\um_1} B $ is estimated as above: For $ 1 \le s_1 \le n-2 $, we split
\begin{equation} \label{eq:BBB_bound}
\begin{aligned}
	&|(B^{s_1} (\partial_k^{\um_1} B) B^{n-1-s_1})(k)_{\alpha,\beta}| \\
	&\le \Bigg( \sum_{\alpha_1=1}^M |B^{s_1}(k)_{\alpha,\alpha_1}|^2 \Bigg)^{\frac 12}
		\Bigg( \sum_{\alpha_1,\alpha_2=1}^M |(\partial_k^{\um_1} B)(k)_{\alpha_1,\alpha_2}|^2 \Bigg)^{\frac 12}
		\Bigg( \sum_{\alpha_2=1}^M |B^{n-1-s_1}(k)_{\alpha_2,\beta}|^2 \Bigg)^{\frac 12} \\
	&\le C_{V,m} \fl^{-|\um_1|} \Vert B \Vert^{n-3} M^{-1} \;,
\end{aligned}
\end{equation}
where the first and third bracket was bounded as in~\eqref{eq:Bsr_estimate} and the second bracket as in~\eqref{eq:Bsj_estimate} with $ s_j = 0 $. The final estimate for $ \I_{n,1} $ then agrees with~\eqref{eq:Inr_estimate}. For $ r = 0 $,
\begin{equation}
\begin{aligned}
	\I_{n,0} &:=
	 C_m \sum_{(\um_0,\um_1) \in \cM_0}
		L^{m-|\um_0 + \um_1|} 
		|B^n(k)_{\alpha,\beta}| 
		M \fl^{-|\um_0 + \um_1|} \tilde\lambda_{\alpha,k}^{-\frac 12} \tilde\lambda_{\beta,k}^{-\frac 12} \;, \\
\end{aligned}
\end{equation}
where $ |B^n(k)_{\alpha,\beta}| $ is bounded as in~\eqref{eq:BBB_bound} with $ \um_1 = 0 $, rendering a final bound for $ \I_{n,0} $ that agrees with~\eqref{eq:Inr_estimate}.
Summing over all $ n $ and $ r $, we then obtain
\begin{equation}
\begin{aligned}
	&| \partial_k^{\um} \hat\varphi_K(k)_{\alpha,\beta} | \\
	&\le \sum_{r=0}^\infty
		C_{V,m}^{r+1}
		\sum_{(\um_0,\ldots,\um_{r+1}) \in \cM_r}
		L^{m-|\um_0 +\ldots + \um_{r+1}|}
		\frac{\fl^{-|\um_0 +\ldots + \um_{r+1}|}}{\tilde\lambda_{\alpha,k}^{\frac 12} \tilde\lambda_{\beta,k}^{\frac 12}}
		\sum_{n=r}^\infty {n \choose r}
		\min\{1,\Vert B \Vert^{n-r-2}\} \;.
\end{aligned}
\end{equation}
Now $ \sum_{n=r}^\infty {n \choose r} \Vert B \Vert^{n-r} = (1 - \Vert B \Vert)^{-r-1} \le C_V^{r+1} $, so
\begin{equation}
	| \partial_k^{\um} \hat\varphi_K(k)_{\alpha,\beta} | 
	\le \sum_{r=0}^\infty
		C_{V,m}^{r+1}
		\sum_{(\um_0,\ldots,\um_{r+1}) \in \cM_r}
		L^{m-|\um_0 +\ldots + \um_{r+1}|}
		\fl^{-|\um_0 +\ldots + \um_{r+1}|}
		\tilde\lambda_{\alpha,k}^{-\frac 12} \tilde\lambda_{\beta,k}^{-\frac 12} \;.
\end{equation}
Next, observe that the number of excess derivatives is bounded as $ s := |\um_0 +\ldots + \um_{r+1}|-m \ge r-m $ and $ s \ge 0 $. Then, recalling $ |\cM_r| \le 2^{(r+2)m} $, we get
\begin{equation}
\begin{aligned}
	&| \partial_k^{\um} \hat\varphi_K(k)_{\alpha,\beta} | \\
	&\le \tilde\lambda_{\alpha,k}^{-\frac 12} \tilde\lambda_{\beta,k}^{-\frac 12}
	\Bigg( \sum_{r=0}^m
		C_{V,m}^{r+1}
		2^{(r+2)m}
		\fl^{-m}
	+ \sum_{r=m+1}^\infty
		C_{V,m}^{r+1}
		\sum_{s=r-m}^{(r+1)m}
		2^{(r+2)m}
		L^{-s} \fl^{-m-s} \Bigg) \\
	&\le \tilde\lambda_{\alpha,k}^{-\frac 12} \tilde\lambda_{\beta,k}^{-\frac 12}
		\Bigg( C_{V,m} \fl^{-m} 
		+ \sum_{r=m+1}^\infty
		C_{V,m}^{r+1}
		(r+1)m \; 
		2^{(r+2)m}
		L^{-r+m} \fl^{-r} \Bigg) \\
	&\le C_{V,m} \fl^{-m} \tilde\lambda_{\alpha,k}^{-\frac 12} \tilde\lambda_{\beta,k}^{-\frac 12} \;,
\end{aligned}
\end{equation}
where we used that the sum in $ r $ converges to an arbitrarily small constant for $ L $ large enough.\\
The bound for $ \hat\varphi_{K_j} $ is analogous, where the product rule additionally produces derivatives $ |\partial^{\um_{r+2}}_k p_j(k)| $, $ |\um_{r+2}| \le m $, where $ \partial^{\um_{r+2}}_k p_j(k) $ is again a polynomial in $ k $. Since $ K(k)_{\alpha,\beta} $ is supported for $ |k| \le R $, this polynomial is uniformly bounded as $ |\partial^{\um_{r+2}}_k p_j(k)| \le C_{m,p} $.
\end{proof}

\subsection{Deriving Pointwise Bounds for $I$, $h$, and $P$}
\label{subsec:hPI_Pointwise}

\begin{lemma}[Pointwise bounds on $ I $, $ h $ and $ P $]
\label{lem:regularity_IhP}
Let $ |k| \le 2 R $ and recall the definition \eqref{eq:cEalphahat} of $ I(k) $, as well as $ h(k)$ and  $P(k) $ from~\eqref{eq:hP} and $ \hat\varphi_X(k)_{\alpha,\beta} $ from~\eqref{eq:varphi}.
Then, given $ m \in \NNN_0 $ and any multi-index $ \um \in \NNN_0^3 $ with $ |\um| = m $, there exist $ C_m, C_{V,m} > 0 $ such that for all $ k \in \Lambda_L $ and $ \alpha, \beta \in \{1, \ldots, M\} $,
\begin{equation} \label{eq:regularity_IhP}
\begin{aligned}
	|\partial^{\um}_k (\tilde{n}_{\alpha,k}^2 I(k)_{\alpha,\alpha}) L^{-3} \rho^{-\frac 23}|
	&\le C_m \fl^{-m} M^{-1} \;, \\
	| \partial_k^{\um} \hat\varphi_h(k)_{\alpha,\beta} |
	&\le C_m M \delta_{\alpha,\beta} \;, \\
	| \partial_k^{\um} \hat\varphi_P(k)_{\alpha,\beta} |
	&\le C_{V,m} \fl^{-m} \;,
\end{aligned}
\end{equation}
and the same bounds are true if any number of right-derivatives $ \partial_{k_j} $ is replaced by left-derivatives $ \bar\partial_{k_j} $.
\end{lemma}

\begin{proof}
With~\eqref{eq:cEalphahat} we see that $ \tilde{n}_{\alpha,k}^2 I(k)_{\alpha,\alpha} = n_{\alpha,k}^2 $, whose discrete derivatives are bounded in Lemma~\ref{lem:regularity_lambda}. The first claim~\eqref{eq:regularity_IhP} then readily follows by $ \tilde\lambda_{\alpha,k} \le C $.

\medskip

To prove the bound on the diagonal matrix $ h(k) $, observe that by~\eqref{eq:hP} and~\eqref{eq:ntilde},
\begin{equation}
	| \partial_k^{\um} \hat\varphi_h(k)_{\alpha,\beta} |
	= \left\vert \partial^{\um}_k \frac{h(k)_{\alpha,\beta}}{\tilde{n}_{\alpha,k} \tilde{n}_{\beta,k}} L^3 \rho^{\frac 23} \right\vert
	= \delta_{\alpha,\beta} \left\vert \frac{(2\pi)^3 M}{4 \pi k_{\F}^2} \rho^{\frac 23} \partial^{\um}_k \hat\chi_R(k) \right\vert
	\le C_m M \delta_{\alpha,\beta} \;,
\end{equation}
where we used that $ \hat\chi_R \in C_c^\infty $, introduced above~\eqref{eq:DB}, is density-independent and supported for $ |k| \le 2 R $. Likewise, for $ P(k) $~\eqref{eq:hP},
\begin{equation}
	| \partial_k^{\um} \hat\varphi_P(k)_{\alpha,\beta} |
	= \left\vert \partial^{\um}_k \frac{P(k)_{\alpha,\beta}}{\tilde{n}_{\alpha,k} \tilde{n}_{\beta,k}} L^3 \rho^{\frac 23} \right\vert
	\le C \left\vert \partial^{\um}_k \big( \hat{V}(k) \hat\chi_{>,\alpha}(k) \hat\chi_{>,\beta}(k) \big) \right\vert
	\le C_{V,m} \fl^{-m} \;.
\end{equation}
The proof with left-derivatives is analogous.
\end{proof}

In the proof of Proposition~\ref{prop:trace_estimates}, we will also need that $ I(k) $ approximates the identity matrix.

\begin{lemma} [Approximate identity matrix] \label{lem:approxID}
Let $ |k| \le 2 R $ and recall the definition of the diagonal matrix $ I(k) $~\eqref{eq:cEalphahat}. Then, we have the bounds
\begin{equation} \label{eq:approxID}
\begin{aligned}
	|I(k)_{\alpha,\alpha} - 1|
	&\le C \fl \tilde\lambda_{\alpha,k}^{-1} \qquad
	&&\textnormal{if } k \cdot \hat\omega_\alpha \ge \fl \;, \\
	|I(k)_{\alpha,\alpha}|
	&\le C \qquad
	&&\textnormal{for any } |k| \le 2R \;.
\end{aligned}
\end{equation}
\end{lemma}

\begin{proof}
Recall~\eqref{eq:cEalphahat} that $ I(k)_{\alpha,\alpha} = n_{\alpha,k}^2 \tilde{n}_{\alpha,k}^{-2} $, so we need to compare $ \tilde{n}_{\alpha,k}^2 $ to $ n_{\alpha,k}^2 $. For $ k \cdot \hat\omega_\alpha \ge \fl $, from~\eqref{eq:ntilde} and~\eqref{eq:lambdatilde}, we see
\begin{equation} \label{eq:ntilde_bis}
	\tilde{n}_{\alpha,k}^2
	= \frac{L^3}{(2 \pi)^3} \frac{4 \pi k_{\F}^2}{M} \tilde\lambda_{\alpha,k}
	= \frac{L^3}{(2 \pi)^3} \frac{4 \pi k_{\F}^2}{M} (k \cdot \hat\omega_\alpha) \;,
\end{equation}
whereas $ n_{\alpha,k}^2 $ according to~\eqref{eq:n} is the regularized number of pairs in a patch:
\[
	n_{\alpha,k}^2
	= \sum_{p \in \Lambda_L} \hat{u}^\alpha(p)^2 \hat{v}^\alpha(p-k)^2  \;.
\]
First, we approximate the sum by a Riemann integral, using that $ \hat{u}^{\alpha} , \hat{v}^{\alpha} \le 1 $, and that the surface area of the integration domain $ S_{\alpha,k} := B_\alpha \cap (B_\alpha + k) \cap \overline{B_{k_{\F}}(0)}^c \cap (\overline{B_{k_{\F}}(0)} + k) $ is bounded by $ \le C \rho^{\frac 23} M^{-1} $:
\begin{equation}
	\Big\vert n_{\alpha,k}^2 - \frac{L^3}{(2 \pi)^3} \int_{S_{\alpha,k}} \d p \; \hat{u}^\alpha(p)^2 \hat{v}^\alpha(p-k)^2 \Big\vert
	\le C \rho^{\frac 23} M^{-1} L^2 \;.
\end{equation}
Next, $ \hat{u}^\alpha $ and $ \hat{v}^\alpha $ only differ from 1 in a layer of thickness $ \fl $ above the boundary of the integration domain, so with $ L^{-1} \ll \fl $
\begin{equation} \label{eq:n2_Vol_bound}
	\Big\vert n_{\alpha,k}^2 - \frac{L^3}{(2 \pi)^3} \mathrm{Vol}(S_{\alpha,k}) \Big\vert
	\le C \rho^{\frac 23} M^{-1} \fl L^3 \;.
\end{equation}
The volume of $ S_{\alpha,k} $ can be estimated by a similar argument as in~\cite[Sect.~6]{BNPSS20}: Let $ P_\alpha := B_\alpha \cap \partial B_{k_{\F}}(0) $ be the patch surface on the Fermi sphere, and observe that by the patch geometry, for any $ p \in S_{\alpha,k} $, the line between $ p $ and $ p-k $ has to intersect $ P_\alpha $. Therefore, $ S_{\alpha,k} \subseteq S_{\alpha,k}^+ := \{ q + t k : q \in P_\alpha, t \in [0,1] \} $. Moreover, if $ \hat\omega_p $ is the surface normal vector at $ p $, and $ \d \sigma $ the surface element, then for $ k \cdot \hat\omega_\alpha \ge \fl $,
\begin{equation} \label{eq:Splus_bound}
	\mathrm{Vol}(S_{\alpha,k}^+)
	= \int_{P_\alpha} k \cdot \hat\omega_p \; \d \sigma(p)
	= \mathrm{Area}(P_\alpha) \; (k \cdot \hat\omega_\alpha + \cO(M^{-\frac 12}))
	= \frac{4 \pi k_{\F}^2}{M} k \cdot \hat\omega_\alpha  (1 + \cO(\fl)) \;.
\end{equation}
Here, we used that the surface normal vector $ \hat\omega_p $ varies by an amount $ \le C M^{-\frac 12} $ within $ P_\alpha $, that $ M^{-\frac 12} \ll \fl^2 \le \fl (k \cdot \hat\omega_\alpha) $ and that $ \mathrm{Area}(P_\alpha) = \frac{4 \pi k_{\F}^2}{M} (1 - \cO(\rho^{-\frac 13} M^{\frac 12})) $, since the removed corridors are of width $ \sim 1 $ and the edge length of the patch is $ \sim \rho^{\frac 13} M^{-\frac 12} $.\\
The set $ S_{\alpha,k}^+ $ then agrees with $ S_{\alpha,k} $ up to additional contributions at the corridors, which are of size
\begin{equation} \label{eq:Vol_S_alphak_final}
	|\mathrm{Vol}(S_{\alpha,k}^+) - \mathrm{Vol}(S_{\alpha,k})| \le C \rho^{\frac 13} M^{-\frac 12}
	\quad \Rightarrow \quad
	\mathrm{Vol}(S_{\alpha,k})
	= \frac{4 \pi k_{\F}^2}{M} k \cdot \hat\omega_\alpha (1 + \cO(\fl)) \;.
\end{equation}
Plugging~\eqref{eq:Vol_S_alphak_final} into~\eqref{eq:n2_Vol_bound} and inserting the identity~\eqref{eq:ntilde_bis} for $ \tilde{n}_{\alpha,k}^2 $, we conclude for $ k \cdot \hat\omega_\alpha \ge \fl $:
\begin{equation}
	|n_{\alpha,k}^2 - \tilde{n}_{\alpha,k}^2|
	\le C \rho^{\frac 23} M^{-1} \fl L^3 \;.
\end{equation}
The first claim in~\eqref{eq:approxID} then follows from $ |I(k)_{\alpha,\alpha} - 1| = |n_{\alpha,k}^2 - \tilde{n}_{\alpha,k}^2| |\tilde{n}_{\alpha,k}^{-2}| $, by using~\eqref{eq:ntilde_bis}.\\
For the second claim, observe that the definition of $ f_>(x) $~\eqref{eq:f} implies $ \tilde\lambda_{\alpha,k} \ge \frac 12 |k \cdot \hat\omega_\alpha| $, whenever $ k \cdot \hat\omega_\alpha \ge - \fl $, which is needed for $ S_{\alpha,k} $ to be non-empty. Then,
\begin{equation}
	\mathrm{Vol}(S_{\alpha,k})
	\le \mathrm{Vol}(S_{\alpha,k}^+)
	\le \int_{P_\alpha} |k \cdot \hat\omega_p| \; \d \sigma(p)
	\le \int_{P_\alpha} \big(|k \cdot \hat\omega_\alpha| + \cO(M^{-\frac 12}) \big) \d \sigma(p)
	\le C \rho^{\frac 23} M^{-1} \tilde\lambda_{\alpha,k} \;.
\end{equation}
Also, the bound~\eqref{eq:n2_Vol_bound} is still valid here, so with $ \fl/2 \le \tilde\lambda_{\alpha,k} $,
\begin{equation}
\begin{aligned}
	|I(k)_{\alpha,\alpha}|
	&= \frac{n_{\alpha,k}^2}{\tilde{n}_{\alpha,k}^2}
	\le \frac{C L^3 \mathrm{Vol}(S_{\alpha,k}) + C \rho^{\frac 23} M^{-1} \fl L^3}{\rho^{\frac 23} M^{-1} L^3 \tilde\lambda_{\alpha,k}} 
	\le C\;.
\end{aligned}
\end{equation}
\end{proof}

\subsection{Pointwise Bounds for $\tilde{K}^{(n)}$, $\tilde{h}^{(n)}$, and $\tilde{P}^{(n)}$}
\label{subsec:X_Pointwise}

We proceed to pointwise bounds for the other matrix-valued functions appearing in Proposition~\ref{prop:regularity}.

\begin{lemma}[Pointwise bounds on $\tilde{K}^{(n)}$, $\tilde{h}^{(n)}$, and $\tilde{P}^{(n)}$] \label{lem:regularity_Xtilde}
Recall $ K $ from~\eqref{eq:K}, and $ h$ and $P $ from~\eqref{eq:hP}, as well as  $ \hat\varphi_X(k)_{\alpha,\beta} $ from~\eqref{eq:varphi} and $ \tilde{X}^{(n)} $ from~\eqref{eq:Xtilde} for $ n \in \NNN_0 $. Then, given $ m \in \NNN_0 $, any multi-index $ \um \in \NNN_0^3 $ with $ |\um| = m $, there exists $ C_{V,m} > 0 $ such that for all $ k \in \Lambda_L $, $ n \in \NNN_0 $, and $ \alpha, \beta \in \{1, \ldots, M\} $
\begin{equation} \label{eq:regularity_Xtilde}
\begin{aligned}
	| \partial_k^{\um} \hat\varphi_{\tilde{X}^{(n)}}(k)_{\alpha,\beta} |
	&\le (C_{V,m})^{n+1} \fl^{-m} \tilde\lambda_{\alpha,k}^{-\frac 12} \tilde\lambda_{\beta,k}^{-\frac 12} \;, \qquad \qquad
		&&X \in \{K, P, h(1-I)K, K(1-I)h \} \;, \\
		| \partial_k^{\um} \hat\varphi_{\tilde{h}^{(n)}}(k)_{\alpha,\beta} |
	&\le (C_{V,m})^{n+1} \fl^{-m} \tilde\lambda_{\alpha,k}^{-\frac 12} \tilde\lambda_{\beta,k}^{-\frac 12} \;, \qquad \qquad
		&&\mathrm{if } \; n \ge 1 \;, \\
		| \partial_k^{\um} \hat\varphi_{\tilde{h}^{(0)}}(k)_{\alpha,\beta} |
	&\le C_m M \delta_{\alpha,\beta} \;.
\end{aligned}
\end{equation}
and the same bounds are true if any number of right-derivatives $ \partial_k $ is replaced by left-derivatives $ \bar\partial_k $. Moreover, $ \tilde{K}^{(n)}(k)_{\alpha,\beta} = 0 $ if $ |k| > R $ or $ k \notin \supp(\hat\chi_{>,\alpha}) \cap \supp(\hat\chi_{>,\beta}) $. 
\end{lemma}

\begin{proof}
First, note that by definition, $ h(k), P(k), K(k) $ and therefore also $ \tilde{h}^{(n)}(k) $, $ \tilde{P}^{(n)}(k) $, $ \tilde{K}^{(n)}(k) $ vanish if $ |k| > 2R $, so the previous lemmas of this section apply. Also, if $ |k| > R $ or $ k \notin \supp(\hat\chi_{>,\alpha}) \cap \supp(\hat\chi_{>,\beta}) $ then by Lemma~\ref{lem:regularity_K}, $ K(k)_{\alpha,\beta} = 0 $, so $ \tilde{K}^{(n)}(k)_{\alpha,\beta} = 0 $ immediately follows.

Let us set $ \tilde{X}^{(n,n')} := (KI)^{n'} X (IK)^{n-n'} $, so $ \tilde{X}^{(n)} = \sum_{n' = 0}^n {n \choose n'} \tilde{X}^{(n,n')} $. Then, to bound $ | \partial_k^{\um} \hat\varphi_{\tilde{X}^{(n)}}(k)_{\alpha,\beta} | $, it suffices to bound $ | \partial_k^{\um} \hat\varphi_{\tilde{X}^{(n,n')}}(k)_{\alpha,\beta} | $. For $ X = K $, we proceed using the product rule from Lemma~\ref{lem:product_chain_rule} and Lemmas~\ref{lem:regularity_lambda} and~\ref{lem:regularity_K}. Here, hitting the product of $ (2n+1) $ terms with $ m $ derivatives will produce $ \le 2^{m(2n+1)} $ contributions:
\begin{equation}
\begin{aligned}
	&| \partial_k^{\um} \hat\varphi_{\tilde{K}^{(n,n')}}(k)_{\alpha,\beta} |
	= \left\vert \partial^{\um}_k \frac{\big((KI)^{n'} K(IK)^{n-n'}\big)(k)_{\alpha,\beta}}{\tilde{n}_{\alpha,k} \tilde{n}_{\beta,k}} L^3 \rho^{\frac 23} \right\vert \\
	&\le \Bigg\vert \partial^{\um}_k \sum_{\alpha_1, \ldots, \alpha_n =1}^M
		\hat\varphi_K(k)_{\alpha,\alpha_1}
		\frac{\tilde{n}_{\alpha_1,k}^2 I(k)_{\alpha_1,\alpha_1}}{L^3 \rho^{\frac 23}}
		\hat\varphi_K(k)_{\alpha_1,\alpha_2}
		\ldots
		\frac{\tilde{n}_{\alpha_n,k}^2 I(k)_{\alpha_n,\alpha_n}}{L^3 \rho^{\frac 23}}
		\hat\varphi_K(k)_{\alpha_n,\beta} \Bigg\vert \\
	&\le C 2^{m(2n+1)} \fl^{-m} \sum_{\alpha_1, \ldots, \alpha_n =1}^M
		\frac{C_{V,m}}{\tilde\lambda_{\alpha,k}^{\frac 12} \tilde\lambda_{\alpha_1,k}^{\frac 12}}
		\frac{\tilde\lambda_{\alpha_1,k}}{M}
		\frac{C_{V,m}}{\tilde\lambda_{\alpha_1,k}^{\frac 12} \tilde\lambda_{\alpha_2,k}^{\frac 12}}
		\ldots
		\frac{\tilde\lambda_{\alpha_n,k}}{M}
		\frac{C_{V,m}}{\tilde\lambda_{\alpha_n,k}^{\frac 12} \tilde\lambda_{\beta,k}^{\frac 12}} \\
	&\le C (C_{V,m})^n \fl^{-m} \tilde\lambda_{\alpha,k}^{-\frac 12} \tilde\lambda_{\beta,k}^{-\frac 12} \;.
\end{aligned}
\end{equation}
The bound for $ X = h $, $ n \ge 1 $ is similar, using Lemma~\ref{lem:regularity_IhP}:
\begin{equation}
\begin{aligned}
	&| \partial_k^{\um} \hat\varphi_{\tilde{h}^{(n,n')}}(k)_{\alpha,\beta} | \\
	&\le \Bigg\vert \partial^{\um}_k \sum_{\alpha_1, \ldots, \alpha_n =1}^M
		\hat\varphi_K(k)_{\alpha,\alpha_1}
		\frac{\tilde{n}_{\alpha_1,k}^2 I(k)_{\alpha_1,\alpha_1}}{L^3 \rho^{\frac 23}}
		\ldots
		\hat\varphi_h(k)_{\alpha_{n'},\alpha_{n'+1}}
		\ldots
		\frac{\tilde{n}_{\alpha_n,k}^2 I(k)_{\alpha_n,\alpha_n}}{L^3 \rho^{\frac 23}}
		\hat\varphi_K(k)_{\alpha_n,\beta} \Bigg\vert \\
	&\le C 2^{m(2n+1)} \fl^{-m} \sum_{\alpha_1, \ldots, \alpha_n =1}^M
		\frac{C_{V,m}}{\tilde\lambda_{\alpha,k}^{\frac 12} \tilde\lambda_{\alpha_1,k}^{\frac 12}}
		\frac{\tilde\lambda_{\alpha_1,k}}{M}
		\ldots
		C_m M \delta_{\alpha_{n'},\alpha_{n'+1}}
		\ldots
		\frac{\tilde\lambda_{\alpha_n,k}}{M}
		\frac{C_{V,m}}{\tilde\lambda_{\alpha_n,k}^{\frac 12} \tilde\lambda_{\beta,k}^{\frac 12}} \\
	&\le C (C_{V,m})^n \fl^{-m} \tilde\lambda_{\alpha,k}^{-\frac 12} \tilde\lambda_{\beta,k}^{-\frac 12} \;.
\end{aligned}
\end{equation}
The case $ n = 0 $ is covered by Lemma~\ref{lem:regularity_IhP} since $ \tilde{h}^{(0)} = h $.
The bound for $ X = P $ is analogous, where Lemma~\ref{lem:regularity_IhP} requires replacing $ C_m \delta_{\alpha_{n'},\alpha_{n'+1}} M $ by $ C_{V,m} $ in the third line.\\
Also, for $ X = h (1-I) K $ and $ X = K (1-I) h $, the bounds are analogous, using that the first bound in Lemma~\ref{lem:regularity_IhP} remains valid if $ I $ is replaced by $ (1-I) $: We have $ |\partial_k^{\um} (\tilde{n}_{\alpha,k}^2 (1-I)(k)_{\alpha,\alpha})| = |\partial_k^{\um} (\tilde{n}_{\alpha,k}^2  - n_{\alpha,k}^2)| \le |\partial_k^{\um} \tilde{n}_{\alpha,k}^2| + |\partial_k^{\um} n_{\alpha,k}^2| $, and with Lemma~\ref{lem:regularity_lambda} and~\eqref{eq:ntilde}, both terms are bounded as $ \le C_m \fl^{-m} L^3 \rho^{\frac 23} M^{-1} \tilde\lambda_{\alpha,k} $.
\end{proof}

\subsection{Proof of Proposition \ref{prop:regularity}}
\label{subsec:proof_regularitybounds}

\begin{proof}[Proof of Proposition \ref{prop:regularity}]
We start with the bound~\eqref{eq:regularity_KhP_final} for $ \hat\varphi_{K,\alpha,\beta} $ from~\eqref{eq:varphi}. Recall from \eqref{eq:ell_norms} that the $ \ell^1 $-norm carries a factor of $ L^{-3} $ in front of $ \sum_k $, and that $ m $ is even, so $ \frac m2 $ is integer. Lemma~\ref{lem:regularity_K} then implies
\begin{equation}
\begin{aligned}
	\Vert (1 - \Delta_k)^{\frac m2} \hat\varphi_{K,\alpha,\beta} \Vert_1
	& = \frac{1}{L^3} \sum_{k \in \Lambda_L} | (1 - \Delta_k)^{\frac m2} \hat\varphi_K(k)_{\alpha,\beta} | \\
	&\le C_{V,m} \fl^{-m} \frac{1}{L^3} \sum_{k \in \Lambda_L}
		\frac{\chi(\mathrm{dist}(k, \supp(\hat V \hat\chi_{>,\alpha} \hat\chi_{>,\beta})) \le 2 \pi m L^{-1} )}{\tilde\lambda_{\alpha,k}^{\frac 12} \tilde\lambda_{\beta,k}^{\frac 12}} 
	\\
	& \le C_{V,m}
	\fl^{-m} \int_{\supp(\hat V \hat\chi_{>,\alpha} \hat\chi_{>,\beta})} \tilde\lambda_{\alpha,k}^{-\frac 12} \tilde\lambda_{\beta,k}^{-\frac 12} \; \d k 
		+ o_L(1) \;,
\end{aligned}
\end{equation}
where we used that each discrete derivative enlarges the support by at most $ 2 \pi L^{-1} $ and that the integrand is uniformly bounded in $ L $, so the error in the Riemann integral is $ o_L(1) $.
To estimate the integral, recall that $ \tilde\lambda_{\alpha,k} = k \cdot \hat\omega_\alpha $ if $ k \cdot \hat\omega_\alpha \ge \fl $, and $ \tilde\lambda_{\alpha,k} \ge \frac 12 \fl $ in general. Also, $ k \in \supp(\hat\chi_{>,\alpha}) $ implies that $ k \cdot \hat\omega_\alpha \ge - 2 \fl $. Recall that $ \supp(\hat V) \subset B_R(0) $, and let us split the integration domain $ D_{\alpha,\beta} := B_R(0) \cap \supp(\hat\chi_{>,\alpha} \hat\chi_{>,\beta}) $ into
\begin{equation}
\begin{aligned}
	S_{\mathrm{small}}
	&:= \{ k \in D_{\alpha,\beta} : |k \cdot \hat\omega_\alpha| < 2 \fl \; \mathrm{or} \; |k \cdot \hat\omega_\beta| < 2 \fl \} \;, \\
	S_{\mathrm{big},\alpha}
	&:= \{ k \in D_{\alpha,\beta} \setminus S_{\mathrm{small}} : k \cdot \hat\omega_\alpha \leq k \cdot \hat\omega_\beta \} \;, \\
	S_{\mathrm{big},\beta}
	&:= \{ k \in D_{\alpha,\beta} \setminus S_{\mathrm{small}} : k \cdot \hat\omega_\alpha > k \cdot \hat\omega_\beta \} \;.
\end{aligned}
\end{equation}
Then,
\begin{equation*}
	\Vert (1 - \Delta_k)^{\frac m2} \hat\varphi_{K,\alpha,\beta} \Vert_1
	\le C_{V,m} \fl^{-m} \left( \int_{S_{\mathrm{small}}} \fl^{-1} \; \d k 
		+ 2 \int_{S_{\mathrm{big},\alpha}} (k \cdot \hat\omega_\alpha)^{-1} \; \d k \right)
		+ o_L(1) \;.
\end{equation*}
The set $ S_{\mathrm{small}} $ is contained in the union of two flat cylinders of radius $ R $ and thickness $ 4 \fl $, so $ \mathrm{Vol}(S_{\mathrm{small}}) \le C \fl $. The integral over $ S_{\mathrm{big},\alpha} $ can be parametrized by $ s := k \cdot \hat\omega_\alpha $, where evidently $ \mathrm{Area}(\{ k \in S_{\mathrm{big},\alpha} : k \cdot \hat\omega_\alpha = s \}) \le C $, as this set is contained in the union of two parallel discs of radius $ R $. Thus,
\begin{equation}
	\Vert (1 - \Delta_k)^{\frac m2} \hat\varphi_{K,\alpha,\beta} \Vert_1
	\le C_{V,m} \fl^{-m} \left( 1
		+ \int_{2 \fl}^R s^{-1} \; \d s \right)
		+ o_L(1)
	\le C_{V,m} \fl^{-m} |\log \fl | \;,
\end{equation}
which is the first bound in~\eqref{eq:regularity_KhP_final}. The other two bounds in~\eqref{eq:regularity_KhP_final} are proven by the same strategy, where the pointwise bound in Lemma~\ref{lem:regularity_K} is replaced by the respective bound in Lemma~\ref{lem:regularity_IhP}. For $ h(k) $, we obtain
\[
	\Vert (1 - \Delta_k)^{\frac m2} \hat\varphi_{h,\alpha,\beta} \Vert_1
	\le C_{V,m} \frac{1}{L^3} \sum_{k \in \Lambda_L : |k| \le 2 R + 2 \pi m L^{-1}}
		M \delta_{\alpha,\beta}
	\le C_{V,m} M \delta_{\alpha,\beta} \;,
\]
and the analogous bound for $ P(k) $ is
\[
	\Vert (1 - \Delta_k)^{\frac m2} \hat\varphi_{P,\alpha,\beta} \Vert_1
	\le C_{V,m} \fl^{-m} \frac{1}{L^3} \sum_{k \in \Lambda_L : |k| \le R + 2 \pi m L^{-1}}
		1
	\le C_{V,m} \fl^{-m} \;.
\]

The bound~\eqref{eq:regularity_Kj_final} on $ K_j(k) $ is analogous to the bound~\eqref{eq:regularity_KhP_final} on $ K(k) $, using the second statement of Lemma~\ref{lem:regularity_K} instead of the first.\\

The bounds~\eqref{eq:regularity_KhP_tilde_final} on $ \tilde{K}^{(n)} $ and $ \tilde{X}^{(n)} $ follow by the same strategy as for $ K $, using the pointwise estimates of Lemma~\ref{lem:regularity_Xtilde}.
\end{proof}

\subsection{Proof of Proposition \ref{prop:trace_estimates}}
\label{subsec:proof_trace_estimates}

\begin{proof}[Proof of Proposition \ref{prop:trace_estimates}]
We use the pointwise estimates from Lemmas~\ref{lem:regularity_K} and~\ref{lem:approxID}, as well as $ \tilde{n}_{\alpha,k}^2 \sim L^3 \rho^{\frac 23} M^{-1} \tilde\lambda_{\alpha,k} $ from~\eqref{eq:ntilde}:
\begin{equation} \label{eq:bound_trace_N}
\begin{aligned}
	&\sum_{k \in \Lambda_L} \Big\vert \tr \Big( \big( I(k) K(k) \big)^{n+1} \Big) \Big\vert \\
	&\le \sum_{k \in \Lambda_L \cap \supp(\hat V)} \sum_{\alpha_1 ,\ldots, \alpha_{n+1} = 1}^M 
		\frac{\tilde{n}_{\alpha_1,k}^2 I(k)_{\alpha_1,\alpha_1}}{L^3 \rho^{\frac 23}}  
		|\hat\varphi_K(k)_{\alpha_1,\alpha_2}|
		\ldots
		\frac{\tilde{n}_{\alpha_{n+1},k}^2 I(k)_{\alpha_{n+1},\alpha_{n+1}}}{L^3 \rho^{\frac 23}} 
		|\hat\varphi_K(k)_{\alpha_{n+1},\alpha_1}|
	\\
	&\le (C_{V,0})^{n+1} \sum_{k \in \Lambda_L \cap \supp(\hat V)} 
		\sum_{\alpha_1 ,\ldots, \alpha_{n+1} = 1}^M 
		\frac{\tilde\lambda_{\alpha_1,k}}{M}
		\frac{1}{\tilde\lambda_{\alpha_1,k}^{\frac 12} \tilde\lambda_{\alpha_2,k}^{\frac 12}} 
		\ldots
		\frac{\tilde\lambda_{\alpha_{n+1},k}}{M}
		\frac{1}{\tilde\lambda_{\alpha_{n+1},k}^{\frac 12} \tilde\lambda_{\alpha_1,k}^{\frac 12}} \\
	&\le (C_{V,0})^{n+1} L^3
	\le (C_V)^n L^3 \;.
\end{aligned}
\end{equation}

To prove~\eqref{eq:trace_replacements}, observe that the cutoff $ \hat\chi_{>,\alpha}(k) $ in~\eqref{eq:hP} forces both $ P(k)_{\alpha,\beta} $ and $ K(k)_{\alpha,\beta} $ (see Lemma~\ref{lem:regularity_K}) to vanish whenever $ k \notin A_\alpha := \{ k \in \Lambda_L \cap \supp(\hat V) : k \cdot \hat\omega_\alpha > -2 \fl \} $. We expand
\begin{equation} \label{eq:trip_expansion}
\begin{aligned}
	&\sum_{k \in \Lambda_L} \Big\vert \tr \Big( I(k) P(k) \big( I(k) K(k) \big)^n \Big)
		- \tr \big( P(k) K(k)^n \big)
	\Big\vert \\
	&\le \sum_{\alpha_1, \ldots, \alpha_{n+1} = 1}^M
		\sum_{j = 1}^{n+1}
		\sum_{k \in A_{\alpha_j}}
		\frac{\tilde{n}_{\alpha_1,k}^2 I(k)_{\alpha_1,\alpha_1}}{L^3 \rho^{\frac 23}}
		|\hat\varphi_P(k)_{\alpha_1,\alpha_2}|
		\frac{\tilde{n}_{\alpha_2,k}^2 I(k)_{\alpha_2,\alpha_2}}{L^3 \rho^{\frac 23}}
		|\hat\varphi_K(k)_{\alpha_2,\alpha_3}| \ldots \\
	&\quad \times \ldots \frac{\tilde{n}_{\alpha_j,k}^2 |I(k)_{\alpha_j,\alpha_j} - 1|}{L^3 \rho^{\frac 23}}
		|\hat\varphi_K(k)_{\alpha_j,\alpha_{j+1}}|
		\frac{\tilde{n}_{\alpha_{j+1},k}^2}{L^3 \rho^{\frac 23}}
		\ldots
		\frac{\tilde{n}_{\alpha_{n+1},k}^2}{L^3 \rho^{\frac 23}}
		|\hat\varphi_K(k)_{\alpha_{n+1},\alpha_1}| \;.
\end{aligned}
\end{equation}
Using again $ \tilde{n}_{\alpha,k}^2 \sim L^3 \rho^{\frac 23} \tilde\lambda_{\alpha,k} M^{-1} $ and bounding $ \hat\varphi_K(k) $, $ \hat\varphi_P(k) $ and $ I(k) $ by Lemmas~\ref{lem:regularity_K}, \ref{lem:regularity_IhP}, and \ref{lem:approxID}:
\begin{equation}
\begin{aligned}
	&\sum_{k \in \Lambda_L} \Big\vert \tr \Big( I(k) P(k) \big( I(k) K(k) \big)^n \Big)
		- \tr \big( P(k) K(k)^n \big)
	\Big\vert \\
	&\le C_V C^n \sum_{\alpha_1, \ldots, \alpha_{n+1} = 1}^M
		\sum_{j = 1}^{n+1}
		\sum_{k \in A_{\alpha_j}}
		\frac{\tilde\lambda_{\alpha_1,k}}{M}
		1
		\frac{\tilde\lambda_{\alpha_2,k}}{M}
		\frac{1}{\tilde\lambda_{\alpha_2,k}^{\frac 12} \tilde\lambda_{\alpha_3,k}^{\frac 12}}
		\ldots \\
	&\quad \times \ldots
		\frac{\tilde\lambda_{\alpha_j,k} |I(k)_{\alpha_j,\alpha_j} - 1|}{M}
		\frac{1}{\tilde\lambda_{\alpha_j,k}^{\frac 12} \tilde\lambda_{\alpha_{j+1},k}^{\frac 12}}
		\frac{\tilde\lambda_{\alpha_{j+1},k}}{M}
		\ldots
		\frac{\tilde\lambda_{\alpha_{n+1},k}}{M}
		\frac{1}{\tilde\lambda_{\alpha_{n+1},k}^{\frac 12} \tilde\lambda_{\alpha_1,k}^{\frac 12}}  \\
	&\le C_V C^n
		\sum_{\alpha = 1}^M \sum_{k \in A_\alpha} |I(k)_{\alpha,\alpha} - 1| M^{-1} \;.
\end{aligned}
\end{equation}
To optimally apply Lemma~\ref{lem:approxID}, we split the summation region $ A_\alpha $ into $ A_{\alpha,<} := \{ k \in A_\alpha : -2 \fl < k \cdot \hat\omega_\alpha < \fl \} $ and $ A_{\alpha,\ge} := \{ k \in A_\alpha : k \cdot \hat\omega_\alpha \ge \fl \} $, so that
\begin{equation}
	\sum_{k \in A_\alpha}
	|I(k)_{\alpha,\alpha} - 1|
	\le \sum_{k \in A_{\alpha,<}} C 
		+ \sum_{k \in  A_{\alpha,\ge}} C \fl |k \cdot \hat\omega_\alpha|^{-1}
	\le C \fl |\log \fl| L^3 \;.
\end{equation}
We conclude
\begin{equation}
	\sum_{k \in \Lambda_L} \Big\vert \tr \Big( I(k) P(k) \big( I(k) K(k) \big)^n \Big)
		- \tr \big( P(k) K(k)^n \big)
	\Big\vert
	\le C_V C^n \fl |\log \fl| L^3 \;,
\end{equation}
as claimed. The proof for $ h(k) $ instead of $ P(k) $ is analogous, where the bound in Lemma~\ref{lem:regularity_IhP} renders an additional factor of $ M \delta_{\alpha_1,\alpha_2} $ in the second line of~\eqref{eq:trip_expansion}, which eliminates one of the sums over the patches for $ n \ge 1 $.

Finally, the estimate~\eqref{eq:trace_estimates_X} follows by the same strategy, noting in particular that $ |\tr\big(IK(1-I)h(IK)^n \big)| = |\tr\big((1-I)h(IK)^{n+1} \big)| = |\tr\big(Ih(1-I)K(IK)^n \big)| $ is bounded as the contribution of $ j = 1 $ from~\eqref{eq:trip_expansion} with $ P $ replaced by $ h $ and $ n $ by $ n+1 $.
\end{proof}

\section*{Declarations}
The authors declare not to have any competing interests. GPT-6 Astra was used for proofreading. No AI tools were used in developing any of the mathematical arguments.

\section*{Acknowledgments}
NB was supported by the European Union through the ERC Starting Grant \textsc{FermiMath}, grant agreement nr.~101040991. SL was supported by the European Union partially through the ERC Starting Grant \textsc{FermiMath}, grant agreement nr.~101040991 and partially through the ERC Advanced Grant MathBEC, grant agreement nr.~101095820.

Views and opinions expressed are however those of the authors only and do not necessarily reflect those of the European Union or the European Research Council Executive Agency. Neither the European Union nor the granting authority can be held responsible for them.

This work has been carried out under the auspices of GNFM of INdAM.

\renewcommand{\bibfont}{\footnotesize}
\printbibliography

\end{document}